\documentclass[ba,preprint]{imsart}

\RequirePackage{amsthm,amsmath,amsfonts,amssymb, array, float}
\RequirePackage{natbib}
\usepackage{algorithm}
\usepackage{algpseudocode}
\usepackage{longtable}
\RequirePackage[colorlinks,citecolor=blue,urlcolor=blue,backref=page,backref=page]{hyperref}
\usepackage[capitalise]{cleveref}
\RequirePackage{graphicx}

\usepackage{booktabs}
\usepackage{multirow}
\usepackage{float}
\usepackage{array}
\usepackage{placeins}
\usepackage{longtable}
\usepackage{ragged2e}
\usepackage{mathrsfs}

\usepackage[normalem]{ulem}
\usepackage[dvipsnames]{xcolor}   
\usepackage[most]{tcolorbox} 
\definecolor{DarkGreen}{RGB}{0,100,40}

\newcommand{\del}[1]{{\textcolor{BrickRed}{\sout{#1}}}}

\renewcommand{\del}[1]{}
\usepackage{framed}   

\usepackage{framed}

\newtcolorbox{delblock}{
  enhanced, breakable,
  frame hidden,
  colback=white,
  coltext=BrickRed,                              
  borderline west={2.5pt}{0pt}{BrickRed},        
  left=8pt, top=2pt, bottom=2pt, right=0pt,
  before upper={\noindent\textsf{[Deleted from old draft:]}\par}
}

\newtcolorbox{addblock}{
  enhanced, breakable,
  frame hidden,
  colback=white,
  coltext=DarkGreen,
  borderline west={2.5pt}{0pt}{ForestGreen},
  left=8pt, top=2pt, bottom=2pt, right=0pt,
  before upper={\noindent\textsf{[New:]}\par}
}

\crefname{assumption}{Assumption}{Assumptions}
\Crefname{assumption}{Assumption}{Assumptions}

\pubyear{2026}
\volume{TBA}
\issue{TBA}
\firstpage{1}
\lastpage{1}

\startlocaldefs
\theoremstyle{plain}

\newtheorem{theorem}{Theorem}[section]
\newtheorem{lemma}[theorem]{Lemma}
\newtheorem{corollary}[theorem]{Corollary}
\theoremstyle{definition}
\newtheorem{definition}[theorem]{Definition}
\newtheorem{assumption}{Assumption}

\newtheorem{proposition}[theorem]{Proposition}
\theoremstyle{remark}

\endlocaldefs

\begin{document}

\begin{frontmatter}
\title{Calibrated Bayesian Inference for Stochastic Intervention Effects}
\runtitle{Calibrated Bayesian Inference for Stochastic Intervention Effects}

\begin{aug}
\author[A]{\fnms{Tyler M.}~\snm{Schmidt}\ead[label=e1]{tyler-schmidt-1@uiowa.edu}},
\author[A]{\fnms{Nathan B.}~\snm{Wikle}\ead[label=e2]{nathan-wikle@uiowa.edu}} 

\address[A]{Department of Statistics and Actuarial Science,
University of Iowa\printead[presep={,\ }]{e1}\printead[presep={;\ }]{e2}}

\runauthor{T. M. Schmidt and N. B. Wikle}
\end{aug}

\begin{abstract}

Causal inference increasingly extends beyond classical causal effects defined by deterministic treatment assignments, such as the average treatment effect, to stochastic intervention effects that can weaken positivity requirements and offer greater policy relevance. Nonparametric Bayesian models are attractive for estimating these effects due to their flexibility and inherent uncertainty propagation, but this posterior uncertainty need not be well calibrated for the causal effect of interest. We develop a simple post-processing correction that can be applied to posterior samples without changing the prior or fitting algorithm. We prove that, for a broad class of smooth stochastic interventions, the corrected posterior yields asymptotically efficient inference and credible intervals with asymptotically valid frequentist coverage; formally, it satisfies a semiparametric Bernstein-von Mises theorem. The theory covers interventions specified independently of the observed treatment process, as well as interventions that modify it, including incremental propensity score interventions and a new power-tilt intervention. A central contribution is new theory for SoftBART, including conditions under which this flexible tree-based Bayesian model supports calibrated Bayesian inference for stochastic intervention effects. In simulations, the correction reduces bias and improves coverage relative to the uncorrected Bayesian analysis while remaining competitive with frequentist alternatives. We illustrate the method by estimating how expected LDL cholesterol would change under hypothetical increases or decreases in the odds of receiving statin therapy.
\end{abstract}


\end{frontmatter}

\section{Introduction} \label{sec:intro}

Many scientifically relevant interventions do not assign treatments deterministically. For example, new physician guidelines or a public health campaign may substantially alter the odds of treatment within a population, while the treatment received by any particular individual remains somewhat random, reflecting clinical judgment, patient preference, and the variable nature of human decision making. Similarly, regulations intended to limit point-source air pollution may induce well-understood changes in the distribution of pollution exposure, but the exact exposure level at a given location is the result of a complex, physical-chemical stochastic process. In such settings, causal estimands based on fixed treatment assignments --- such as contrasts between outcomes if all individuals were assigned treatment versus control, or if all individuals followed a prespecified longitudinal treatment regime --- may not align with the scientific or regulatory objectives of interest. In response to these limitations, a growing number of causal estimands are defined using \textit{stochastic interventions}, which quantify how outcomes might change under modified treatment distributions. These estimands can weaken positivity requirements, support more realistic policy questions, and produce interpretable causal effect curves, but they also raise new questions about how to conduct Bayesian inference with desirable asymptotic guarantees.

In this paper, we consider stochastic interventions in this broad sense. Rather than focusing on a single deterministic contrast, we study causal estimands that replace the observed treatment mechanism with a hypothetical intervention distribution, which may be fixed externally or may depend on features of the observed data-generating process. The resulting estimands are often a population mean outcome, or a collection of population mean outcomes, under a family of modified treatment distributions. The basic premise is simple and has been recognized as a way to accommodate the inherent  uncertainty in policy interventions since at least \citet{stock1989nonparametric}. In recent years, such estimands have become increasingly common in the causal inference literature. Early work by \citet{Robins2004, didelez2006direct}, and \citet{Tian2008} described probabilistic interventions that assign treatment according to a user-specified distribution, and extensions have been used in the analysis of static \citep{Diaz2012, Diaz2013, haneuse2013modified} and dynamic longitudinal treatment regimes \citep{Young2011, Young2014, Kennedy2018ipsi}; in mediation analysis \citep{didelez2006direct, diaz2020causal, hejazi2023mediation}; and in settings with network \citep{Ogburn2024} and spatial interference \citep{Papadogeorgou2019, Papadogeorgou2022}, among others.

The increased use of stochastic interventions can be partly attributed to their generality. Estimands have been proposed for a wide range of exposure types, ranging from the routine --- e.g., (time-varying) binary treatments \citep{Robins2004, Young2011, Young2014, Kennedy2018ipsi, Wu2024} and continuous exposures \citep{Diaz2012, Diaz2013, haneuse2013modified, sani2020shift, schindl2026incremental} --- to the non-standard, such as multivariate environmental mixtures \citep{Huang2026} and spatio-temporal point process exposures \citep{Papadogeorgou2022}. This generality is more than notational: stochastic estimands often offer analytic and inferential advantages over their deterministic counterparts. Because the hypothetical exposure distribution can be constrained to remain close to, or absolutely continuous with respect to, the observed treatment mechanism, stochastic intervention effects may require weaker support conditions than interventions that set treatment to a fixed value. Incremental propensity score interventions, for example, reweight each unit's odds of treatment rather than imposing a common treatment rule, and are identifiable without requiring propensity scores to be bounded away from zero and one \citep{Kennedy2018ipsi}. Analogously, for continuous treatments, so-called shift, modified treatment policy, and exponential tilt estimands target feasible perturbations of the exposure distribution rather than a dose-response contrast evaluated at fixed exposure values, helping to reduce reliance on extrapolation outside the observed range of covariate data \citep{Diaz2012, haneuse2013modified, sani2020shift, schindl2026incremental}. This generality has led to the use of stochastic interventions in a variety of applied domains, including economics, political science, public health, environmental epidemiology, infectious disease, and criminology.

It is in the context of this rising popularity that we consider Bayesian estimation of stochastic intervention effects. Bayesian methods for causal inference have received increased attention in settings with deterministic interventions, due in part to the flexibility of modern Bayesian nonparametric priors \citep{hill2011bnp, hahn2020bcf, roy2018bnpmar, alaa2017itegp, ray2019debiased}, the inherent uncertainty quantification of Bayesian methods (particularly in finite samples), and their favorable empirical performance when compared against common non-Bayesian alternatives \citep{dorie2019competition, Yiu2025pcorrect}; see \citet{daniels2023bnp}, \citet{oganisian2021bayesian}, and \citet{li2023bayesian} for general reviews on Bayesian causal inference. The use of these methods in the stochastic intervention setting has a natural appeal --- estimands based on stochastic interventions are meant to account for uncertainty in the treatment assignment process, and in the Bayesian framework this coincides nicely with posterior predictive uncertainty quantification. 

Despite this natural compatibility, the use of Bayesian inference for stochastic intervention effects is rare. Most existing inferential work is frequentist, relying on popular estimating equation and semiparametric-based methods such as inverse probability weighting, targeted learning, or double machine learning estimators \citep{Kennedy2018ipsi, Diaz2012, Diaz2013}. Notable exceptions include \citet{Comment2022}, who extend the Bayesian parametric g-formula \citep{Keil2018} to the estimation of dynamic and stochastic longitudinal treatment regimes. Similarly, \citet{Chen2025} introduce a general Bayesian g-formula estimator that utilizes Bayesian additive regression tree (BART) priors to estimate deterministic and stochastic longitudinal treatment effects. \citet{Tang2025} use stochastic interventions to evaluate the safety of alternative radiation dosages received by normal tissue in cancer treatment, and propose a Bayesian monotone marginal structural model for inference. Finally, \citet{Josefsson2026} introduce an incremental threshold intervention to investigate the effect of reductions in systolic blood pressure on memory, which they estimate using the Bayesian semiparametric g-formula equipped with modified SoftBART priors; the intervention is framed as a deterministic downward shift in blood pressure by an unknown quantity, however, the assignment of a prior to this quantity induces a corresponding distributional change to the conditional intervention density. 

These examples establish the feasibility of Bayesian stochastic intervention analyses, but they do not resolve a basic calibration question, in the sense of \citet{little2011calibrated}: when do posterior summaries for stochastic intervention effects yield inferences with desirable frequentist guarantees? This question is nontrivial because stochastic intervention effects are semiparametric functionals of the observed data distribution, and many popular stochastic interventions, such as incremental propensity score interventions and shift interventions, may themselves depend on the unknown treatment mechanism. As a result, regularization from flexible Bayesian priors can distort the induced marginal posterior for the causal effect \citep{linero2023bnpcause, hahn2020bcf}, and the resulting plug-in credible intervals may not have nominal frequentist coverage \citep{rousseau2016frequentist, ray2020semiparametric, Yiu2025pcorrect}. A formal treatment is therefore needed to connect Bayesian posterior uncertainty with semiparametric efficiency theory for stochastic intervention effects.

Motivated by this gap, we develop a calibrated Bayesian procedure for stochastic intervention effects. Building on the semiparametric posterior correction introduced by \citet{Yiu2025pcorrect}, we prove that the resulting one-step corrected posterior satisfies a Bernstein--von Mises (BvM) theorem for a broad class of smooth stochastic intervention effects. In particular, the result applies to popular 
time-varying stochastic interventions for binary and categorical treatments, including incremental propensity score interventions. 
In the process, we provide a set of easier-to-verify sufficient conditions to avoid a Donsker condition in the BvM theorem of \citet{Yiu2025pcorrect}. Furthermore, we show that the conditions of our BvM theorem can be verified when the unknown outcome regression and treatment assignment mechanism are modeled using flexible SoftBART priors \citep{Linero2018brte}. Specifically, we derive a general sup-norm covering bound for the SoftBART function class and combine it with random design posterior contraction results on the outcome regression and propensity score to establish the empirical process and convergence requirements needed for calibrated posterior inference. 
Together, these results give conditions under which posterior credible regions for stochastic intervention effects are asymptotically equivalent to efficient frequentist confidence regions, while retaining the finite-sample uncertainty propagation and modeling flexibility of Bayesian nonparametric methods. We compare the favorable finite-sample performance of the proposed estimator against frequentist and Bayesian competitors via a simulation study. Finally, we illustrate its use in an analysis of the effect of statin therapy on lowering LDL cholesterol levels among adults in the United States.

The remainder of the article is organized as follows. \cref{sec:stochastic_int} formally defines stochastic intervention effects 
and introduces two Bayesian estimators: a plug-in estimator and the proposed one-step posterior correction. \cref{sec:Correction} establishes the main semiparametric Bernstein--von Mises result for stochastic intervention effects estimated with the one-step posterior correction. \cref{sec:SoftBART} verifies the entropy and contraction conditions needed for this result when the outcome regression and treatment assignment mechanism are modeled using SoftBART priors. \cref{sec:simulation_study} evaluates the finite-sample performance of the proposed correction in simulations, and \cref{sec:Real_data} applies the method to the statin therapy analysis. Finally, \cref{sec:Discussion} concludes with a discussion of limitations and promising extensions of our work.

\section{Stochastic Interventions} \label{sec:stochastic_int}

\subsection{Notation and Preliminaries}

We consider the case where we observe independent and identically distributed data, $(Z_1, \ldots, Z_n)$, from an unknown distribution $P_0 \in \mathcal{P}$, where $Z = (X, A, Y)$. Here, $X$ denotes a vector of covariates measured prior to treatment $A$, and $Y$ is the outcome of interest. Throughout, we assume binary treatments (so the support of $A$ is $\mathcal{A} = \{0,1\}$), although our results can be easily generalized to any set of categorical treatments. For ease of exposition, in the main text we restrict our focus to a single-time-point intervention, however, all results and corresponding proofs are extended to the longitudinal setting in the Supplementary Material. 

 Let $\mathcal{P}$ denote a collection of probability measures on a Polish sample space $(\mathcal{Z}, \mathcal{A}_{\mathcal{Z}})$. We denote $P(B)$ as the probability of some arbitrary event $B$, while $P[f] = \int f(z) dP(z)$ represents the expectation of $f$ under $P \in \mathcal{P}$. Additionally, $\mathbb{P}_n[f] = n^{-1}\sum_{i=1}^n f(Z_i)$, where $\mathbb{P}_n$ is the empirical measure, and $\mathbb{G}_n = \sqrt{n}(\mathbb{P}_n - P_0)$ denotes the empirical process. For $P \in \mathcal{P}$, let $L_2(P)$ be the Hilbert space consisting of all measurable functions $h: \mathcal{Z} \rightarrow \mathbb{R}$ with $P[h^2] < \infty$ equipped with the inner product $\langle h_1, h_2 \rangle _P = P[h_1h_2]$ and norm $\| h \|_{L_2(P)} = \sqrt{P[h^2]}$. We adopt a Bayesian framework and equip $\mathcal{P}$ with a $\sigma$-field $\mathcal{B}$ such that $(\mathcal{P}, \mathcal{B})$ is a standard Borel space. We further specify a prior probability distribution as $\Pi$ on $\mathcal{P}$ and let $\Pi(\cdot \mid Z_1, \ldots, Z_n)$ denote a version of the posterior distribution given $(Z_1, \ldots, Z_n)$. 

We adopt the potential outcomes framework for causal inference \citep{rubin1974causal}, letting $Y^{a}$ denote the outcome that would have occurred had the treatment $A = a$ been received. This represents a counterfactual outcome under a deterministic intervention, where treatment is assigned to level $a$ with probability one. The canonical example of an estimand based on deterministic interventions is the average treatment effect, which compares the average difference in potential outcomes when all units are assigned to treatment versus control.

\subsection{Stochastic Intervention Effects: Definitions and Identifiability} \label{subsec:stoch-int-defintion}

While deterministic interventions are conceptually appealing, they can be unrealistic or ill-defined in many applied settings. In practice, treatments are often not assigned deterministically, but instead arise through complex mechanisms involving eligibility criteria, institutional discretion, risk stratification, or resource constraints. Moreover, deterministic interventions may require strong positivity assumptions, particularly when some individuals have near-zero or near-one probability of treatment. These limitations motivate the study of stochastic interventions, which define counterfactual treatment regimes through modified treatment distributions rather than fixed treatment assignments. Rather than considering counterfactuals that assign treatment uniformly for all individuals, stochastic interventions alter the probability of treatment in ways that remain anchored to the observed data-generating process, often yielding more realistic and statistically stable causal estimands.

More generally, a \textit{stochastic intervention} replaces the observed treatment assignment mechanism with an alternative (possibly conditional) treatment distribution, which we denote as $Q(a \mid x)$. Under such an intervention, treatment is no longer deterministically assigned, but instead sampled according to Q, yielding the counterfactual \textit{random} outcome, $Y^Q$. We define the corresponding population mean outcome under intervention $Q$ as 
\begin{equation} \label{eq:stoch-int-def}
    \psi_{P_0}(Q) \equiv E[Y^Q] = \int_{\mathcal{X}} \int_{\mathcal{A}} E[Y^a \mid X = x] dQ(a \mid x) dP_0(x).
\end{equation}
Deterministic interventions arise as a special case of \eqref{eq:stoch-int-def} when $Q(a \mid x)$ assigns a fixed treatment with probability one. Consequently, $\psi_{P_0}(Q)$ can be viewed as an extension of the commonly used average potential outcome, $E[Y^a]$, to a setting that respects possible stochasticity in the treatment assignment each unit receives. 

In practice, we likely wish to estimate $\psi_{P_0}( Q)$ across a sequence or family of distributions, $\{Q_{\delta}\}_{\delta \in \mathcal{D}}$, which we index by $\delta$. We separate $Q$ into two broad classes of interventions. First, interventions may be defined independently of the observed-data distribution, $P_0$, which we refer to as \textit{$\pi$-fixed interventions}. For example, if there exists \textit{a priori} knowledge of the impact of a policy or regulatory action on the likelihood an individual receives treatment, $Q$ may be defined accordingly. Alternatively, $Q$ may be defined as a function of the observed-data distribution, $P_0$. For example, the intervention might be a transformation of the propensity score, $\pi(x) = P_0(A = 1 \mid X = x)$. We refer to these as \textit{$\pi$-dependent interventions}, and present two illustrative examples.


\noindent\textit{Example One} (Incremental propensity score intervention). The IPSI \citep{Kennedy2018ipsi} rescales the odds of treatment by a factor $\delta$, such that
\begin{equation}\label{eq:IPSI_density}
dQ_{\delta}(A=1\mid X=x)=\frac{\delta\pi(x)}{\delta\pi(x)+1-\pi(x)},\quad 0<\delta<\infty,
\end{equation}
or equivalently, $\frac{dQ_{\delta}(A=1\mid x)}{1-dQ_{\delta}(A=1\mid x)}=\delta\,\frac{\pi(x)}{1-\pi(x)}$. For binary treatment this is an exponential tilt of the treatment mechanism \citep{diaz2020causal}. Because it multiplies the odds rather than imposing absolute treatment probabilities, it avoids the positivity violations that arise when baseline probabilities are near zero \citep{Kennedy2018ipsi}.

\noindent\textit{Example Two} (Power tilt intervention). The PTI instead raises the odds of treatment to a power $\delta$, such that
\begin{equation}\label{eq:PT_density}
dQ_{\delta}(A=1\mid X=x)=\frac{\pi(x)^\delta}{\pi(x)^\delta+\{1-\pi(x)\}^\delta},\quad 0<\delta<\infty,
\end{equation}
or equivalently, $\frac{dQ_{\delta}(A=1\mid x)}{1-dQ_{\delta}(A=1\mid x)}=\big(\frac{\pi(x)}{1-\pi(x)}\big)^\delta$. Note that $\delta>1$ concentrates treatment among high-propensity units, while $\delta<1$ shrinks the propensity toward $1/2$, yielding a more uniform assignment. Whereas the IPSI shifts the log-odds additively, the power tilt rescales it which, to our knowledge, has not previously been used as a causal intervention. 
A visual comparison of these two $\pi$-dependent interventions is shown in \cref{fig:interventions-plot}.

\begin{figure}
  \centering
  \includegraphics[width=0.75\textwidth]{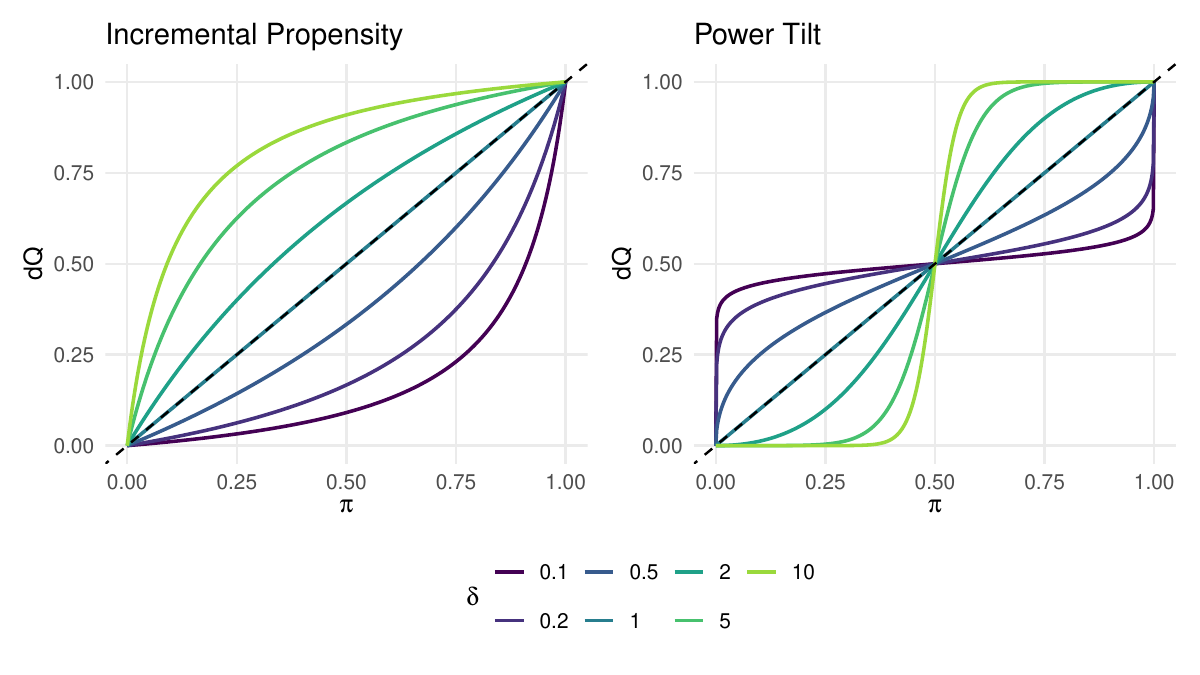}
  \caption{Changes to $dQ_{\delta}(A = 1 \mid \pi)$ as $\pi$ ranges from 0 to 1. On the left, $Q_{\delta}$ is defined as an incremental propensity score intervention, where $\log\delta$ is an additive shift to the log-odds of treatment; on the right, $Q_{\delta}$ is defined as a power tilt intervention, where $\delta$ rescales the log-odds of treatment.}
  \label{fig:interventions-plot}
\end{figure}

Although the stochastic intervention effect in \eqref{eq:stoch-int-def} is a well-defined \textit{counterfactual} quantity, additional identifying assumptions are required to link it to the observed-data distribution, $P_0$. 
\begin{assumption}  \label{ass:identifying}
\text{ }
\begin{enumerate}
    \item[(a)] (Consistency) $Y = Y^{a}$ if $A = a$.
    \item[(b)] (Exchangeability) $A \perp\!\!\!\perp Y^{a} \mid X$.
    \item[(c)] (Weak Positivity) $dP_0(A \mid X) = 0 \implies dQ (A\mid X) = 0$.
\end{enumerate}
\end{assumption}

Consistency requires that the observed outcome equals the counterfactual outcome corresponding to the observed treatment, ensuring that interventions are well defined and implicitly eliminating interference between units. Exchangeability assumes that treatment assignment is independent of the potential outcomes conditional on the set of observed confounders, ensuring that there is no unmeasured confounding. Finally, 
weak positivity states that conditional on $X$, the intervention $Q$ assigns positive probability only to treatment values that are possible under the observed treatment mechanism $P_0$.
If these assumptions are satisfied, the following identification result holds:  

\begin{definition}  \label{def:identification}
     Under \cref{ass:identifying}, the estimand for a stochastic intervention, $\psi_{P_0}(Q)$, is identified by 
    \begin{equation} \label{eq:stoch-int-P0}
        \psi_{P_0}(Q) =  \int_{\mathcal{X}}\int_{\mathcal{A}} \mu(x,a) dQ(a \mid x)dP_0(x),
    \end{equation}
    where $\mu(x,a) =  E\left[Y \mid X = x, A = a\right]$.
\end{definition}

A proof of identification is given in \citet{Kennedy2018ipsi}; the result follows from a straightforward application of the g-formula \citep{robins1986gformula}. 
Having established identification, the remaining challenge is estimation of the intervention functional $\psi_{P_0}(Q)$. A natural approach is Bayesian inference, which provides a coherent framework for uncertainty quantification by placing a likelihood and priors on the observed data-generating law and propagating posterior uncertainty to the causal estimand. We therefore focus on Bayesian estimation of stochastic intervention effects.

\subsection{Bayesian Estimation of Stochastic Intervention Effects}


At first glance, Bayesian estimation of $\psi_{P_0}(Q)$ appears straightforward. Under \cref{ass:identifying}, the causal estimand is identified as a low-dimensional functional of the observed-data law $P_0$. Consequently, a natural Bayesian procedure is to place a prior on the components of $P_0$, obtain posterior draws from $\Pi(\cdot \mid Z_{1:n})$, and push these draws through the identifying functional, $P \mapsto \psi_P := \psi_P(Q_P)$, where $Q_P$ denotes the intervention distribution associated with $P$. 

In practice, posterior draws $\{P^{(1)}, \ldots, P^{(B)} \}$ are often obtained by specifying a likelihood based on the observed-data factorization and independent priors for the corresponding nuisance components of the estimand \citep{linero2024}. For a $\pi$-fixed stochastic intervention (i.e., $Q_P \equiv Q$ for all $P \in \mathcal{P}$), evaluating the identifying functional requires posterior draws of the outcome regression $\mu$ and the covariate distribution $P_X$. For a $\pi$-dependent intervention, it also requires posterior draws of the propensity score $\pi$. Evaluating $\psi_{P^{(b)}} (Q_{P^{(b)}})$ for $b = 1, \ldots, B$ yields a marginal posterior sample for the stochastic intervention effect; we refer to this distribution as the \textit{plug-in posterior}. This procedure is summarized in \cref{alg:plug-in}, which assumes a flat Dirichlet prior (i.e., Bayesian bootstrap (BB)) for $P_X$ \citep{rubin1981bayesianbootstrap}.

\begin{algorithm} 
\caption{Plug-In Posterior Sampling Scheme}
\label{alg:plug-in}
\begin{algorithmic}[]
\State \textbf{Input:} number of posterior samples $B$, stochastic intervention $Q_P$;
\State Sample $\{\mu^{(1)}, \ldots, \mu^{(B)}\}$ from the initial posterior $\Pi(\mu \mid Z_{1:n})$;
\State If $Q_P$ is $\pi$-dependent, sample $\{\pi^{(1)}, \ldots, \pi^{(B)}\}$ and compute $\{dQ^{(1)}, \ldots, dQ^{(B)}\}$;
\For{$b \gets 1$ to $B$}
  \State Sample 
  $(W^{(b)}_1, \ldots, W^{(b)}_n) \sim \text{Dir}(n; 1, \ldots, 1)$, the BB posterior for $P_X$;
  \State Compute
  \begin{equation*}
    \psi^{(b)}_{\text{plug-in}} = \sum_{i=1}^n W^{(b)}_i \,  \int_{\mathcal{A}} \mu^{(b)}(X_i, a) \, dQ^{(b)}(a \mid X_i);
  \end{equation*}
\EndFor
\State \textbf{Return:} $\big\{\psi^{(b)}_{\text{plug-in}} : b = 1,\dots,B \}$.
\end{algorithmic}
\end{algorithm}
\FloatBarrier

Bayesian nonparametric methods are especially appealing in this setting because they allow the nuisance components of $P_0$ to be modeled flexibly while propagating uncertainty through the causal functional. Common choices in causal applications include Bayesian Additive Regression Trees (BART) and related variants such as Bayesian causal forests (BCF) and SoftBART, Gaussian process priors, and Dirichlet process priors \citep{hill2011bnp, dorie2019competition, hahn2020bcf, linero2023bnpcause, Yiu2025pcorrect}. These priors are often used to model function-valued nuisance components rather than the causal estimand directly. Thus, even when $\psi_{P_0}(Q)$ is scalar, or a finite-dimensional curve indexed by $\delta$, its estimation depends on potentially infinite-dimensional quantities such as $\mu$, $P_X$, and, for $\pi$-dependent interventions, $\pi$.

Thus, estimation of $\psi_{P_0}(Q)$ confronts a central tension in semiparametric Bayesian causal inference: the target parameter is low-dimensional, but the posterior is placed on a much larger object, namely the observed-data law or its nuisance components. Unlike with parametric regular models, where the influence of the prior disappears asymptotically to first-order under the celebrated Bernstein-von Mises (BvM) theorem \citep{vandervaart1998asymptotic}, for infinite dimensional models the impact of the prior can remain asymptotically non-negligible \citep{diaconis1986consistency, rousseau2016frequentist}. Furthermore, even if certain nonparametric models have desirable adaptive posterior contraction rates \citep{ghosal2017fundamentals}, these largely apply to the estimation of the nuisance functions themselves and do not automatically extend to the lower-dimensional plug-in marginal posterior distribution, even to first-order \citep{Yiu2025pcorrect, rousseau2016frequentist}. This is closely related to the concept of regularization induced confounding (RIC) in the causal inference literature \citep{hahn2018regularization}, in which the regularization used to control the complexity of nuisance parameters can indirectly regularize the causal estimand of interest, sometimes with an undesirable effect on the frequentist properties of the associated marginal posterior \citep{linero2024}. Relating this to our stochastic intervention estimand, the choice of prior for $\mu$ and $\pi$ may indirectly affect whether the posterior for $\psi_P(Q)$ is centered correctly, has the correct asymptotic spread, or yields credible intervals with nominal frequentist coverage \citep{castillo2015bvm, rousseau2016frequentist, ray2020semiparametric, Yiu2025pcorrect}.

These considerations motivate Bayesian procedures that explicitly target $\psi_{P}(Q)$, rather than relying solely on the uncorrected plug-in posterior. One strategy is to tailor the prior specification so that the marginal posterior for a particular functional satisfies a semiparametric BvM theorem, for example through undersmoothing or augmentation in a ``least favorable'' direction \citep{castillo2015bvm, ray2020semiparametric}. In causal inference, related approaches include prior specifications designed to mitigate regularization induced confounding, such as BCF-style decompositions that separate prognostic and treatment effect components \citep{hahn2020bcf, linero2024}. While these approaches have merit, they are often tied to a particular estimand, prior, or model implementation. Instead, we adopt the semiparametric posterior correction framework of \citet{Yiu2025pcorrect}. Starting from the initial posterior $\Pi(\cdot \mid Z_{1:n})$, this approach leaves the prior specification unchanged and instead applies an efficient influence function-based \citep{vandervaart1998asymptotic} stochastic correction to the marginal plug-in posterior. This serves as a Bayesian analog to the classical one-step estimator \citep{pfanzagl1982, newey1998}, as it corrects each plug-in draw in the direction of a first-order semiparametric expansion of $\psi_P(Q)$ while retaining posterior uncertainty in $\mu$, $\pi$, and $P_X$ with desirable frequentist properties.


\subsection{A One-Step Posterior Correction for Stochastic Interventions}

We now present the one-step posterior correction for the stochastic intervention setting. 
We assume that $\psi$ is sufficiently regular to admit an \textit{efficient influence function} (EIF), $\phi_P$, for all $P \in \mathcal{P}$ \citep{vandervaart1998asymptotic}. Informally, $\phi_P$ acts as a mean-zero gradient that describes the local sensitivity of $\psi$ to perturbations of $P$. Among all influence functions for $\psi$, the EIF has the smallest variance, attaining the semiparametric efficiency bound \citep{tsiatis2006semiparametric}. Most importantly for our purposes, it yields a first-order von Mises expansion,\footnote{This is a functional analog to the first-order Taylor series expansion of a function.}
\begin{equation} \label{eq:eif-von-mises}
    \psi_{P_0} - \psi_{P} = P_0[\phi_P] + R_2(P_0, P),
\end{equation}
where $R_2(P_0, P)$ is a higher-order remainder and we have used the fact that $P[\phi_P] = 0$ by construction. Thus, $P_0[\phi_P]$ captures the leading bias of $\psi_{P}(Q)$. In the non-Bayesian setting, this motivates the \textit{one-step estimator},
\begin{equation} \label{eq:freq-one-step}
    \hat{\psi}_{\text{1-step}} = \psi_{\hat{P}} + \mathbb{P}_n[\phi_{\hat{P}}],
\end{equation}
wherein an initial estimate $\hat{P}$ of the data distribution is plugged in to the functional $\psi_{\hat{P}}$ and then ``corrected'' with an empirical estimate of the first-order term in \eqref{eq:eif-von-mises}. Notably, under sufficient conditions on the form of the estimand and the choice of nonparametric estimators of $\hat{P}$, the one-step estimator can be shown to be asymptotically efficient, with limiting distribution $\mathcal{N}(0, \|\phi_{P_0} \|_{L_2(P_0)}^2 )$. The one-step estimator and related EIF methods are a core component of many modern causal inference methods, as they often facilitate the use of flexible machine learning algorithms to estimate $\hat{P}$ while retaining $\sqrt{n}$-consistent estimation of the causal estimand; see \citet{kennedy2024semiparametric} for a review of this theory.

The \textit{one-step posterior} can be viewed as a Bayesian analog to \eqref{eq:freq-one-step}. Once again, denote the stochastic intervention effect as $\psi_P(Q_P)$, where $Q_P$ may be fixed independently of $P$ (i.e., $\pi$-fixed) or may depend on $P$ ($\pi$-dependent). Similarly, let $\phi_{P}(Z_i) := \phi_P(Z_i; Q_P)$ represent the efficient influence function of $\psi_P(Q_P)$, evaluated at $Z_i$. Given a draw $P^{(b)}$ from the initial posterior and an independent set of Bayesian bootstrap weights, 
\begin{equation*}
    (\widetilde{W}_{1}^{(b)}, \dots, \widetilde{W}_n^{(b)}) \sim \text{Dirichlet}(1, \ldots, 1),
\end{equation*}
the corresponding one-step posterior draw is
\begin{equation}
    \psi_{\text{1-step}}^{(b)} = \psi_{\text{plug-in}}^{(b)} + \sum_{i = 1}^{n} \widetilde{W}_i^{(b)} \phi_{P^{(b)}}(Z_i),
\end{equation}
where $\psi_{\text{plug-in}}^{(b)}$ is the plug-in estimate from the $b$'th posterior draw (\cref{alg:plug-in}). In other words, each plug-in posterior draw is augmented by a \textit{randomized} efficient influence function-based correction. Importantly, this randomization is accomplished via a draw from a Bayesian bootstrap posterior, $\Pi_{BB}(\cdot \mid Z_{1:n})$, that is conditionally independent from $\Pi(\cdot \mid Z_{1:n})$. The resulting collection, $\{\psi_{\text{1-step}}^{(1)}, \ldots, \psi_{\text{1-step}}^{(B)}\}$, forms a sample from the one-step posterior for $\psi_P(Q_P)$; the procedure is summarized in \cref{alg:onestep}.

\begin{algorithm}[th]
\caption{One-Step Posterior Sampling Scheme}
\label{alg:onestep}
\begin{algorithmic}
\State \textbf{Input:} number of posterior samples $B$, stochastic intervention $Q_P$;
\State Sample $\{P^{(1)}, \ldots, P^{(B)}\}$ from the initial posterior $\Pi(\cdot \mid Z_{1:n})$;
\For{$b \gets 1$ to $B$}
  \State Sample $(\widetilde{W}^{(b)}_1, \ldots, \widetilde{W}^{(b)}_n) \sim \text{Dir}(n; 1, \ldots, 1)$ independently of $P^{(b)}$;
  \State Compute
     \begin{equation*}
     \psi^{(b)}_{\text{1-step}} =  \psi_{\text{plug-in}}^{(b)} + \sum_{i=1}^n \widetilde{W}^{(b)}_i \, \phi_{P^{(b)}}(Z_i),
    \end{equation*}
  \State where $\psi_{\text{plug-in}}^{(b)}$ is computed as in \cref{alg:plug-in};
  \State $\phi_{P^{(b)}}(Z_i)$ is the efficient influence function evaluated at $Z_i$;
\EndFor
\State \textbf{Return:} $\big\{\psi^{(b)}_{\text{1-step}} : b = 1,\dots,B \big\}$.
\end{algorithmic}
\end{algorithm}

The advantage of the one-step posterior correction is two-fold. First, it operates as a post-processing step on the original posterior distribution $\Pi(\cdot \mid Z_{1:n})$, meaning that the initial prior does not need to be tailored to the specific estimand of interest. In fact, the posterior samples from $\Pi(\cdot \mid Z_{1:n})$ can be reused to target different estimands, such as those corresponding to a sequence of stochastic interventions, $\{Q_{\delta}\}_{\delta \in \mathcal{D}}$. Second, \citet{Yiu2025pcorrect} outline general conditions under which the one-step posterior satisfies the semiparametric Bernstein-von Mises (BvM) theorem, which we define as follows.
\begin{definition} \label[definition]{def:BVM-def}
    Let $\mathcal{L}_{\Pi \times \Pi_{BB}}(\sqrt{n}(\psi_{\text{1-step}} - \hat{\psi}_n) \mid Z_{1:n})$ denote the posterior law of $\sqrt{n}(\psi_{\text{1-step}} - \hat{\psi}_n)$, where $\hat{\psi}_n$ is the asymptotically efficient sequence $\hat{\psi}_n = \psi_{P_0} + \mathbb{P}_n[\phi_{P_0}]$.  The one-step posterior satisfies the \textit{semiparametric Bernstein-von Mises} theorem if
    \begin{equation*}
         d_{BL} \left(\mathcal{L}_{\Pi \times \Pi_{\text{BB}}}
         (\sqrt{n}(\psi_{\text{1-step}} - \hat{\psi}_n) \mid Z_{1:n} ),\;
         \mathcal{N}(0, \|\phi_{P_0} \|_{L_2(P_0)}^2 )
    \right)
    \xrightarrow{P_0} 0,
    \end{equation*}
    where $d_{BL}$ denotes the bounded Lipschitz distance and $\Pi \times \Pi_{BB}$ denotes the product posterior of $\Pi(\cdot \mid Z_{1:n})$ and the Bayesian bootstrap posterior $\Pi_{BB}(\cdot \mid Z_{1:n})$.
\end{definition}
The semiparametric BvM is analogous to the parametric one in that the posterior law of $\sqrt{n}(\psi_{\text{1-step}} - \hat{\psi}_n)$ converges weakly to a Gaussian distribution, except that, in place of the inverse Fisher information, the limiting variance is equal to the squared $L_2(P_0)$ norm of the efficient influence function of $\psi$ evaluated at the true distribution $P_0$. Consequently, the one-step corrected credible intervals of $\psi$ are asymptotically equivalent to efficient frequentist confidence intervals \citep{rousseau2016frequentist}, so long as the following assumption is satisfied \citep{Yiu2025pcorrect}, where $\phi^*_P := \phi_P + \psi_P$ denotes the uncentered EIF. \footnote{In \citet{Yiu2025pcorrect}, they state (b) and (c) in terms of the centered EIF $\phi_P$, and later note on p.11 that $\phi^*_P$ is sufficient.}

\begin{assumption} 
\label{ass:yiu-main}
There exists a sequence of measurable subsets $(\mathcal{S}_n)_n \subset \mathcal{P}$ satisfying
\begin{equation*}
\Pi(P \in \mathcal{S}_n \mid Z_{1:n}) \xrightarrow{P_0} 1,
\end{equation*}
with the following properties:
\begin{enumerate}
    \item[(a)] \textit{No second-order bias}.
    \begin{equation*}
\sup_{P \in \mathcal{S}_n}
\Big|
\sqrt{n} \, R_2(P_0, P)
\Big|
=
\sup_{P \in \mathcal{S}_n}
\Big|
\sqrt{n} \, \big\{ \psi_{P_0} - \psi_{P} - P_0[\phi_P] \big\}
\Big|
\to 0.
\end{equation*}
    \item[(b)] \textit{$L_2$-convergence}.
    \begin{equation*}
  \sup_{P \in \mathcal{S}_n}
\| \phi^*_P - \phi^*_{P_0} \|_{L_2(P_0)} \to 0.
\end{equation*}

    \item[(c)] \textit{Donsker class or empirical process conditions}.
    
    Either the sets $\{\phi^*_P : P \in \mathcal{S}_n\}$ are eventually contained in a fixed
$P_0$-Donsker class \citep{vandervaartwellner2023weak},  
or both of the following hold:
    \begin{enumerate}
        \item[(i)] (Convergence of $\phi^*_P$ under the empirical process).
        \begin{equation*}
\sup_{P \in \mathcal{S}_n}
\big|
\mathbb{G}_n[\phi^*_P - \phi^*_{P_0}]
\big|
\xrightarrow{P_0} 0.
\end{equation*}
This statement is interpreted in terms of outer probability \citep{vandervaartwellner2023weak} if the supremum is not measurable.
        \item[(ii)] (Bounding of envelope functions).
        The sets $\{\phi^*_P: P \in \mathcal{S}_n\}$ have envelope functions $G_n$ satisfying
\begin{equation*}
\lim_{C \to \infty} \limsup_{n \to \infty}
P_0\!\left(G_n^2 \mathbf{1}\{G_n^2 > C\}\right) = 0,
\quad
P_0 G_n^4 = o(n).
\end{equation*}
    \end{enumerate}
\end{enumerate}
\end{assumption}

Whether \cref{ass:yiu-main} is satisfied for a general stochastic intervention effect requires careful consideration and is the focus of \cref{sec:Correction}. In practice, even when the assumption is satisfied, the implementation of the one-step posterior algorithm requires knowledge of the EIF of $\psi$. When $Q$ is a $\pi$-fixed intervention (i.e., $Q_P \equiv Q$ for all $P \in \mathcal{P}$), \citet[Lemma 2]{Kennedy2018ipsi} gives the EIF for $\psi_{P_0}(Q)$ as follows.
\begin{definition}\label[definition]{def:static_eif}
     The efficient influence function (EIF) for a $\pi$-fixed stochastic intervention 
     is
     \begin{equation} \label{eq:eif-fixed}
         \phi_{\text{fix}, P_0} = \int_{\mathcal{A}} \mu(X,a)dQ(a \mid X) +   \frac{dQ(A\mid X)}{dP_0(A\mid X)} \left[  Y - \mu(X,A) \right] - \psi_{P_0}(Q).
     \end{equation}
\end{definition}
When $Q$ is a $\pi$-dependent stochastic intervention, such as the IPSI and PTI interventions from \cref{subsec:stoch-int-defintion}, the intervention itself depends on features of $P$ and may therefore vary according to local perturbations of $P$. Consequently, in this setting the causal estimand is a pathwise-dependent functional whose variation reflects both changes in the outcome regression and changes in the intervention distribution induced by changes in $P$. This dependence alters the tangent space of the model and requires an adjustment to the efficient influence function \citep{Kennedy2018ipsi}.
\begin{definition}\label[definition]{def:dynamic_eif}
The efficient influence function (EIF) for a $\pi$-dependent stochastic intervention is
\begin{equation} \label{eq:eif-dependent}
    \phi_{\text{dep}, P_0} = \phi^{*}_{\text{fix}, P_0} + 
    \int_{\mathcal{A}} \chi(X,A;a)\mu(X,a)d\nu(a) - \psi_{P_0}(Q)
 \end{equation}
 where $\phi^{*}_{\text{fix}, P_0}$ is the uncentered $\pi$-fixed efficient influence function, $\nu$ is a dominating measure for the distribution of $A$, and $\{ \mathbf{1}(X = x)/ dP(x) \} \chi(X, A; a)$ is the efficient influence function for $Q(a \mid x)$.
\end{definition}
Thus, the $\pi$-dependent efficient influence function decomposes into two additive components: the (uncentered) $\pi$-fixed EIF and an additional term that accounts for the sensitivity of $Q_P$ to perturbations of $P$. 
This second term requires an EIF for $Q(a \mid x)$, which we write generally as 
\begin{equation*}
    \frac{\mathbf{1}(X = x)}{dP(x)} \chi(X, A; a).
\end{equation*} 
For binary treatments, a sufficient condition for deriving $\chi$ is that the intervention depends on $P$ only through the propensity score, $dQ_{P,\delta}(a \mid x) = \Gamma_\delta(a, \pi_P(x); x)$, where $\pi \mapsto \Gamma_\delta(a, \pi; x)$ is twice differentiable. It then follows from the chain-rule property of influence functions \citep{kennedy2024semiparametric} that 
\begin{equation}
\label{eq:chi-def}
    \chi(X, A; a) = \bigg\{\frac{\partial}{\partial \pi} \Gamma_\delta(a,\pi(X);X)\bigg\}\{A - \pi(X)\}.
\end{equation}
This is particularly relevant when considering IPSI and PTI causal estimands.

\section{Semiparametric Bernstein-von Mises Theory for Stochastic Intervention Effects} \label{sec:Correction}

As outlined above, \citet{Yiu2025pcorrect} show that,
under \cref{ass:yiu-main}, the one-step posterior satisfies the semiparametric BvM theorem. Although general, these conditions are formulated in terms of the estimand-specific second-order remainder and posterior-indexed class of efficient influence functions, and can therefore be difficult to verify directly. In particular, the no second-order bias and $L_2$-convergence conditions depend on the stochastic intervention under study. Moreover, when the nuisance functions $\mu$ and $\pi$ are assigned flexible nonparametric priors, the resulting EIF class may not be contained in a fixed $P_0$-Donsker class. In that case, convergence of the EIF under the empirical process, together with suitable envelope conditions, must be established directly. Consequently, it may be unclear whether the general BvM result applies for a particular stochastic intervention and choice of nuisance function priors.

\subsection{BvM Result}
This motivates sufficient conditions for the BvM stated directly in terms of the posterior behavior of the nuisance functions and the regularity of the stochastic intervention. Our result relies on two main properties. First, under suitable regularity conditions, the second-order remainder can be bounded by products of the estimation errors in $\pi$ and $\mu$. Second, smoothness of the intervention map, together with bounded intervention weights and outcome regressions, implies that the EIF is pointwise Lipschitz in these nuisance functions, with a Lipschitz constant that has a finite $\kappa$-th moment. This Lipschitz property transfers posterior contraction and sup-norm entropy control from the nuisance function classes to the EIF class, yielding the $L_2(P_0)$-convergence and empirical process control required by \cref{ass:yiu-main}; the accompanying boundedness and moment conditions provide the required envelope control. \cref{ass:main-bvm} formalizes these requirements and provides a tractable route for verifying the semiparametric BvM for the one-step posterior under a given stochastic intervention and nuisance function prior.

\begin{assumption}
\label{ass:main-bvm}
Fix $\delta \in \mathcal{D}$, which we suppress in the notation. There exists a sequence of measurable subsets $(\mathcal{S}_n)_n$ of $\mathcal{P}$ satisfying the following properties:
\begin{enumerate}
    \item[(a)] \textit{Posterior mass.}
    \begin{equation*}
    \Pi(\mathcal{S}_n \mid Z_{1:n}) \xrightarrow{P_0} 1,
    \end{equation*}
    \item[(b)] \textit{Nuisance contraction and second-order rate}.

    There exist sequences $\rho_{n}, \varepsilon_{n} \to 0$ such that
    \begin{equation*}
        \sup_{P \in \mathcal{S}_n}\|\pi_{P} - \pi_{P_0}\|_{2} \le \rho_{n},
        \quad
        \sup_{P \in \mathcal{S}_n} \max_{a \in \{0,1\}}
        \|\mu_{P}(\cdot, a) - \mu_{P_0}(\cdot, a)\|_{2} \le \varepsilon_{n},
    \end{equation*}
    where the norm on $\mu_{P}(\cdot,a)$ is $L_2(P_0)$ over $X$. In addition,
    \begin{equation*}
        \sqrt{n}\, \rho_{n}\left(\rho_{n} + \varepsilon_{n}\right) \to 0.
    \end{equation*}

    \item[(c)] \textit{Regularity conditions.}

    For some $C < \infty$ and some $\kappa > 4$, the following hold,
    \begin{enumerate}
        \item[(i)] (Bounds on the truth).
        \begin{equation*}
            P_0\big[|Y|^{\kappa}\big] \le C, \qquad
            \sup_{x,a} P_0\big[Y^2 \mid X = x, A = a\big] \le C,
        \end{equation*}
        \begin{equation*}
            \|\mu_{P_0}\|_\infty \le C, \qquad
            \left\| \frac{dQ_{P_0}}{dP_{0}} \right\|_\infty \le C.
        \end{equation*}
        \item[(ii)] (Bounds on the sieve).
        \begin{equation*}
            \sup_{P \in \mathcal{S}_n} \|\mu_{P} \|_\infty \le C, \quad \sup_{P \in \mathcal{S}_n} \left\| \frac{dQ_{P}}{dP}  \right\|_\infty \le C.
        \end{equation*}
        \item[(iii)] (Smooth intervention map). There is a function $\Gamma: \{0,1\} \times [0,1] \times \mathcal{X} \to [0,1]$, jointly measurable in $(\pi, x)$ and not depending on $P$, with $\Gamma(0, \pi; x) + \Gamma(1, \pi; x) = 1$, such that for all $P \in \mathcal{S}_n \cup \{P_0\}$,
        \begin{equation*}
            dQ_{P} (A = a \mid X = x) = \Gamma(a, \pi_{P}(x); x).
        \end{equation*}
        Further, for each $a \in \{0,1\}$ and $x \in \mathcal{X}$, the map $\pi \mapsto \Gamma(a, \pi; x)$ is twice continuously differentiable on $[0,1]$ with
        \begin{equation*}
            \sup_{a \in \{0,1\}} \sup_{x \in \mathcal{X}} \sup_{\pi \in [0,1]}
            \left\{ |\partial_\pi \Gamma(a, \pi; x)| + |\partial^2_\pi \Gamma(a, \pi; x)| \right\}
            \leq C.
        \end{equation*}
        Moreover, $\Gamma(1, 0; x) = 0$ and $\Gamma(0, 1; x) = 0$ for all $x$.\footnote{Together, these conditions imply that the intervention depends on $P$ only through the propensity score, and that the intervention never assigns a treatment that has zero probability under $P$.}
    \end{enumerate}

    \item[(d)] \textit{Entropy condition on nuisance classes}.
Define the centered nuisance function classes
\begin{equation*}
\mathscr{P}_{n} = \{ \pi_{P} - \pi_{P_0} : P \in \mathcal{S}_n \}, \qquad
\mathscr{M}_{n} = \{ \mu_{P} - \mu_{P_0} : P \in \mathcal{S}_n \},
\end{equation*}
and let $\gamma_n \asymp \rho_{n} + \varepsilon_{n}$. Then
\begin{equation*}
J_\infty\big(\gamma_n,\mathscr{P}_{n}\big)
+ J_\infty\big(\gamma_n,\mathscr{M}_{n}\big)
= o\!\left(1 \wedge \gamma_n\, n^{\frac{\kappa-2}{4\kappa-4}} \right),
\end{equation*}
where $J_\infty(r,\mathscr F)=\int_0^r\sqrt{1+\log N(\varepsilon,\mathscr F,\|\cdot\|_\infty)}\,d\varepsilon$
is the sup-norm entropy integral.
\end{enumerate}
\end{assumption}

\cref{ass:main-bvm} provides primitive conditions under which the high-level requirements of \cref{ass:yiu-main} can be verified. These conditions do not correspond one-to-one with the clauses of \cref{ass:yiu-main}. Rather, the posterior mass condition in part (a), the nuisance contraction rates in part (b), the boundedness and intervention stability conditions in part (c), and the local complexity condition in part (d) work jointly: part (a) supplies the sets $\mathcal{S}_n$, parts (b) and (c) control the second-order remainder and the $L_2(P_0)$ behavior of the efficient influence function, and parts (b)--(d) together provide the required empirical process and envelope control. We formalize this implication in the following theorem.

\begin{theorem}\label{thm:bvm_single_time}
    Under \cref{ass:main-bvm}, 
    the one-step posterior of a stochastic intervention effect satisfies \cref{def:BVM-def}, the semiparametric Bernstein-von Mises theorem. 
\end{theorem}

The full proof is given in the Supplementary Material, where we establish a longitudinal version of the result and recover \cref{thm:bvm_single_time} as the single-time-point special case. The proof proceeds by verifying that the three conditions of \cref{ass:yiu-main} hold under \cref{ass:main-bvm}. Verification of the no second-order bias condition is the most delicate step. It is the only one that depends on the form of the intervention, and it is the main obstacle to extending the result beyond smooth stochastic interventions. For stochastic interventions satisfying \cref{ass:main-bvm} (c)(iii), we derive an exact expression for the remainder and, after integrating out the observed treatment, bound it using a second-order Taylor expansion of $\Gamma$ in $\pi$ and the Lipschitz continuity of the intervention weights in $\pi$ (\cref{lem:no-snd-bias}). This gives
\begin{equation}\label{eq:remainder-outline}
    \sup_{P \in \mathcal{S}_n} |R_2(P_0, P)| \lesssim \rho_n(\rho_n + \varepsilon_n),
\end{equation}
so \cref{ass:main-bvm} (b) implies the no second-order bias condition. The $\rho_n^2$ term arises from the dependence of the intervention on $\pi$, and for $\pi$-fixed interventions it is absent and the rate condition reduces to the standard product-rate condition. The remaining conditions of \cref{ass:yiu-main} are verified in \cref{subsec:emp-process}.

\subsection{Empirical Process Conditions} \label{subsec:emp-process}

 The remaining conditions of \cref{ass:yiu-main}, $L_2$ convergence and control of the empirical process, depend on the efficient influence function only through its regularity in the nuisance functions. We state this reduction in general form, since it applies to any functional which satisfies the following assumptions.

\begin{proposition}[Empirical process conditions] \label[proposition]{prop:emp-process}
Suppose that $\phi^*_P$ depends on $P$ only through real-valued nuisance functions $\eta_P = (\eta_{P,1}, \ldots, \eta_{P,K})$. Let
\begin{equation*}
    d(P,P') = \max_{k \le K} \|\eta_{P,k} - \eta_{P',k}\|_\infty,
    \qquad
    \mathscr{H}_{k,n} = \{\eta_{P,k} - \eta_{P_0,k} : P \in \mathcal{S}_n\}.
\end{equation*}
Suppose there exist $\kappa > 4$, $D < \infty$, and a sequence $\gamma_n \to 0$ such that:
\begin{enumerate}
    \item[(a)] $\sup_{P \in \mathcal{S}_n} d(P, P_0) \le D$;
    \item[(b)] There is a measurable $M \ge 1$, not depending on $n$, with $P_0[M^\kappa] < \infty$, $|\phi^*_{P_0}| \le M$, and
    \begin{equation*}
        |\phi^*_P(z) - \phi^*_{P'}(z)| \le M(z)\, d(P, P') \quad \text{for all } P, P' \in \mathcal{S}_n \text{ and all } z;
    \end{equation*}
    \item[(c)] $\sup_{P \in \mathcal{S}_n} \|\phi^*_P - \phi^*_{P_0}\|_{L_2(P_0)} \lesssim \gamma_n$;
    \item[(d)] $\sum_{k=1}^K J_\infty(\gamma_n, \mathscr{H}_{k,n}) = o\big(1 \wedge \gamma_n\, n^{\frac{\kappa-2}{4\kappa-4}}\big)$.
\end{enumerate}
Then \cref{ass:yiu-main} (b) and (c) hold.
\end{proposition}
The $L_2(P_0)$-convergence condition, \cref{ass:yiu-main} (b), follows directly from \cref{prop:emp-process}(c), as only $\gamma_n \to 0$ is needed and no rate is required. For the envelope condition, \cref{prop:emp-process} (a) and (b) give $|\phi^*_P| \le (1+D)M$ for all $P \in \mathcal{S}_n$. Since this envelope does not depend on $n$ and $P_0[M^\kappa] < \infty$ with $\kappa > 4$, it is uniformly square-integrable and satisfies $P_0[(1+D)^4M^4] = O(1) = o(n)$, which is \cref{ass:yiu-main}(c)(ii). 

The remaining condition, convergence of $\phi^*_P$ under the empirical process, is the most involved, as we must show $\sup_{f \in \mathcal{F}^*_n} |\mathbb{G}_n f| \xrightarrow{P_0} 0$ for the class $\mathcal{F}^*_n := \{\phi^*_P - \phi^*_{P_0} : P \in \mathcal{S}_n\}$. We control it with a maximal inequality of \citet[Theorem~2.14.8]{vandervaartwellner2023weak}, which bounds the expected supremum of the empirical process over $\mathcal F^*_n$ in terms of the class's uniform entropy integral $J$, its relative size $\gamma_n$, and the moments of an envelope $F$,\footnote{Here $J(\delta, \mathcal{F}^*_n)=\sup_Q\int_0^\delta \sqrt{1+\log N(\varepsilon\|F\|_{Q,2},\mathcal F^*_n,L_2(Q))}\,d\varepsilon$, with the supremum over finitely discrete probability measures $Q$;
$\gamma_n\in(0,1]$ is any constant such that $P_0 f^2 < \gamma_n^2\, P_0 F^2$
for every $f\in\mathcal F^*_n$; and $\kappa>4$ is such that $\|F\|_{P_0,\kappa}<\infty$.}
\begin{equation}
\label{eq:thm-2.14.8}
\mathbb E^*_{P_0}\|\mathbb G_n\|_{\mathcal F^*_n}
\lesssim
J(\gamma_n, \mathcal{F}^*_n)\left[1+
\frac{J\big(\gamma_n^{\frac{\kappa-4}{\kappa-2}}\big)}{\gamma_n^2\sqrt n}
\left(\frac{\|F\|_{P_0,\kappa}}{\|F\|_{P_0,2}}\right)^{\frac{\kappa}{\kappa-2}}
\right]^{1-\frac2\kappa}
\|F\|_{P_0,2}.
\end{equation}
Then through the Lipschitz bound in \cref{prop:emp-process}(b), we can obtain a bound of the form 
\begin{equation*}
    J(\gamma_n, \mathcal{F}^*_n) \lesssim \sum_k J_\infty(\gamma_n, \mathscr{H}_{k,n}).
\end{equation*}
Finally, \cref{prop:emp-process}(d) makes the right-hand side of \eqref{eq:thm-2.14.8} $o(1)$, and Markov's inequality gives \cref{ass:yiu-main}(c)(i). 

\cref{prop:emp-process} is not exclusive to stochastic interventions; it applies to any functional whose uncentered EIF is pointwise Lipschitz in finitely many nuisance functions. Therefore if the conditions of \cref{prop:emp-process} can be verified for a particular setting, only \cref{ass:yiu-main}(a) requires separate treatment. For example, under strong positivity the EIF of the average treatment effect satisfies \cref{prop:emp-process}(b) with $M \propto 1 + |Y|$, and the second-order remainder of the average treatment effect admits a product form in the propensity score and outcome regression, allowing for a rate condition to imply \cref{ass:yiu-main}(a). Thus, \cref{prop:emp-process} provides verifiable sufficient conditions for a route to a semiparametric BvM theorem without a Donsker or sample-splitting condition.

To invoke \cref{prop:emp-process} for a stochastic intervention, the Supplementary Material shows that, under \cref{ass:main-bvm}(c), the uncentered EIF satisfies the pointwise Lipschitz bound
\begin{equation*}
    | \phi^*_P - \phi^*_{P'} | \lesssim (1 + |Y|) \Big( \sum_{a \in \{0,1\}} |\mu_P(X,a) - \mu_{P'}(X,a)| + |\pi_P(X) - \pi_{P'}(X)| \Big).
\end{equation*}
Combined with the bounds in \cref{ass:main-bvm}(c)(i)--(ii), the rates in \cref{ass:main-bvm}(b), and the entropy control of the nuisance functions in \cref{ass:main-bvm}(d), we may invoke \cref{prop:emp-process} to finish proving \cref{thm:bvm_single_time} for our class of smooth stochastic interventions.

\subsection{Verification of Stochastic Intervention Conditions}
\label{subsec:intervention-conditions}
\cref{thm:bvm_single_time} is useful because it recasts the high-level BvM requirements in terms of the posterior behavior of the nuisance functions and the choice of stochastic intervention. In particular, the intervention-specific requirements of \cref{ass:main-bvm} are concentrated in part (c): smooth dependence of the intervention on the propensity score and bounded intervention likelihood ratios. For a $\pi$-fixed intervention, $Q_{P} \equiv Q$, so $\Gamma(a, \pi; x) = dQ(a \mid x)$ does not depend on $\pi$ and so its derivatives with respect to $\pi$ are trivially bounded. However, the boundary condition in (c) (iii) fails when $dQ(1 \mid x) > 0$ and the weights are unbounded as $\pi_P(a \mid x) \to 0$. Therefore $\pi$-fixed interventions require strong positivity, and are covered by clipping the propensity score (\cref{cor:clipped}(ii)). For a $\pi$-dependent intervention, smoothness and boundedness must instead be established from the map $\pi \mapsto \Gamma(\cdot, \pi; x)$. The next result identifies a broad class of interventions where both follow directly from construction.

\begin{proposition}\label[proposition]{prop:logodds_family}
Let $s:[0,1]\times\mathcal{X}\to\mathbb{R}$ be jointly measurable, with $\pi \mapsto s(\pi; x)$ twice continuously differentiable for each $x$ and
\begin{equation*}
\sup_{x \in \mathcal{X}} \sup_{\pi \in [0,1]} \big\{ |s(\pi;x)| + |\partial_\pi s(\pi;x)| + |\partial_\pi^2 s(\pi;x)| \big\} < \infty.
\end{equation*}
Define the log-odds shift intervention by
\begin{equation*}
\Gamma_{s}(1, \pi; x)=\frac{\pi\,e^{s(\pi;x)}}{\pi\,e^{s(\pi;x)}+1-\pi}, \qquad \Gamma_s(0,\pi;x) = 1 - \Gamma_s(1,\pi;x).
\end{equation*}
Then $\Gamma_s$ satisfies \cref{ass:main-bvm}(c)(iii) and the intervention-weight bounds in \cref{ass:main-bvm}(c)(i)--(ii). In particular, the incremental propensity score intervention is obtained by taking $s \equiv \log\delta$, and the required bounds hold uniformly for $\delta$ in any compact subset of $(0,\infty)$.
\end{proposition}

A full proof is given in the Supplementary Material. Because $s$, $\partial_\pi s$, and $\partial_\pi^2 s$ are bounded uniformly in $(\pi,x)$, $\exp\{s(\pi;x)\}$ is bounded above and away from zero. The denominator in \cref{prop:logodds_family} is therefore uniformly bounded away from zero, which yields both bounded intervention weights and bounded first and second derivatives. While the IPSI is the simplest intervention that \cref{prop:logodds_family} allows, many other interventions may arise that are of scientific interest. As mentioned in the discussion of \citet{Kennedy2018ipsi}, $s$ could depend on covariates, which \cref{prop:logodds_family} allows. For example, we may only want to shift the odds of treatment for some subset of the population of interest.

Although the IPSI is covered by \cref{prop:logodds_family}, the power tilt intervention defined in \eqref{eq:PT_density} is not. In particular, because the power tilt instead rescales the log-odds, its weights blow up like $\pi^{\delta-1}$ near the boundary when $\delta < 1$, and its second derivative like $\pi^{\delta-2}$ when $\delta < 2$. Therefore for the PTI, \cref{ass:main-bvm}(c)(ii)--(iii) holds \textit{only} for $\delta \ge 2$ or $\delta = 1$. For other $\delta > 0$, the conditions are satisfied under an additional assumption of strong overlap, $\tilde c \le \pi_{P_0}(x) \le 1 - \tilde c$ for some $\tilde c > 0$, together with clipping of the propensity score (\cref{cor:clipped}).

Finally, we note that the intervention conditions in \cref{ass:main-bvm}(c) are sufficient, rather than necessary, for \cref{thm:bvm_single_time}. They deliberately exclude intervention maps that are invalid, nonsmooth, or induce unbounded likelihood ratios. For example, $\Gamma_\delta(1, \pi; x) = \delta \pi$ need not take values in $[0,1]$ when $\delta > 1$ and
may violate weak positivity at the boundary. The capped modification,
\begin{equation*}
    \Gamma_\delta(1, \pi; x)=\min\{1,\delta\pi\}
\end{equation*}
has a kink at $\pi = 1 / \delta$, while the threshold rule
\begin{equation*}
    \Gamma_\delta(1, \pi; x)=\mathbf{1}\{\pi>c\}
\end{equation*}
is discontinuous. For these maps, the Taylor expansion and Lipschitz arguments used above do not apply without additional assumptions. Such interventions may still admit specialized analyses, but they fall outside the scope of \cref{thm:bvm_single_time}.

\subsection{Verification of Nuisance Function Conditions}

Together, \cref{thm:bvm_single_time} and \cref{prop:logodds_family} separate the intervention-specific regularity conditions and prior-dependent conditions governing posterior concentration and complexity. Importantly, these conditions are local to the nuisance-function sieves rather than the full model class, so it is sufficient that the posterior concentrate on subsets controlling $\pi$ and $\mu$, with no separate contraction or complexity condition imposed on the marginal covariate distribution. Thus, once the boundedness and intervention-stability conditions in \cref{ass:main-bvm}(c) have been verified, the remaining prior-dependent requirements are the posterior mass condition in part (a), the nuisance contraction and second-order rate in part (b), and the local complexity condition in part (d).

Next we construct the sets $\mathcal{S}_n$ in \cref{ass:main-bvm}. In practice, priors on $\pi$ and $\mu$ come with their own sieves and contraction rates, so $\mathcal{S}_n$ is built by intersecting the nuisance function sieves. The following lemma shows that this construction implies \cref{ass:main-bvm}(a) and (b), and reduces (d) to entropy control on the sieves.

\begin{lemma}[Nuisance sieve]\label[lemma]{lem:sieve}
Suppose $\pi_{P_0}(x) \in [c_\pi, 1 - c_\pi]$ for all $x$ and some $c_\pi > 0$. Suppose there exist sets of functions $\mathcal{H}^\pi_n$ and $\mathcal{H}^\mu_n$ and sequences $\rho_n, \bar\varepsilon_n \to 0$ such that, in $P_0$-probability,
\begin{align*}
    \Pi\big(\pi_P \in \mathcal{H}^\pi_n,\ \|\pi_P - \pi_{P_0}\|_{L_2(P_0)} \le \rho_n \mid Z_{1:n}\big) &\to 1, \\
    \Pi\big(\mu_P \in \mathcal{H}^\mu_n,\ \|\mu_P(X,A) - \mu_{P_0}(X,A)\|_{L_2(P_0)} \le \bar\varepsilon_n \mid Z_{1:n}\big) &\to 1.
\end{align*}
Let $\mathcal{S}_n$ be the set of $P$ satisfying both events. Then \cref{ass:main-bvm}(a) holds, \cref{ass:main-bvm}(b) holds with $\varepsilon_n = c_\pi^{-1/2}\bar\varepsilon_n$ provided $\sqrt{n}\,\rho_n(\rho_n + \varepsilon_n) \to 0$, and $\mathscr{P}_n \subseteq \mathcal{H}^\pi_n - \pi_{P_0}$, $\mathscr{M}_n \subseteq \mathcal{H}^\mu_n - \mu_{P_0}$.
\end{lemma}

\begin{proof}
\noindent Posterior mass follows from a union bound, and the per-arm bound from
\begin{align*}
    \|\mu_P(X,A) - \mu_{P_0}(X,A)\|_{L_2(P_0)}^2 &= P_0\Big[\sum_{a} \pi_{P_0}(a \mid X)\{\mu_P(X,a) - \mu_{P_0}(X,a)\}^2\Big] \\
    &\ge c_\pi \max_a \|\mu_P(\cdot,a) - \mu_{P_0}(\cdot,a)\|_{L_2(P_X)}^2.
\end{align*}
\end{proof}

In \cref{lem:sieve} we must assume strong positivity, which converts contraction in the observed-arm norm $\|\mu_P(X,A) - \mu_{P_0}(X,A)\|_{L_2(P_0)}$\footnote{The norm in which regression posteriors typically contract.} into the per-arm norm used in \cref{ass:main-bvm} (b). The per-arm norm can be replaced by the observed-arm norm, at the expense of a separate $L_2$-convergence argument that does not route through the pointwise Lipschitz bound. We use the per-arm norm for simplicity, as strong positivity is already implied by the nuisance models we consider (see \cref{sec:SoftBART}).

\cref{lem:sieve} implies that the entropy condition \cref{ass:main-bvm}(d) reduces to sup-norm covering bounds for the sieves $\mathcal{H}^\pi_n$ and $\mathcal{H}^\mu_n$, with no condition imposed on the prior for $P_X$. Further, \cref{ass:main-bvm}(c)(ii)--(iii) require boundedness conditions on the posterior draws that may not hold for a nonparametric prior, or for an intervention whose weights or derivatives are unbounded near the boundary, such as the power tilt. Both can be solved by clipping.

\begin{definition}[Clipped one-step posterior] \label[definition]{def:clipped}
For constants $C > 0$ and $\tilde c \in (0, 1/2)$, let
\begin{equation*}
    T_C(v) = \max\{-C, \min\{C, v\}\}, \qquad T_{\tilde c}(\pi) = \min\{\max\{\pi, \tilde c\}, 1 - \tilde c\}.
\end{equation*}
For a posterior draw $P^{(b)}$, let $\tilde P^{(b)}$ denote the law obtained by replacing $\mu_{P^{(b)}}$ with $T_C(\mu_{P^{(b)}})$ and $\pi_{P^{(b)}}$ with $T_{\tilde c}(\pi_{P^{(b)}})$. The \textit{clipped one-step posterior} is obtained by \cref{alg:onestep} with $\tilde P^{(b)}$ in place of $P^{(b)}$, in both the plug-in and the efficient influence function. Either map may be applied alone.
\end{definition}

When $\|\mu_{P_0}\|_\infty \le C$ and $\pi_{P_0} \in [\tilde c, 1 - \tilde c]$, clipping leaves the target unchanged and does not increase the distance between any draw and the truth.
Clipping the outcome regression is used for the SoftBART results in \cref{sec:SoftBART}, and clipping the propensity score is needed for $\pi$-fixed interventions and the power tilt. This motivates the following corollary.

\begin{corollary}\label[corollary]{cor:clipped}
Consider the clipped one-step posterior of \cref{def:clipped}, and suppose \cref{ass:main-bvm} holds with the following modifications.
\begin{enumerate}
    \item[(i)] If the outcome regression is clipped at $C$, where $\|\mu_{P_0}\|_\infty \le C$, then the sieve bound $\sup_{P \in \mathcal{S}_n}\|\mu_P\|_\infty \le C$ in (c) (ii) is not required.
    \item[(ii)] If the propensity score is clipped at some $\tilde c \in (0,1/2)$, where $\pi_{P_0}(x) \in [\tilde c, 1-\tilde c]$ for all $x$, then (c) (iii) is required only for $\pi \in [\tilde c, 1-\tilde c]$, and neither the boundary condition in (c) (iii) nor the intervention-weight bounds in (c) (i)--(ii) are required.
\end{enumerate}
Then the conclusion of \cref{thm:bvm_single_time} holds. In particular, under (ii), it holds for the power tilt intervention with $\delta$ in any compact subset of $(0, \infty)$, and for any $\pi$-fixed intervention.
\end{corollary}

With the sieve and clipping completed, the remaining prior-dependent requirements are the rate condition in \cref{ass:main-bvm}(b) and the entropy condition in \cref{ass:main-bvm}(d). Many nonparametric priors admit sieves whose log-covering numbers grow like $A_n \log(B_n/\varepsilon)$ for some sequences $A_n$ and $B_n$, including sparse deep ReLU network priors \citep{polson2018posterior, schmidthieber2020nonparametric} and, as we show in \cref{sec:SoftBART}, SoftBART. In that case, a rate condition suffices, as shown in \cref{prop:rate-condition}.

\begin{proposition}\label[proposition]{prop:rate-condition}
Suppose that for each sieve $\mathcal{H}_n \in \{\mathcal{H}^\pi_n, \mathcal{H}^\mu_n\}$,
\begin{equation*}
    \log N(\varepsilon, \mathcal{H}_n, \|\cdot\|_\infty) \lesssim A_n \log(B_n/\varepsilon), \qquad \varepsilon \in (0, \gamma_n],
\end{equation*}
with $A_n \lesssim n\gamma_n^2$, $B_n \lesssim n^{c}$, and $\gamma_n \gtrsim n^{-c}$ for some $c > 0$. If
\begin{equation*}
    \gamma_n = o\big(n^{-\frac{\kappa}{4(\kappa-1)}} (\log n)^{-1/2}\big),
\end{equation*}
then \cref{ass:main-bvm}(b) and (d) hold.
\end{proposition}

The proof, given in the Supplementary Material, follows from a sup-norm entropy integral bound and direct rate calculations. Therefore, for a given nonparametric prior, it suffices to establish a sup-norm covering bound for its sieves and posterior contraction at a sufficiently fast rate. We verify these for SoftBART in \cref{sec:SoftBART}.

\section{SoftBART Priors for Stochastic Intervention Effects} \label{sec:SoftBART}

\citet{chipman2010bart} introduced Bayesian Additive Regression Trees (BART), which is a powerful Bayesian machine learning method that is commonly used for causal inference problems \citep{hill2011bnp, dorie2019competition, hahn2020bcf}. Broadly, BART was developed in the context of the model,
\begin{equation*}
    Y = f(x) + \varepsilon, \quad \varepsilon \sim N(0, \sigma^2),
\end{equation*}
where $f(x) = E[Y \mid X = x]$ is an unknown regression function. BART approximates $f(x)$ using a sum of $K$ regression trees; that is,
\begin{equation*}
    f(x) = \sum_{k=1}^{K} g_k(x; \mathcal{K}_k, \mathcal{L}_k),
\end{equation*}
where each $g_k(x; \mathcal{K}_k, \mathcal{L}_k)$ has tree structure $\mathcal{K}_k$ and leaf parameters $\mathcal{L}_k$. This setup is analogous to the decision tree boosting framework introduced by \citet{friedman2001gbm}, where many weak learners get aggregated into a single flexible function approximation.

The tree structure $\mathcal{K}_k$ specifies a binary partition of the predictor space. BART imposes priors on tree topologies, splitting variables, split points, and leaf values, thereby strongly regularizing the model toward shallow trees that capture additive effects and low-order interactions. Inference is then performed using Markov chain Monte Carlo, which yields a posterior distribution for $f(x)$.

Although BART has been shown to achieve posterior consistency and a near-optimal contraction rate under H\"older smoothness and sparsity assumptions \citep{rockova2019theory}, its reliance on piecewise-constant tree ensembles limits its ability to efficiently approximate smooth functions. SoftBART addresses this limitation by introducing smooth decision rules, yielding a prior that is well aligned with smoother function classes and capable of achieving optimal posterior contraction rates across a range of H\"older smoothness levels \citep{Linero2018brte}. Compared to the hard-decision rules behind BART, SoftBART utilizes probabilistic paths, with $x$ going left at branch $b$ with probability
\begin{equation*}
    s_b(x) = \varsigma \left(\frac{x_{j_b} - c_b}{\tau} \right),
\end{equation*}
for some covariate $x_{j_b}$ and split location $c_b \in [0,1]$. Here $\tau$ is the bandwidth, which controls the smoothness of the split, and $\varsigma(\cdot)$ is some function mapping $\mathbb{R}$ to $[0,1]$. While the logistic function is most commonly used, other smooth functions may also be employed. If SoftBART is to be used in our setup, a bound on its sup-norm covering number is required.

We place independent SoftBART priors on the probit latent function $f$, with $\pi = \Phi(f)$, and on the outcome regression $m(x,a)$, treating $a$ as an additional covariate. The priors follow \citet{Linero2018brte}, as stated in \cref{ass:SoftBart_prior} in the Supplementary Material. We work on the sieves $\tilde{\mathcal{F}}_n$ of \cref{def:SoftBART-Sieve}: functions with at most $K_n$ trees of depth at most $\bar D$, so that each tree has at most $\Lambda = 2^{\bar D}$ leaves, a shared bandwidth $\tau \in [\tau_n^L, \tau_n^U]$, split locations in $[0,1]$, and leaf values in $[-\bar\ell_n, \bar\ell_n]$, where $K_n$, $\bar\ell_n$, $\tau_n^U$, and $(\tau_n^L)^{-1}$ grow polynomially in $n$.

\begin{lemma} \label[lemma]{lem:SoftBART-covering}
    For the sieve of \cref{def:SoftBART-Sieve}, with $\varsigma$ Lipschitz and $p$ fixed, the sup-norm covering number of the SoftBART sieve admits the following bound,
    \begin{equation*}
        \log N\left(\varepsilon, \tilde{\mathcal{F}}_n, \|\cdot\|_\infty\right) \lesssim
     K_n \log \left(\frac{\tau_n^U \, \bar\ell_n \, K_n} {\varepsilon \, \tau_n^L} \right).
    \end{equation*}
\end{lemma}
The full proof is contained in the Supplementary Material, but we briefly provide an illustrative overview. Due to its additive nature, the covering number of SoftBART can be bounded by the product of the covering numbers of the individual trees,
\begin{equation*}
    N\left(\varepsilon, \tilde{\mathcal{F}}_n, \|\cdot\|_\infty\right) \le \prod_{k=1}^{K_n} N\left(\varepsilon/K_n, \mathcal{G}_n, \|\cdot\|_\infty\right),
\end{equation*}
where $\mathcal{G}_n$ is the class of possible functions for a single tree. A single SoftBART tree can be expressed as
\begin{equation*}
    g(x; \mathcal{K}, \mathcal{L}) = \sum_{l = 1}^{\Lambda} \omega_l(x) \ell_l,
\end{equation*}
for the gates $\omega_l$ and leaf values $\ell_l$. Since the gates sum to one, we cover $g$ in terms of $\omega_l$ and $\ell_l$ separately to obtain
\begin{equation*}
    N\left(\varepsilon/K_n, \mathcal{G}_n, \|\cdot\|_\infty\right) \le
    N\left(\frac{\varepsilon}{2\bar\ell_n\Lambda K_n}, \Omega_n, \|\cdot\|_\infty\right)
    N\left(\frac{\varepsilon}{2 K_n}, \mathcal{V}_n, \|\cdot\|_\infty\right),
\end{equation*}
where $\Omega_n$ denotes the class of all SoftBART gates for a fixed tree configuration, each of which is a product of at most $\bar D$ splitting functions, and $\mathcal{V}_n = [-\bar\ell_n, \bar\ell_n]^\Lambda$ is the class of leaf values. Covers for $\Omega_n$ are obtained by covering the splitting functions, which are Lipschitz in $(c_b, \tau)$ with constants of order $1/(\tau_n^L)^2$, and combining them through a telescoping product. The resulting bound is then multiplied by the number of tree configurations, which gives the desired result. Because $\Phi$ and the clipping map $T_C$ are Lipschitz, the same bound holds, up to constants, for the propensity and clipped outcome sieves.

Of particular interest is the posterior contraction rate induced by the prior. For SoftBART, \citet{Linero2018brte} show that, when $f$ belongs to a H\"older space and depends on at most $d$ of the $p$ covariates, the posterior contracts at the minimax rate $n^{-\alpha/(2\alpha+d)} + \sqrt{n^{-1}d\log p}$ of \citet{yang2015minimax}, up to logarithmic factors. Their result is established in the empirical $L_2$ norm for fixed-design regression, whereas \cref{lem:sieve} requires contraction in $L_2(P_0)$ under random design, for both the outcome regression and the probit propensity score. We obtain this by combining their prior-mass and sieve-mass bounds with \citet[Theorem 8.9]{ghosal2017fundamentals}.

\begin{lemma}\label[lemma]{lem:SoftBART-contraction}
Suppose $f_0 \in C^{\alpha,R}([0,1]^p)$ and $\mu_{P_0} \in C^{\alpha,R}([0,1]^{p+1})$ each depend on at most $d$ coordinates, and $Y \mid X, A \sim N(\mu_{P_0}(X,A), \sigma_0^2)$ with $\sigma_0$ known. Let $\epsilon_n = n^{-\alpha/(2\alpha+d)}(\log n)^{t + 1/2}$ with $t = \alpha(d+1)/(2\alpha+d)$. Then for $M$ sufficiently large and $C \ge \|\mu_{P_0}\|_\infty$, in $P_0$-probability,
\begin{equation*}
    \Pi\big(\|\Phi(f) - \pi_{P_0}\|_{L_2(P_0)} > M\epsilon_n \mid Z_{1:n}\big) \to 0, \qquad
    \Pi\big(\|T_C(m) - \mu_{P_0}\|_{L_2(P_0)} > M\epsilon_n \mid Z_{1:n}\big) \to 0,
\end{equation*}
and the posterior mass of the sieves $\tilde{\mathcal{F}}_n$ tends to one.
\end{lemma}

Note that since $\|f_0\|_\infty \le R$, the propensity score satisfies strong positivity with $c_\pi = \Phi(-R)$. Combining \cref{lem:SoftBART-covering,lem:SoftBART-contraction} with \cref{lem:sieve}, \cref{prop:rate-condition}, and \cref{cor:clipped} gives the following result.

\begin{corollary}\label[corollary]{cor:SoftBART-bvm}
Suppose the conditions of \cref{lem:SoftBART-contraction} hold with $\alpha > d/2$, and the intervention satisfies \cref{ass:main-bvm}(c)(iii). Then the one-step posterior with SoftBART priors on $f$ and $m$, with the outcome regression clipped at $C$, satisfies the conclusion of \cref{thm:bvm_single_time}.
\end{corollary}

Unlike the longitudinal extension of \cref{thm:bvm_single_time} in the Supplementary Material, our SoftBART results are restricted to the single-time-point setting. In the longitudinal setting, the sequential regressions are defined recursively as regressions of pseudo-outcomes constructed from later regressions and the intervention, and their posterior contraction is not established here. Because $Y$ has moments of all orders, $\kappa$ in \cref{prop:rate-condition} may be taken arbitrarily large, so the rate condition reduces to $\alpha/(2\alpha+d) > 1/4$, i.e., $\alpha > d/2$. 
If instead $Y$ only has a finite $\kappa$-th moment for some $\kappa > 4$, the rate condition becomes $\alpha/(2\alpha+d) > 1/4 + 1/\{4(\kappa-1)\}$, so the required improvement over $n^{-1/4}$ is inversely proportional to $\kappa - 1$. That is, fewer moments require smoother nuisance functions. If the propensity score is additionally clipped at a known $\tilde c \le \Phi(-R)$, the same conclusion holds for any intervention covered by \cref{cor:clipped}(ii), including the power tilt for all $\delta > 0$ and $\pi$-fixed interventions.


Much of this analysis will be useful for other nonparametric priors, leaving many avenues for future work. Many nonparametric priors attain minimax-optimal posterior contraction rates, including Bayesian neural networks \citep{egels2025posterior}, Gaussian processes \citep{vandervaart2008gp}, and Dirichlet process mixture models \citep{ghosal2007dirichlet,shen2013dirichlet}. Therefore, given a sup-norm covering bound of the form in \cref{prop:rate-condition} and random-design contraction in $L_2(P_0)$, the argument of this section extends to any such prior.

In summary, we have derived a sup-norm covering bound for SoftBART in \cref{lem:SoftBART-covering} and random-design contraction for the outcome regression and probit propensity score in \cref{lem:SoftBART-contraction}. Therefore, assuming Gaussian outcome errors, strong positivity, and sufficiently smooth nuisance functions, the one-step posterior with SoftBART nuisance function priors satisfies the semiparametric Bernstein--von Mises theorem for a broad class of stochastic interventions.

\section{Simulation Study}\label{sec:simulation_study}

While asymptotic results are important, in this section we investigate the finite sample performance of the one-step posterior estimator in the incremental propensity score intervention setting. The setup is based on the simulation study considered in \citet{Kennedy2018ipsi}, which borrows the following data-generating process from \citet{Kang2007dr}:
\begin{equation*}
    (X_1, X_2, X_3, X_4) \sim N(0, I), 
\end{equation*}
\begin{equation*}
    P_0(A = 1\mid X) = \text{expit}(-X_1 + 0.5X_2 - 0.25X_3 - 0.1X_4),
\end{equation*}
\begin{equation*}
    (Y\mid X , A) \sim N(\mu(X,A), 1),
\end{equation*}
where the outcome mean is $\mu(x, a) = 200 + a\{10 + 13.7(2x_1 + x_2 + x_3 + x_4)\}$. 
We further consider a case with model misspecification, where nuisance functions are estimated using nonlinearly transformed covariates $t(X)$, as defined in \citet{Kang2007dr}.

We focus on estimating $\psi_{\text{IPSI}, P_0, \delta} := \psi_{\text{IPSI}, P_0}(Q_{P_0, \delta})$, the expected outcome under the IPSI intervention defined in \eqref{eq:IPSI_density}. When the identifiability conditions hold in the single-timepoint setting, this is equivalent to estimating
\begin{equation}
    \psi_{\text{IPSI}, P_0, \delta} = E_{P_0}\left[\frac{\delta \pi_{P_0}(X) \mu_{P_0}(X,1) + \{ 1 - \pi_{P_0}(X)\}\mu_{P_0}(X,0)}{\delta \pi_{P_0}(X) + 1 - \pi_{P_0}(X)} \right].
\end{equation}
The efficient influence function can be found using \cref{def:dynamic_eif} and \eqref{eq:IPSI_density}; it appears in the Supplementary Material for reference.

We utilize three different BART-type priors for the outcome regression and propensity score nuisance functions, namely, BART, SoftBART and SoftBCF. SoftBCF follows the Bayesian causal forest model of \citet{hahn2020bcf}, which decomposes the outcome regression into a prognostic function and a treatment effect function with independent SoftBART priors, and the estimated propensity score is included as a covariate in the prognostic function. Our theory in \cref{sec:SoftBART} covers SoftBART, but not BART or SoftBCF, which are included as empirical comparisons. Moreover, our implementation departs from the conditions of \cref{cor:SoftBART-bvm} in that the error variance $\sigma^2$ is estimated rather than known and the outcome regression is not clipped. We view these as technical conveniences of the theory and do not expect them to affect the results in practice. Plug-in and one-step corrected posterior samples of $\psi_{P_0, \delta}$ are obtained according to Algorithms~\ref{alg:plug-in} and~\ref{alg:onestep}, respectively. For point-wise intervals we simply take the central $100(1-\alpha)\%$ credible interval for a given $\delta$ value. While \cref{thm:bvm_single_time} is only argued pointwise in $\delta$, an extension to the uniform setting may be possible and is considered for this simulation study in the Supplementary Material.

We compare the performance of these Bayesian estimators with a frequentist counterpart. The frequentist estimators correspond to those defined in \citet{Kennedy2018ipsi}, with sample-splitting used to control the entropy of the nuisance estimators \footnote{In the frequentist setting, nuisance functions were estimated using Super Learner \citep{vanderlaan2007superlearner}, combining generalized additive models, multivariate adaptive regression splines, support vector machines, random forests, and parametric models (with and without interactions, model selection through stepwise AIC).}. Frequentist uncertainty quantification utilizes Wald-type asymptotic confidence intervals guaranteed by asymptotic linearity.

Estimator performance is assessed using several metrics, including integrated bias, root-mean-squared error, uniform coverage, and interval length. The target function $\psi_{P_0, \delta}$ was approximated using Monte Carlo integration with $n = 10^6$ samples. 
We use $J = 500$ simulations and $I = 100$ values of $\delta$, log-uniformly spaced between $\exp(-2.3)$ and $\exp(2.3)$. The simulation is then completed across the following sample sizes: $n = 500, 1000, 5000$. Additional implementation details are given in the Supplementary Material.

\begin{table}
\small
\caption{IPSI regular data comparison of plug-in and EIF estimators across methods and sample sizes, with best in class performance being bold-faced.}
\label{tab:reg_table}
\centering
\begin{tabular}[t]{llcccccccc}
\toprule
\multicolumn{2}{c}{ } & \multicolumn{4}{c}{Plug-in} & \multicolumn{4}{c}{EIF} \\
\cmidrule(l{3pt}r{3pt}){3-6} \cmidrule(l{3pt}r{3pt}){7-10}
$n$ & Method & Bias & RMSE & Cov & Int Len & Bias & RMSE & Cov & Int Len\\
\midrule
 & Frequentist & 0.206 & 24.343 & 0.840 & \textbf{3.298} & 0.098 & 24.547 & 0.955 & 4.359\\

 & BART & \textbf{0.187} & \textbf{23.593} & \textbf{0.950} & 4.151 & 0.080 & 24.203 & \textbf{0.961} & 4.413\\

 & SoftBART & 0.382 & 24.639 & 0.936 & 4.142 & \textbf{0.023} & 23.798 & 0.958 & 4.270\\

\multirow{-4}{*}{\raggedright\arraybackslash 500} & SoftBCF & 0.427 & 25.047 & 0.927 & 4.134 & 0.036 & \textbf{23.777} & 0.957 & \textbf{4.262}\\
\cmidrule{1-10}
 & Frequentist & 0.163 & 25.970 & 0.824 & \textbf{2.342} & 0.062 & 25.715 & 0.951 & 3.067\\

 & BART & \textbf{0.133} & \textbf{25.104} & \textbf{0.944} & 2.968 & 0.062 & 25.706 & \textbf{0.954} & 3.112\\

 & SoftBART & 0.287 & 26.364 & 0.923 & 2.964 & \textbf{0.033} & 25.403 & 0.948 & \textbf{3.021}\\

\multirow{-4}{*}{\raggedright\arraybackslash 1000} & SoftBCF & 0.303 & 26.555 & 0.919 & 2.962 & 0.037 & \textbf{25.401} & 0.948 & \textbf{3.021}\\
\cmidrule{1-10}
 & Frequentist & 0.103 & 24.914 & 0.837 & \textbf{1.049} & 0.049 & 24.637 & 0.951 & 1.369\\

 & BART & \textbf{0.072} & \textbf{24.576} & \textbf{0.943} & 1.346 & \textbf{0.033} & 24.565 & \textbf{0.954} & 1.384\\

 & SoftBART & 0.153 & 26.340 & 0.923 & 1.355 & 0.047 & 24.506 & 0.950 & \textbf{1.361}\\

\multirow{-4}{*}{\raggedright\arraybackslash 5000} & SoftBCF & 0.155 & 26.399 & 0.924 & 1.353 & 0.047 & \textbf{24.498} & 0.950 & \textbf{1.361}\\
\bottomrule
\end{tabular}
\end{table}

\begin{table}
\small
\caption{IPSI transformed data comparison of plug-in and EIF estimators across methods and sample sizes, with best in class performance being bold-faced.}
\label{tab:trans_table}
\centering
\begin{tabular}[t]{llcccccccc}
\toprule
\multicolumn{2}{c}{ } & \multicolumn{4}{c}{Plug-in} & \multicolumn{4}{c}{EIF} \\
\cmidrule(l{3pt}r{3pt}){3-6} \cmidrule(l{3pt}r{3pt}){7-10}
$n$ & Method & Bias & RMSE & Cov & Int Len & Bias & RMSE & Cov & Int Len\\
\midrule
 & Frequentist & 0.568 & 28.112 & 0.703 & \textbf{2.869} & 0.150 & 25.317 & 0.946 & 4.448\\

 & BART & \textbf{0.519} & 26.570 & 0.896 & 4.057 & 0.140 & 24.427 & \textbf{0.959} & 4.360\\

 & SoftBART & 0.587 & \textbf{26.521} & \textbf{0.907} & 4.157 & \textbf{0.041} & 23.902 & 0.954 & 4.262\\

\multirow{-4}{*}{\raggedright\arraybackslash 500} & SoftBCF & 0.662 & 27.716 & 0.885 & 4.134 & 0.047 & \textbf{23.887} & 0.953 & \textbf{4.248}\\
\cmidrule{1-10}
 & Frequentist & 0.565 & 32.028 & 0.633 & \textbf{2.026} & 0.178 & 26.673 & 0.931 & 3.125\\

 & BART & \textbf{0.499} & 30.304 & 0.835 & 2.905 & 0.082 & 25.804 & \textbf{0.950} & 3.078\\

 & SoftBART & 0.506 & \textbf{29.494} & \textbf{0.866} & 2.982 & \textbf{0.031} & 25.444 & 0.948 & 3.021\\

\multirow{-4}{*}{\raggedright\arraybackslash 1000} & SoftBCF & 0.556 & 30.299 & 0.846 & 2.969 & 0.034 & \textbf{25.404} & 0.948 & \textbf{3.014}\\
\cmidrule{1-10}
 & Frequentist & 0.498 & 44.317 & 0.404 & \textbf{0.922} & 0.119 & 26.216 & 0.925 & 1.392\\

 & BART & 0.286 & 32.907 & 0.812 & 1.340 & 0.053 & 24.773 & 0.951 & 1.377\\

 & SoftBART & 0.262 & 30.819 & \textbf{0.878} & 1.393 & 0.041 & 24.502 & 0.952 & 1.368\\

\multirow{-4}{*}{\raggedright\arraybackslash 5000} & SoftBCF & \textbf{0.260} & \textbf{30.744} & 0.875 & 1.380 & \textbf{0.036} & \textbf{24.476} & \textbf{0.953} & \textbf{1.367}\\
\bottomrule
\end{tabular}
\end{table}

The results of the simulation study for the regular (i.e., untransformed covariate) case are shown in \cref{tab:reg_table}, while \cref{tab:trans_table} contains results when the models are fitted with transformed covariates. We find that the estimators utilizing the one-step correction substantially improve bias, RMSE, and coverage, at the expense of an increased interval length. Furthermore, Bayesian methods outperform frequentist alternatives in terms of bias and coverage in both the plug-in and EIF cases. This is particularly apparent in \cref{tab:trans_table}, where the covariate transformation results in nuisance functions that are more difficult to learn. It is important to note that the incremental propensity score intervention is a $\pi$-dependent stochastic intervention. 
In this setting, the efficient influence function is especially useful for bias correction, as the intervention $Q_P$ depends directly on the estimated propensity score. In the supplement, we consider a simulation study with $\pi$-fixed stochastic interventions. We find that the one-step correction is still necessary for asymptotic guarantees, but if the response surface $\mu(x,a)$ can be estimated well, the EIF-based correction may not provide a notable improvement in finite-sample performance. Additional details can be found in the Supplementary Material, and reproducible code for the simulation study is available on GitHub\footnote{\url{https://github.com/tyler3schmidt/Stochastic_Interventions}}.


\section{An Application to Statin Prescribing and LDL Cholesterol} \label{sec:Real_data}

To illustrate the utility of Bayesian inference for stochastic intervention effects, we present an application in public health. Heart disease is the leading cause of death in the United States, and statins are a common treatment used to prevent cardiovascular disease \citep{law2003statins}. The goal of this medication is to lower LDL cholesterol, which is strongly linked with cardiovascular events. In \citet{collins2016statin}, they find lowering LDL cholesterol by 77 mg/dL for 5 years prevents 1,000 cases of major cardiovascular events in a population of 10,000 people. Although many clinical trials have established the efficacy of this class of drugs, policy makers are often interested in understanding the population-level effects of changes in prescribing behavior. 

In this case a researcher may first consider a deterministic intervention such as the average treatment effect. However, serious positivity violations occur in practice. For instance, patients that are young with no serious medical issues have virtually no chance of being prescribed statins. Therefore in this setting an incremental propensity score intervention may be more realistic and appropriate. We utilize the National Health and Nutrition Examination Survey (NHANES) 2017-2018 cross-sectional study to obtain data on this problem. Originally consisting of 9,254 samples, missing data limited the final sample size to $n =$ 2,340.
Covariates were chosen based upon medical guidelines for prescribing statins for the treatment of elevated LDL cholesterol \citep{stone2014accaha, uspstf2022statins}. These covariates include age, sex, race, diabetes, hypertension, smoking, and BMI. Then, \cref{alg:onestep} is used to obtain estimates of $\psi_{P_0, \delta}$, the average LDL cholesterol level under an IPSI, with independent SoftBART priors chosen for $\mu$ and $\pi$. The resulting effect curve is shown in \cref{fig:pc-shift-plot}.

\begin{figure}[!ht]
\includegraphics[width=0.70\textwidth]{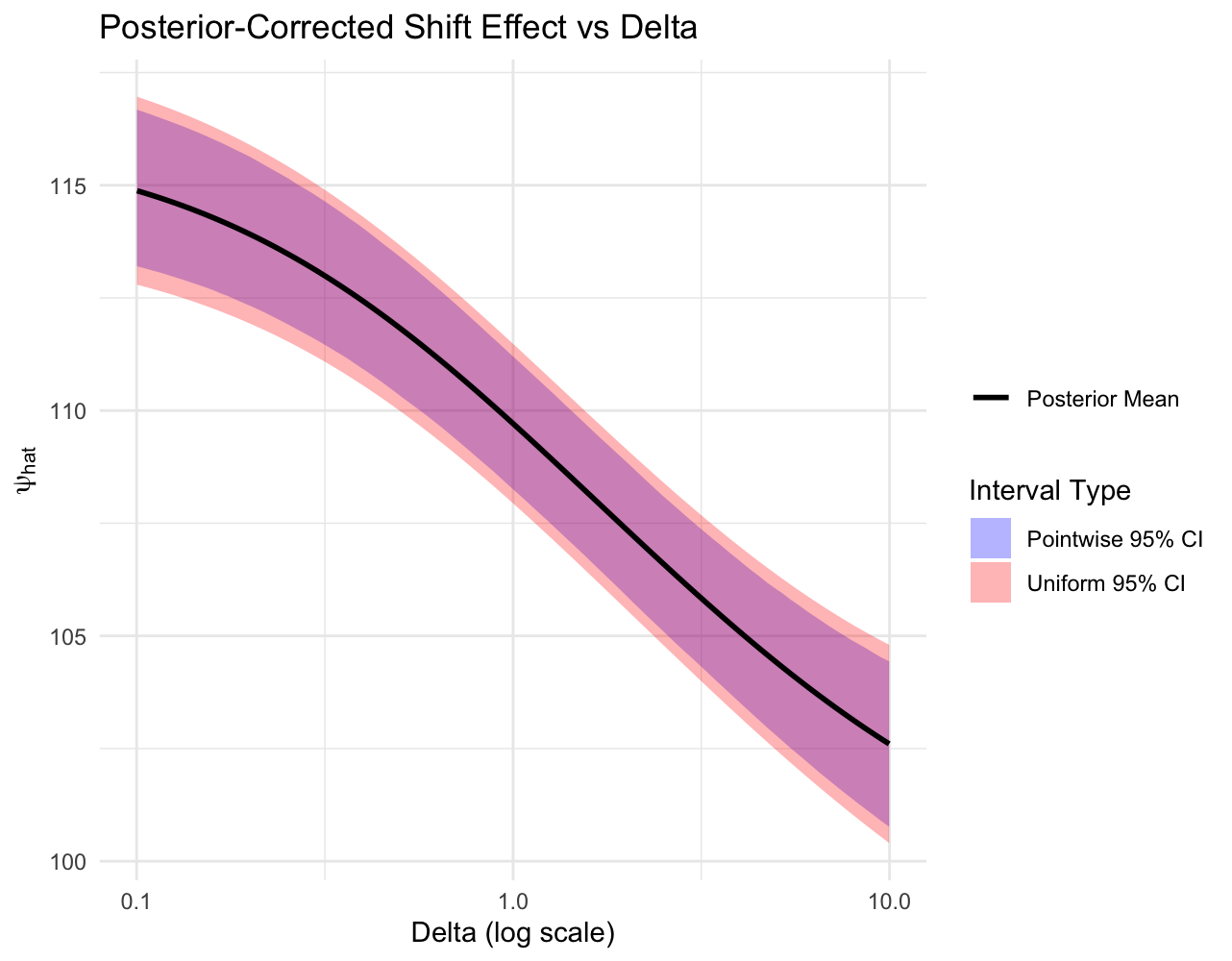} 
\caption{Posterior-corrected IPSI estimates of the effect of a hypothetical shift in statin treatment probability on LDL levels. The x-axis shows the magnitude of the shift in odds of receiving a statin ($\delta$, log scale), and the y-axis shows the corresponding estimated change in LDL ($\hat{\psi}$). The solid line represents the posterior-corrected mean estimate, while the shaded areas indicate pointwise (blue) and uniform (red) 95\% credible intervals.} \label{fig:pc-shift-plot}
\end{figure}

From \cref{fig:pc-shift-plot}, we observe an approximately linear and monotonic relationship between increases in an individual's conditional likelihood of receiving statin therapy and decreases in expected LDL cholesterol, suggesting that progressively greater shifts toward treatment correspond to progressively lower population-average LDL levels. Moving from policies that substantially reduce treatment odds to those that increase treatment odds corresponds to an estimated reduction in expected LDL of approximately 13 mg/dL. This effect is smaller than reductions reported in meta-analyses of statin clinical trials, where treatment commonly lowers LDL cholesterol by approximately 39 mg/dL (1 mmol/L) on average \citep{ctt2010ldl}. However, it can be argued that this difference is expected, as our intervention reflects realistic changes in prescribing behavior in a general adult population, whereas clinical trials typically study selected higher-risk individuals under deterministic treatment assignment.

Although the relationship between statins and LDL cholesterol is well established, our analysis demonstrates how a realistic policy change by decision makers will alter the population-level expected LDL. Furthermore, the use of Bayesian methods allows us to easily obtain credible intervals for $\psi_{P_0, \delta}$, which both have a direct probabilistic interpretation and are asymptotically equivalent to frequentist confidence intervals under regularity conditions. One such condition, strong positivity, may be strained here. Many estimated propensity scores are near zero (see the Supplementary Material), so the IPSI remains identified under weak positivity, but the frequentist guarantees of the credible intervals should be interpreted with this in mind. Additional details on this analysis are included in the Supplementary Material, with diagnostic plots provided to assess model fit and the validity of the modeling assumptions.

\section{Discussion} \label{sec:Discussion}

In this paper, we develop a calibrated Bayesian framework for causal inference with stochastic intervention effects. Building on the one-step posterior correction of \citet{Yiu2025pcorrect}, we establish sufficient conditions under which the corrected stochastic intervention posterior satisfies a semiparametric Bernstein--von Mises theorem, so that credible intervals retain a Bayesian interpretation while being asymptotically equivalent to efficient frequentist confidence intervals. Our results accommodate both fixed and propensity-score-dependent interventions, including the incremental propensity score intervention \citep{Kennedy2018ipsi} and the power tilt intervention introduced here. A central Bayesian contribution is our theoretical analysis of SoftBART: we derive a new sup-norm covering bound and combine it with a random-design posterior contraction result to verify the empirical process and nuisance-rate conditions required by the Bernstein-von Mises theorem under explicit smoothness and sieve conditions. This connects an abstract semiparametric posterior correction to a flexible and widely applicable Bayesian nonparametric prior, while also providing a template for studying other priors and functionals. In simulations, the corrected posterior generally reduced bias and improved coverage relative to the plug-in posterior, particularly when the nuisance functions were difficult to estimate. The NHANES analysis further illustrates how the method can produce policy-relevant intervention curves, finding that increases in the odds of receiving statin therapy were associated with lower expected LDL cholesterol.

Several limitations delineate the scope of the present results and motivate future work. First, \cref{thm:bvm_single_time} is pointwise in the intervention index $\delta$. Although we investigate uniform credible bands empirically, a functional Bernstein-von Mises theorem that justifies simultaneous inference over $\{Q_\delta\}_{\delta \in \mathcal{D}}$ remains to be established. Second, the applied analysis uses complete records, and the current implementation does not automatically propagate uncertainty from a separate missing covariate imputation procedure through the Bayesian bootstrap correction \citep{linero2023bnpcause}. This motivates joint models for missingness and the nuisance functions. Further, \cref{thm:bvm_single_time} does not require strong positivity for interventions such as the IPSI, but our results for SoftBART do, and future work could potentially relax this. Broader extensions include considerations of continuous treatments \citep{schindl2026incremental} and stochastic policies for Markov decision processes \citep{kallus2024natural}.

\begin{acks}[Acknowledgments]
The authors would like to thank Sanvesh Srivastava for helpful discussions about this work.
\end{acks}



\begin{supplement}
\stitle{Supplementary Material for ``Calibrated Bayesian Inference for Stochastic Intervention Effects''}
\sdescription{The Supplementary Material extends the results of this paper to the time-varying treatment setting; it contains proofs of all results in the main text, as well as additional details for the simulation and applied studies.}

\end{supplement}

\bibliographystyle{ba}
\bibliography{refs}

@inproceedings{alaa2017itegp,
  author    = {Alaa, Ahmed M. and van der Schaar, Mihaela},
  title     = {{Bayesian} inference of individualized treatment effects using multi-task {Gaussian} processes},
  booktitle = {Proceedings of the 31st International Conference on Neural Information Processing Systems},
  pages     = {3427--3435},
  year      = {2017},
  publisher = {Curran Associates Inc.},
  address   = {Red Hook, NY, USA},
  location  = {Long Beach, California, USA},
  series    = {NIPS'17}
}

@article{castillo2015bvm,
  author  = {Castillo, Isma{\"e}l and Rousseau, Judith},
  title   = {A {Bernstein--von Mises} Theorem for Smooth Functionals in Semiparametric Models},
  journal = {The Annals of Statistics},
  volume  = {43},
  number  = {6},
  pages   = {2353--2383},
  year    = {2015},
  month   = dec,
  doi     = {10.1214/15-AOS1336},
}

@article{Chen2025,
  author  = {Chen, Xinyuan and Hu, Liangyuan and Li, Fan},
  title   = {A flexible {Bayesian} g-formula for causal survival analyses with time-dependent confounding},
  journal = {Lifetime Data Analysis},
  year    = {2025},
  volume  = {31},
  pages   = {394--421},
  doi     = {10.1007/s10985-025-09652-3}
}

@article{chipman2010bart,
  author  = {Chipman, Hugh A. and George, Edward I. and McCulloch, Robert E.},
  title   = {BART: Bayesian Additive Regression Trees},
  journal = {The Annals of Applied Statistics},
  volume  = {4},
  number  = {1},
  pages   = {266--298},
  year    = {2010},
  month   = mar,
  doi     = {10.1214/09-AOAS285},
  url     = {https://doi.org/10.1214/09-AOAS285}
}

@article{collins2016statin,
  author  = {Collins, Rory and Reith, Christina and Emberson, Jonathan and Armitage, Jane and Baigent, Colin and Blackwell, Lisa and Blumenthal, Roger and Danesh, John and Smith, George Davey and DeMets, David and Evans, Stephen and Law, Malcolm and MacMahon, Stephen and Martin, Sonal and Neal, Bruce and Poulter, Neil and Sandercock, Peter and Schulz, Kenneth and Sever, Peter and Simes, John and Smeeth, Liam and Wald, Nicholas and Yusuf, Salim and Peto, Richard},
  title   = {Interpretation of the evidence for the efficacy and safety of statin therapy},
  journal = {The Lancet},
  volume  = {388},
  number  = {10059},
  pages   = {2532--2561},
  year    = {2016},
  month   = nov,
  doi     = {10.1016/S0140-6736(16)31357-5},
  url     = {https://doi.org/10.1016/S0140-6736(16)31357-5}
}

@article{Comment2022,
  author  = {Comment, L. and Coull, B. A. and Zigler, C. and Valeri, L.},
  title   = {{Bayesian} data fusion: Probabilistic sensitivity analysis for unmeasured confounding using informative priors based on secondary data},
  journal = {Biometrics},
  year    = {2022},
  volume  = {78},
  number  = {2},
  pages   = {730--741},
  month   = jun,
  doi     = {10.1111/biom.13436},
  pmid    = {33527348},
  pmcid   = {PMC8326294},
  note    = {Epub 2021 Feb 16}
}

@article{ctt2010ldl,
  author  = {{Cholesterol Treatment Trialists' (CTT) Collaboration} and Baigent, Colin and Blackwell, Lisa and Emberson, Jonathan and Holland, Louise E. and Reith, Christine and Bhala, Neeraj and Peto, Richard and Barnes, Elisabeth H. and Keech, Anthony and Simes, John and Collins, Rory},
  title   = {Efficacy and Safety of More Intensive Lowering of {LDL} Cholesterol: A Meta-Analysis of Data from 170,000 Participants in 26 Randomised Trials},
  journal = {The Lancet},
  volume  = {376},
  number  = {9753},
  pages   = {1670--1681},
  year    = {2010},
  month   = nov,
  doi     = {10.1016/S0140-6736(10)61350-5},
  url     = {https://doi.org/10.1016/S0140-6736(10)61350-5}
}

@book{daniels2023bnp,
  author    = {Daniels, Michael J. and Linero, Antonio and Roy, Jason},
  title     = {{Bayesian} Nonparametrics for Causal Inference and Missing Data},
  publisher = {Chapman and Hall/CRC},
  edition   = {1},
  year      = {2023},
  month     = aug,
  doi       = {10.1201/9780429324222},
  url       = {https://doi.org/10.1201/9780429324222}
}

@article{diaconis1986consistency,
  author  = {Diaconis, Persi and Freedman, David},
  title   = {On the Consistency of {Bayes} Estimates},
  journal = {The Annals of Statistics},
  volume  = {14},
  number  = {1},
  pages   = {1--26},
  year    = {1986},
  month   = mar,
  doi     = {10.1214/aos/1176349830},
  url     = {https://doi.org/10.1214/aos/1176349830}
}

@article{diaz2020causal,
  author  = {D{\'\i}az, Iv{\'a}n and Hejazi, Nima S.},
  title   = {Causal Mediation Analysis for Stochastic Interventions},
  journal = {Journal of the Royal Statistical Society: Series B (Statistical Methodology)},
  volume  = {82},
  number  = {3},
  pages   = {661--683},
  year    = {2020},
  month   = jul,
  doi     = {10.1111/rssb.12362},
  url     = {https://doi.org/10.1111/rssb.12362}
}

@inproceedings{didelez2006direct,
  author    = {Didelez, Vanessa and Dawid, A. Philip and Geneletti, Sara},
  title     = {Direct and Indirect Effects of Sequential Treatments},
  booktitle = {Proceedings of the Twenty-Second Conference on Uncertainty in Artificial Intelligence},
  series    = {UAI'06},
  pages     = {138--146},
  year      = {2006},
  publisher = {AUAI Press},
  address   = {Arlington, Virginia, USA},
  location  = {Cambridge, MA, USA},
  isbn      = {0974903922}
}

@article{dorie2019competition,
  author  = {Dorie, Vincent and Hill, Jennifer and Shalit, Uri and Scott, Marc and Cervone, Dan},
  title   = {Automated versus Do-It-Yourself Methods for Causal Inference: Lessons Learned from a Data Analysis Competition},
  journal = {Statistical Science},
  volume  = {34},
  number  = {1},
  pages   = {43--68},
  year    = {2019},
  month   = feb,
  doi     = {10.1214/18-STS667},
  url     = {https://doi.org/10.1214/18-STS667}
}

@article{egels2025posterior,
  author  = {Egels, Paul and Castillo, Isma{\"e}l},
  title   = {Posterior and Variational Inference for Deep Neural Networks with Heavy-Tailed Weights},
  journal = {Journal of Machine Learning Research},
  year    = {2025},
  volume  = {26},
  number  = {122},
  pages   = {1--58},
  url     = {http://jmlr.org/papers/v26/24-0894.html}
}

@article{friedman2001gbm,
  author  = {Friedman, Jerome H.},
  title   = {Greedy Function Approximation: A Gradient Boosting Machine},
  journal = {The Annals of Statistics},
  volume  = {29},
  number  = {5},
  pages   = {1189--1232},
  year    = {2001},
  month   = oct,
  doi     = {10.1214/aos/1013203451},
  url     = {https://doi.org/10.1214/aos/1013203451}
}

@article{ghosal2007dirichlet,
  author  = {Ghosal, Subhashis and van der Vaart, Aad},
  title   = {Posterior convergence rates of {Dirichlet} mixtures at smooth densities},
  journal = {The Annals of Statistics},
  volume  = {35},
  number  = {2},
  pages   = {697--723},
  year    = {2007},
  month   = apr,
  doi     = {10.1214/009053606000001271},
  url     = {https://doi.org/10.1214/009053606000001271}
}

@article{ghosal2007noniid,
  author  = {Ghosal, Subhashis and van der Vaart, Aad W.},
  title   = {Convergence Rates of Posterior Distributions for Noniid Observations},
  journal = {The Annals of Statistics},
  volume  = {35},
  number  = {1},
  pages   = {192--223},
  year    = {2007},
  month   = feb,
  doi     = {10.1214/009053606000001172},
  url     = {https://doi.org/10.1214/009053606000001172}
}

@book{ghosal2017fundamentals,
  author    = {Ghosal, Subhashis and van der Vaart, Aad W.},
  title     = {Fundamentals of Nonparametric Bayesian Inference},
  publisher = {Cambridge University Press},
  year      = {2017}
}

@article{hahn2018regularization,
  author  = {Hahn, P. Richard and Carvalho, Carlos M. and Puelz, David and He, Jingyu},
  title   = {Regularization and Confounding in Linear Regression for Treatment Effect Estimation},
  journal = {Bayesian Analysis},
  volume  = {13},
  number  = {1},
  pages   = {163--182},
  year    = {2018},
  month   = mar,
  doi     = {10.1214/16-BA1044},
  url     = {https://doi.org/10.1214/16-BA1044}
}

@article{hahn2020bcf,
  author  = {Hahn, P. Richard and Murray, Jared S. and Carvalho, Carlos M.},
  title   = {Bayesian Regression Tree Models for Causal Inference: Regularization, Confounding, and Heterogeneous Effects (with Discussion)},
  journal = {Bayesian Analysis},
  volume  = {15},
  number  = {3},
  pages   = {965--1056},
  year    = {2020},
  month   = sep,
  doi     = {10.1214/19-BA1195},
  url     = {https://doi.org/10.1214/19-BA1195}
}

@article{haneuse2013modified,
  author  = {Haneuse, Sebastian and Rotnitzky, Andrea},
  title   = {Estimation of the Effect of Interventions That Modify the Received Treatment},
  journal = {Statistics in Medicine},
  volume  = {32},
  number  = {30},
  pages   = {5260--5277},
  year    = {2013},
  month   = dec,
  doi     = {10.1002/sim.5907},
  url     = {https://doi.org/10.1002/sim.5907}
}

@article{hejazi2023mediation,
  author  = {Hejazi, Nima S. and Rudolph, Katie E. and Van Der Laan, Mark J. and D{\'\i}az, Ismael},
  title   = {Nonparametric Causal Mediation Analysis for Stochastic Interventional (In)Direct Effects},
  journal = {Biostatistics},
  volume  = {24},
  number  = {3},
  pages   = {686--707},
  year    = {2023},
  month   = jul,
  doi     = {10.1093/biostatistics/kxac002},
  url     = {https://doi.org/10.1093/biostatistics/kxac002}
}

@article{hill2011bnp,
  author  = {Hill, Jennifer L.},
  title   = {Bayesian Nonparametric Modeling for Causal Inference},
  journal = {Journal of Computational and Graphical Statistics},
  volume  = {20},
  number  = {1},
  pages   = {217--240},
  year    = {2011},
  doi     = {10.1198/jcgs.2010.08162},
  url     = {https://doi.org/10.1198/jcgs.2010.08162}
}

@article{Huang2026,
    title={Multivariate incremental effects for continuous treatments: Studying the health effects of environmental mixtures}, 
    author={Zhuochao Huang and Kejin Dong and Tuo Lin and Joseph Antonelli},
    year={2026},
    journal = {arXiv preprint arXiv:2604.23534},
    doi     = {10.48550/arXiv.2604.23534},
    url     = {https://doi.org/10.48550/arXiv.2604.23534}
}

@article{Josefsson2026,
    author = {Josefsson, Maria and Karalija, Nina and Daniels, Michael J},
    title = {Long-term memory effects of an incremental blood pressure intervention in a mortal cohort},
    journal = {Biometrics},
    volume = {82},
    number = {1},
    pages = {ujaf176},
    year = {2026},
    month = {03},
    issn = {0006-341X},
    doi = {10.1093/biomtc/ujaf176},
    url = {https://doi.org/10.1093/biomtc/ujaf176}
}

@article{kallus2024natural,
  author  = {Kallus, Nathan and Uehara, Masatoshi},
  title   = {Efficient Evaluation of Natural Stochastic Policies in Off-Line Reinforcement Learning},
  journal = {Biometrika},
  volume  = {111},
  number  = {1},
  pages   = {51--69},
  year    = {2024},
  month   = mar,
  doi     = {10.1093/biomet/asad059},
  url     = {https://doi.org/10.1093/biomet/asad059}
}

@article{Kang2007dr,
  author  = {Kang, Joseph D. Y. and Schafer, Joseph L.},
  title   = {Demystifying Double Robustness: A Comparison of Alternative Strategies for Estimating a Population Mean from Incomplete Data},
  journal = {Statistical Science},
  volume  = {22},
  number  = {4},
  pages   = {523--539},
  year    = {2007},
  doi     = {10.1214/07-STS227},
  url     = {https://doi.org/10.1214/07-STS227}
}

@article{Kennedy2018ipsi,
  author  = {Kennedy, Edward H.},
  title   = {Nonparametric Causal Effects Based on Incremental Propensity Score Interventions},
  journal = {Journal of the American Statistical Association},
  volume  = {114},
  number  = {526},
  pages   = {645--656},
  year    = {2019},
  doi     = {10.1080/01621459.2017.1422737},
  url     = {https://doi.org/10.1080/01621459.2017.1422737}
}

@incollection{kennedy2024semiparametric,
  author       = {Kennedy, Edward H.},
  title        = {Semiparametric Doubly Robust Targeted Double Machine Learning: A Review},
  booktitle    = {Handbook of Statistical Methods for Precision Medicine},
  editor       = {Laber, Eric and Chakraborty, Bibhas and Moodie, Erica E. M. and Cai, Tianxi and van der Laan, Mark},
  publisher    = {Chapman \& Hall/CRC},
  year         = {2024},
  pages        = {207--235},  
  isbn         = {9781032106151},
}

@article{Keil2018,
  author  = {Keil, Alexander P. and Daza, Eric J. and Engel, Stephanie M. and Buckley, Jessie P. and Edwards, Jessie K.},
  title   = {A {Bayesian} approach to the g-formula},
  journal = {Statistical Methods in Medical Research},
  year    = {2018},
  volume  = {27},
  number  = {10},
  pages   = {3183--3204},
  doi     = {10.1177/0962280217694665}
}

@article{law2003statins,
  author  = {Law, M. R. and Wald, N. J. and Rudnicka, A. R.},
  title   = {Quantifying effect of statins on low density lipoprotein cholesterol, ischaemic heart disease, and stroke: systematic review and meta-analysis},
  journal = {BMJ},
  volume  = {326},
  number  = {7404},
  pages   = {1423},
  year    = {2003},
  month   = jun,
  doi     = {10.1136/bmj.326.7404.1423},
  pmid    = {12829554},
  pmcid   = {PMC162260},
  url     = {https://doi.org/10.1136/bmj.326.7404.1423}
}

@article{Linero2018brte,
  author    = {Linero, Antonio R. and Yang, Yun},
  title     = {Bayesian Regression Tree Ensembles that Adapt to Smoothness and Sparsity},
  journal   = {Journal of the Royal Statistical Society: Series B (Statistical Methodology)},
  volume    = {80},
  number    = {5},
  pages     = {1087--1110},
  year      = {2018},
  month     = nov,
  doi       = {10.1111/rssb.12293},
  url       = {https://doi.org/10.1111/rssb.12293}
}

@article{linero2022softbart,
  author  = {Linero, Antonio R.},
  title   = {{SoftBart}: {Soft Bayesian} Additive Regression Trees},
  journal = {arXiv preprint arXiv:2210.16375},
  year    = {2022},
  month   = oct,
  doi     = {10.48550/arXiv.2210.16375},
  url     = {https://arxiv.org/abs/2210.16375}
}

@article{linero2024,
    author = {Antonio R. Linero},
    title = {In Nonparametric and High-Dimensional Models, Bayesian Ignorability is an Informative Prior},
    journal = {Journal of the American Statistical Association},
    volume = {119},
    number = {548},
    pages = {2785--2798},
    year = {2024},
    publisher = {Taylor \& Francis},
    doi = {10.1080/01621459.2023.2278202}
}

@article{linero2023bnpcause,
  author  = {Linero, Antonio R. and Antonelli, Joseph L.},
  title   = {The How and Why of Bayesian Nonparametric Causal Inference},
  journal = {WIREs Computational Statistics},
  volume  = {15},
  number  = {1},
  pages   = {e1583},
  year    = {2023},
  doi     = {10.1002/wics.1583},
  url     = {https://doi.org/10.1002/wics.1583}
}

@article{little2011calibrated,
  author  = {Little, Roderick J. A.},
  title   = {Calibrated Bayes, for Statistics in General, and Missing Data in Particular},
  journal = {Statistical Science},
  volume  = {26},
  number  = {2},
  pages   = {162--174},
  year    = {2011},
  month   = may,
  doi     = {10.1214/10-STS318},
  url     = {https://doi.org/10.1214/10-STS318}
}

@article{li2023bayesian,
  author  = {Li, Fan and Ding, Peng and Mealli, Fabrizia},
  title   = {Bayesian Causal Inference: A Critical Review},
  journal = {Philosophical Transactions of the Royal Society A: Mathematical, Physical and Engineering Sciences},
  volume  = {381},
  number  = {2247},
  pages   = {20220153},
  year    = {2023},
  month   = may,
  doi     = {10.1098/rsta.2022.0153},
  url     = {https://doi.org/10.1098/rsta.2022.0153}
}

@article{Diaz2012,
  author  = {D{\'i}az, Iv{\'a}n and van der Laan, Mark J.},
  title   = {Population Intervention Causal Effects Based on Stochastic Interventions},
  journal = {Biometrics},
  volume  = {68},
  number  = {2},
  pages   = {541--549},
  year    = {2012},
  month   = jun,
  doi     = {10.1111/j.1541-0420.2011.01685.x},
  url     = {https://doi.org/10.1111/j.1541-0420.2011.01685.x}
}

@article{Diaz2013,
  author  = {D{\'i}az, Iv{\'a}n and van der Laan, Mark J.},
  title   = {Assessing the Causal Effect of Policies: An Example Using Stochastic Interventions},
  journal = {The International Journal of Biostatistics},
  volume  = {9},
  number  = {2},
  pages   = {161--174},
  year    = {2013}
}

@techreport{newey1998,
  author      = {Newey, Whitney K. and Hsieh, Fushing and Robins, James},
  title       = {Undersmoothing and Bias Corrected Functional Estimation},
  institution = {Massachusetts Institute of Technology, Department of Economics},
  type        = {Working Paper},
  number      = {98-17},
  year        = {1998},
  month       = oct,
  url         = {https://ideas.repec.org/p/mit/worpap/98-17.html},
  note        = {First version October 1991; revised July 1998}
}

@article{oganisian2021bayesian,
  author  = {Oganisian, Alexis and Roy, Jon A.},
  title   = {A Practical Introduction to Bayesian Estimation of Causal Effects: Parametric and Nonparametric Approaches},
  journal = {Statistics in Medicine},
  volume  = {40},
  number  = {2},
  pages   = {518--551},
  year    = {2021},
  month   = jan,
  doi     = {10.1002/sim.8761},
  url     = {https://doi.org/10.1002/sim.8761}
}

@article{Ogburn2024,
  author  = {Ogburn, Elizabeth L. and Sofrygin, Oleg and D{\'i}az, Iv{\'a}n and van der Laan, Mark J.},
  title   = {Causal Inference for Social Network Data},
  journal = {Journal of the American Statistical Association},
  volume  = {119},
  number  = {545},
  pages   = {597--611},
  year    = {2024},
  doi     = {10.1080/01621459.2022.2131557}
}

@article{Papadogeorgou2019,
author = {Papadogeorgou, Georgia and Mealli, Fabrizia and Zigler, Corwin M.},
title = {Causal inference with interfering units for cluster and population level treatment allocation programs},
journal = {Biometrics},
volume = {75},
number = {3},
pages = {778-787},
doi = {10.1111/biom.13049},
url = {https://onlinelibrary.wiley.com/doi/abs/10.1111/biom.13049},
year = {2019}
}

@article{Papadogeorgou2022,
author = {Papadogeorgou, Georgia and Imai, Kosuke and Lyall, Jason and Li, Fan},
title = {Causal inference with spatio-temporal data: Estimating the effects of airstrikes on insurgent violence in Iraq},
journal = {Journal of the Royal Statistical Society: Series B (Statistical Methodology)},
volume = {84},
number = {5},
pages = {1969-1999},
doi = {10.1111/rssb.12548},
url = {https://rss.onlinelibrary.wiley.com/doi/abs/10.1111/rssb.12548},
year = {2022}
}

@book{pfanzagl1982,
  author    = {Pfanzagl, Johann},
  title     = {Contributions to a General Asymptotic Statistical Theory},
  series    = {Lecture Notes in Statistics},
  volume    = {13},
  publisher = {Springer},
  address   = {New York, NY},
  year      = {1982},
  edition   = {1},
  pages     = {315},
  isbn      = {978-0-387-90776-5},
  doi       = {10.1007/978-1-4612-5769-1}
}

@inproceedings{ray2019debiased,
  author    = {Ray, Kolyan and Szabo, Botond},
  title     = {Debiased Bayesian Inference for Average Treatment Effects},
  booktitle = {Advances in Neural Information Processing Systems},
  volume    = {32},
  year      = {2019},
  publisher = {Curran Associates, Inc.},
  url       = {https://proceedings.neurips.cc/paper_files/paper/2019/file/342285bb2a8cadef22f667eeb6a63732-Paper.pdf}
}

@article{ray2020semiparametric,
  author  = {Ray, Kolyan and van der Vaart, Aad W.},
  title   = {Semiparametric Bayesian Causal Inference},
  journal = {The Annals of Statistics},
  volume  = {48},
  number  = {5},
  pages   = {2999--3020},
  year    = {2020},
  month   = oct,
  doi     = {10.1214/19-AOS1919},
  url     = {https://doi.org/10.1214/19-AOS1919}
}

@article{robins1986gformula,
  author  = {Robins, James},
  title   = {A New Approach to Causal Inference in Mortality Studies with a Sustained Exposure Period—Application to Control of the Healthy Worker Survivor Effect},
  journal = {Mathematical Modelling},
  volume  = {7},
  number  = {9--12},
  pages   = {1393--1512},
  year    = {1986},
  doi     = {10.1016/0270-0255(86)90088-6},
  url     = {https://doi.org/10.1016/0270-0255(86)90088-6}
}

@incollection{Robins2004,
  author    = {Robins, James M. and Hern{\'a}n, Miguel A. and Siebert, Uwe},
  title     = {Effects of Multiple Interventions},
  booktitle = {Comparative Quantification of Health Risks: Global and Regional Burden of Disease Attributable to Selected Major Risk Factors},
  editor    = {Ezzati, Majid and Lopez, Alan D. and Rodgers, Anthony and Murray, Christopher J. L.},
  volume    = {1},
  pages     = {2191--2230},
  publisher = {World Health Organization},
  address   = {Geneva},
  year      = {2004}
}

@inproceedings{rockova2019theory,
  author    = {Ro{\v{c}}kov{\'a}, Veronika and Saha, Enakshi},
  title     = {On Theory for {BART}},
  booktitle = {Proceedings of the Twenty-Second International Conference on Artificial Intelligence and Statistics},
  series    = {Proceedings of Machine Learning Research},
  volume    = {89},
  pages     = {2839--2848},
  year      = {2019},
  publisher = {PMLR},
  url       = {http://proceedings.mlr.press/v89/rockova19a.html}
}

@article{rousseau2016frequentist,
  author  = {Rousseau, Judith},
  title   = {On the Frequentist Properties of Bayesian Nonparametric Methods},
  journal = {Annual Review of Statistics and Its Application},
  volume  = {3},
  number  = {1},
  pages   = {211--231},
  year    = {2016},
  doi     = {10.1146/annurev-statistics-041715-033523},
  url     = {https://doi.org/10.1146/annurev-statistics-041715-033523}
}

@article{roy2018bnpmar,
  author  = {Roy, Jason and Lum, Kristian J. and Zeldow, Blair and Dworkin, Jordan D. and Re, Vincent Lo and Daniels, Michael J.},
  title   = {{Bayesian} nonparametric generative models for causal inference with missing at random covariates},
  journal = {Biometrics},
  volume  = {74},
  number  = {4},
  pages   = {1193--1202},
  year    = {2018},
  month   = dec,
  doi     = {10.1111/biom.12875},
  pmid    = {29579341},
  pmcid   = {PMC7568223},
  note    = {Epub 2018 Mar 26}
}

@article{rubin1974causal,
  author  = {Rubin, Donald B.},
  title   = {Estimating Causal Effects of Treatments in Randomized and Nonrandomized Studies},
  journal = {Journal of Educational Psychology},
  volume  = {66},
  number  = {5},
  pages   = {688--701},
  year    = {1974},
  doi     = {10.1037/h0037350},
  url     = {https://doi.org/10.1037/h0037350}
}

@article{rubin1981bayesianbootstrap,
  author  = {Rubin, Donald B.},
  title   = {The Bayesian Bootstrap},
  journal = {The Annals of Statistics},
  volume  = {9},
  number  = {1},
  pages   = {130--134},
  year    = {1981},
  doi     = {10.1214/aos/1176345338},
  url     = {https://doi.org/10.1214/aos/1176345338}
}

@inproceedings{sani2020shift,
  author    = {Sani, Numair and Lee, Jaron and Shpitser, Ilya},
  title     = {Identification and Estimation of Causal Effects Defined by Shift Interventions},
  booktitle = {Proceedings of the 36th Conference on Uncertainty in Artificial Intelligence},
  series    = {Proceedings of Machine Learning Research},
  volume    = {124},
  pages     = {949--958},
  year      = {2020},
  publisher = {PMLR},
  url       = {https://proceedings.mlr.press/v124/sani20a.html}
}

@misc{schindl2026incremental,
  author        = {Schindl, Kyle and Shen, Shuying and Kennedy, Edward H.},
  title         = {Incremental effects for continuous exposures},
  year          = {2026},
  eprint        = {2409.11967},
  archivePrefix = {arXiv},
  primaryClass  = {stat.ME},
  url           = {https://arxiv.org/abs/2409.11967}
}

@article{shen2013dirichlet,
  author  = {Shen, Weining and Tokdar, Surya T. and Ghosal, Subhashis},
  title   = {Adaptive {Bayesian} multivariate density estimation with {Dirichlet} mixtures},
  journal = {Biometrika},
  volume  = {100},
  number  = {3},
  pages   = {623--640},
  year    = {2013},
  month   = sep,
  doi     = {10.1093/biomet/ast015},
  url     = {https://doi.org/10.1093/biomet/ast015}
}

@article{sparapani2021bart,
  author  = {Sparapani, Rodney and Spanbauer, Charles and McCulloch, Robert},
  title   = {Nonparametric Machine Learning and Efficient Computation with {Bayesian} Additive Regression Trees: The {BART} {R} Package},
  journal = {Journal of Statistical Software},
  volume  = {97},
  number  = {1},
  pages   = {1--66},
  year    = {2021},
  month   = mar,
  doi     = {10.18637/jss.v097.i01},
  url     = {https://doi.org/10.18637/jss.v097.i01}
}

@article{stock1989nonparametric,
  author  = {Stock, James H.},
  title   = {Nonparametric Policy Analysis},
  journal = {Journal of the American Statistical Association},
  volume  = {84},
  number  = {406},
  pages   = {567--575},
  year    = {1989},
  doi     = {10.2307/2289944},
  url     = {https://doi.org/10.2307/2289944}
}

@article{stone2014accaha,
  author  = {Stone, Neil J. and Robinson, Jennifer G. and Lichtenstein, Alice H. and Bairey Merz, C. Noel and Blum, Conrad B. and Eckel, Robert H. and Goldberg, Anne C. and Gordon, David and Levy, Daniel and Lloyd-Jones, Donald M. and McBride, Patrick and Schwartz, J. Sanford and Shero, Susan T. and Smith, Sidney C. and Watson, Karol and Wilson, Peter W. F.},
  title   = {2013 {ACC/AHA} Guideline on the Treatment of Blood Cholesterol to Reduce Atherosclerotic Cardiovascular Risk in Adults},
  journal = {Circulation},
  volume  = {129},
  number  = {25\_suppl\_2},
  pages   = {S1--S45},
  year    = {2014},
  month   = jun,
  doi     = {10.1161/01.cir.0000437738.63853.7a},
  url     = {https://www.ahajournals.org/doi/abs/10.1161/01.cir.0000437738.63853.7a}
}

@article{Tang2025,
    author = {Tang, Thai-Son and Liu, Zhihui and Hosni, Ali and Kim, John and Saarela, Olli},
    title = {A marginal structural model for normal tissue complication probability},
    journal = {Biostatistics},
    volume = {26},
    number = {1},
    pages = {kxae019},
    year = {2025},
    month = {01},
    issn = {1465-4644},
    doi = {10.1093/biostatistics/kxae019},
    url = {https://doi.org/10.1093/biostatistics/kxae019}
}

@inproceedings{Tian2008,
  author    = {Tian, Jin},
  title     = {Identifying Dynamic Sequential Plans},
  booktitle = {Proceedings of the 24th Conference on Uncertainty in Artificial Intelligence},
  pages     = {554--561},
  publisher = {Association for Uncertainty in Artificial Intelligence Press},
  year      = {2008}
}

@book{tsiatis2006semiparametric,
  author    = {Tsiatis, Anastasios A.},
  title     = {Semiparametric Theory and Missing Data},
  publisher = {Springer},
  year      = {2006},
  doi       = {10.1007/0-387-37345-4},
  url       = {https://doi.org/10.1007/0-387-37345-4}
}

@article{uspstf2022statins,
  author  = {{US Preventive Services Task Force}},
  title   = {Statin Use for the Primary Prevention of Cardiovascular Disease in Adults: {US} Preventive Services Task Force Recommendation Statement},
  journal = {JAMA},
  volume  = {328},
  number  = {8},
  pages   = {746--753},
  year    = {2022},
  month   = aug,
  doi     = {10.1001/jama.2022.13044},
  url     = {https://doi.org/10.1001/jama.2022.13044}
}

@book{vandervaartwellner2023weak,
  author    = {van der Vaart, Aad W. and Wellner, Jon A.},
  title     = {Weak Convergence and Empirical Processes: With Applications to Statistics},
  series    = {Springer Series in Statistics},
  publisher = {Springer Cham},
  year      = {2023},
  doi       = {10.1007/978-3-031-29040-4},
  url       = {https://doi.org/10.1007/978-3-031-29040-4}
}

@book{vandervaart1998asymptotic,
  author    = {van der Vaart, Aad W.},
  title     = {Asymptotic Statistics},
  publisher = {Cambridge University Press},
  year      = {1998}
}

@article{vandervaart2008gp,
  author  = {van der Vaart, A. W. and van Zanten, J. H.},
  title   = {Rates of contraction of posterior distributions based on {Gaussian} process priors},
  journal = {The Annals of Statistics},
  volume  = {36},
  number  = {3},
  pages   = {1435--1463},
  year    = {2008},
  month   = jun,
  doi     = {10.1214/009053607000000613},
  url     = {https://doi.org/10.1214/009053607000000613}
}

@article{vanderlaan2007superlearner,
  author  = {van der Laan, Mark J. and Polley, Eric C. and Hubbard, Alan E.},
  title   = {Super Learner},
  journal = {Statistical Applications in Genetics and Molecular Biology},
  volume  = {6},
  number  = {1},
  pages   = {Article 25},
  year    = {2007},
  month   = sep,
  doi     = {10.2202/1544-6115.1309},
  url     = {https://doi.org/10.2202/1544-6115.1309}
}

@article{Wu2024,
    author = {Wu, Xiao and Weinberger, Kate R and Wellenius, Gregory A and Dominici, Francesca and Braun, Danielle},
    title = {Assessing the causal effects of a stochastic intervention in time series data: are heat alerts effective in preventing deaths and hospitalizations?},
    journal = {Biostatistics},
    volume = {25},
    number = {1},
    pages = {57-79},
    year = {2024},
    month = {01},
    issn = {1465-4644},
    doi = {10.1093/biostatistics/kxad002},
    url = {https://doi.org/10.1093/biostatistics/kxad002}
}

@article{yang2015minimax,
  author  = {Yang, Yun and Tokdar, Surya T.},
  title   = {Minimax-optimal Nonparametric Regression in High Dimensions},
  journal = {The Annals of Statistics},
  volume  = {43},
  number  = {2},
  pages   = {652--674},
  year    = {2015},
  month   = apr,
  doi     = {10.1214/14-AOS1289},
  url     = {https://doi.org/10.1214/14-AOS1289}
}

@article{Yiu2025pcorrect,
  author  = {Yiu, Andrew and Fong, Edwin and Holmes, Chris and Rousseau, Judith},
  title   = {Semiparametric posterior corrections},
  journal = {Journal of the Royal Statistical Society Series B: Statistical Methodology},
  volume  = {87},
  number  = {4},
  pages   = {1025--1054},
  year    = {2025},
  doi     = {10.1093/jrsssb/qkaf005},
  url     = {https://doi.org/10.1093/jrsssb/qkaf005}
}

@article{Young2011,
  author  = {Young, Jessica G. and Cain, Lauren E. and Robins, James M. and O'Reilly, Eilis J. and Hern{\'a}n, Miguel A.},
  title   = {Comparative Effectiveness of Dynamic Treatment Regimes: An Application of the Parametric G-Formula},
  journal = {Statistics in Biosciences},
  volume  = {3},
  number  = {1},
  pages   = {119--143},
  year    = {2011},
  doi     = {10.1007/s12561-011-9040-7}
}

@article{Young2014,
  author  = {Young, Jessica G. and Hern{\'a}n, Miguel A. and Robins, James M.},
  title   = {Identification, Estimation and Approximation of Risk under Interventions That Depend on the Natural Value of Treatment Using Observational Data},
  journal = {Epidemiologic Methods},
  volume  = {3},
  number  = {1},
  pages   = {1--19},
  year    = {2014}
}

@inproceedings{polson2018posterior,
  title = {Posterior Concentration for Sparse Deep Learning},
  author = {Polson, Nicholas G. and Ro{\v{c}}kov{\'a}, Veronika},
  booktitle = {Advances in Neural Information Processing Systems},
  volume = {31},
  year = {2018},
}

@article{schmidthieber2020nonparametric,
  author = {Schmidt-Hieber, Johannes},
  title = {Nonparametric regression using deep neural networks with {ReLU} activation function},
  journal = {The Annals of Statistics},
  volume = {48},
  number = {4},
  pages = {1875--1897},
  year = {2020},
  month = aug,
  doi = {10.1214/19-AOS1875},
  url = {https://doi.org/10.1214/19-AOS1875},
}
\clearpage

\section{Supplementary Material} 
\begingroup
\setcounter{tocdepth}{2} 
\renewcommand{\contentsname}{Table of Contents}
\tableofcontents
\endgroup

\clearpage
\subsection{Notation List}

\centerline{\bf Data Generating Process}
\bgroup
\def\arraystretch{1.5}
\noindent
\begin{longtable}{p{1.25in}p{3.25in}}
$Z_i$ & $i$th observation, $(X_1, A_1, \ldots, X_T, A_T, Y)$ \\
$X_t$ & Covariates at time $t$ \\
$A_t \in \mathcal{A} = \{0,1\}$ & Binary treatment at time $t$ \\
$Y$ & Outcome \\
$\bar{X}_t$, $\bar{A}_t$ & Covariate and treatment histories, $(X_1, \ldots, X_t)$ and $(A_1, \ldots, A_t)$ \\
$H_t$ & History prior to treatment at time $t$, $(\bar{X}_t, \bar{A}_{t-1})$ \\
$h_t$, $\bar{x}_t$, $\bar{a}_t$ & Realizations of $H_t$, $\bar{X}_t$, $\bar{A}_t$
\end{longtable}
\egroup
\vspace{0.25cm}

\centerline{\bf Measure Theory \& Analysis}
\bgroup
\def\arraystretch{1.5}
\noindent
\begin{longtable}{p{1.25in}p{3.25in}}
$\mathcal{P}$ & Model: a collection of probability measures \\
$P_0$ & True distribution, $P_0 \in \mathcal{P}$ \\
$P(B)$ & Probability of the event $B$ \\
$P[f]$ & $\int f(z)\, dP(z)$ \\
$\Pi(\cdot \mid Z_{1:n})$ & Posterior distribution \\
$\Pi_{\text{BB}}(\cdot \mid Z_{1:n})$ & Bayesian bootstrap posterior \\
$\Pi \times \Pi_{\text{BB}}$ & Product posterior used by the one-step correction \\
$\mathcal{S}_n$ & Posterior sieve: a set of laws with $\Pi(\mathcal{S}_n \mid Z_{1:n}) \xrightarrow{P_0} 1$ \\
$\mathbb{P}_n$ & Empirical measure, $\mathbb{P}_n[f] = n^{-1}\sum_{i=1}^n f(Z_i)$ \\
$\mathbb{G}_n$ & Empirical process, $\sqrt{n}(\mathbb{P}_n - P_0)$ \\
$\|h\|_2$ & $L_2(P_0)$ norm, $\sqrt{P_0[h^2]}$; for $h$ a function of $(H_t, A_t)$ this is the observed-arm norm \\
$\|h\|_\infty$ & Supremum norm \\
$N(\varepsilon, \mathcal{F}, \|\cdot\|)$ & $\varepsilon$-covering number of function class $\mathcal{F}$ in the norm $\|\cdot\|$ \\
$J(r, \mathcal{F})$ & Uniform entropy integral, $\sup_Q \int_0^{r} \sqrt{1 + \log N(\varepsilon \|F\|_{Q,2}, \mathcal{F}, L_2(Q))}\, d\varepsilon$, with the supremum over finitely discrete $Q$ and $F$ an envelope of $\mathcal{F}$ \\
$J_\infty(r, \mathcal{F})$ & Sup-norm entropy integral, $\int_0^{r} \sqrt{1 + \log N(\varepsilon, \mathcal{F}, \|\cdot\|_\infty)}\, d\varepsilon$ \\
$d_{BL}$ & Bounded Lipschitz distance
\end{longtable}
\egroup
\vspace{0.25cm}

\centerline{\bf Causal Inference}
\bgroup
\def\arraystretch{1.5}
\noindent
\begin{longtable}{p{1.25in}p{3.25in}}
$\delta \in \mathcal{D}$ & Intervention parameter (suppressed in subscripts) \\
$\Gamma_{t,\delta}(a, \pi; h_t)$ & Intervention map \\
$dQ_{P,t}(a \mid h_t)$ & Intervention distribution at time $t$ under $P$, $\Gamma_t(a, \pi_{P,t}(h_t); h_t)$ \\
$\pi_{P,t}(h_t)$ & Propensity score under $P$, $P(A_t = 1 \mid H_t = h_t)$ \\
$\pi_{P,t}(a \mid h_t)$ & Treatment mass at $a$, $a\,\pi_{P,t}(h_t) + (1-a)\{1 - \pi_{P,t}(h_t)\}$ \\
$w_{P,t}(a \mid h_t)$ & Intervention weight, $dQ_{P,t}(a \mid h_t) / \pi_{P,t}(a \mid h_t)$ \\
$\widetilde{W}_{P,t}$ & Cumulative weight, $\prod_{s=1}^t w_{P,s}(A_s \mid H_s)$, with $\widetilde{W}_{P,0} = 1$ \\
$\mu_P(h_T, a_T)$ & Outcome regression, $E_P[Y \mid H_T = h_T, A_T = a_T]$ \\
$m_{P,t}(h_t, a_t)$ & Sequential regression: $m_{P,T} = \mu_P$ and $m_{P,t}(h_t,a_t) = E_P[\bar{m}_{P,t+1}(H_{t+1}) \mid H_t = h_t, A_t = a_t]$ \\
$\bar{m}_{P,t}(h_t)$ & $\sum_{a} dQ_{P,t}(a \mid h_t)\, m_{P,t}(h_t, a)$, with $\bar{m}_{P,T+1} = Y$ \\
$\psi_P$ & $\psi_P(Q_P)$, the identified functional for $E[Y^Q]$ \\
$\phi_P$, $\phi^*_P$ & Efficient influence function under $P$, and its uncentered version $\phi^*_P = \phi_P + \psi_P$ \\
$\chi_{P,t}(h_t, a'; a)$ & $\partial_\pi \Gamma_t(a, \pi_{P,t}(h_t); h_t)\{a' - \pi_{P,t}(h_t)\}$ \\
$U_{P,t}$ & $\bar{m}_{P,t}(H_t) - w_{P,t}(A_t \mid H_t)\, m_{P,t}(H_t, A_t)$ \\
$\Xi_P$ & $\chi$-correction term of the $\pi$-dependent efficient influence function \\
$R_2(P_0, P)$ & Second-order remainder, $\psi_{P_0} - \psi_P - P_0[\phi_P]$ \\
$\Delta^\pi_t$, $\Delta^m_t$ & $\pi_{P,t} - \pi_{P_0,t}$ and $m_{P,t} - m_{P_0,t}$, for a fixed $P \in \mathcal{S}_n$ (suppressed) \\
$\mathcal{F}^*_n$ & Class of uncentered EIF differences, $\{\phi^*_P - \phi^*_{P_0} : P \in \mathcal{S}_n\}$ \\
$\mathscr{P}_{t,n}$ & Class of propensity differences, $\{\pi_{P,t} - \pi_{P_0,t} : P \in \mathcal{S}_n\}$ \\
$\mathscr{M}_{t,n}$ & Class of regression differences, $\{m_{P,t} - m_{P_0,t} : P \in \mathcal{S}_n\}$ \\
$\rho_{t,n}$ & Propensity rate, $\sup_{P \in \mathcal{S}_n} \|\pi_{P,t} - \pi_{P_0,t}\|_2 \le \rho_{t,n}$ \\
$\varepsilon_{t,n}$ & Regression rate (maximum per-arm norm), $\sup_{P \in \mathcal{S}_n} \|m_{P,t} - m_{P_0,t}\|_2 \le \varepsilon_{t,n}$ \\
$\gamma_n$ & Combined nuisance rate, $\gamma_n \asymp \sum_{t=1}^T (\rho_{t,n} + \varepsilon_{t,n})$ \\
$\varphi_P$ & Plug-in influence function (simulations only)
\end{longtable}
\egroup
\vspace{0.25cm}

\centerline{\bf SoftBART (single time point)}
\bgroup
\def\arraystretch{1.5}
\noindent
\begin{longtable}{p{1.25in}p{3.25in}}
$f$, $f_0$ & SoftBART function and its true value \\
$\Phi$ & Standard normal CDF; probit propensity $\pi = \Phi(f)$ \\
$T_C(v)$ & Clipping, $\max\{-C, \min\{C, v\}\}$; the EIF is evaluated at $T_C(m)$ \\
$c_\pi$ & Overlap constant, $\pi_{P_0} \in [c_\pi, 1 - c_\pi]$ with $c_\pi = \Phi(-R)$ \\
$\alpha$, $p$, $d$ & H\"older smoothness, number of covariates, sparsity level \\
$\bar{\epsilon}_n$, $\epsilon_n$ & SoftBART base contraction rate and its log-inflated version \\
$\tilde{\mathcal{F}}_n$ & SoftBART sieve function class \\
$g_k(x; \mathcal{K}_k, \mathcal{L}_k)$ & Tree $k$ \\
$\mathcal{K}_k$, $\mathcal{L}_k$ & Structure and leaf parameters of tree $k$ \\
$\mathcal{G}_n$ & Function class for a single tree \\
$K \le K_n$ & Number of trees, with $K_n$ the sieve bound \\
$\bar{D}$, $d_0$ & Sieve depth bound and prior depth cap \\
$\Lambda$ & Maximum number of leaves per tree, $2^{\bar{D}}$ \\
$\tau \in [\tau_n^L, \tau_n^U]$ & Bandwidth, shared by all branches \\
$j_b$, $d_{\max}$ & Split coordinate at branch $b$; maximum number of distinct split coordinates \\
$c_b \in [0,1]$ & Split location (normalized covariates) \\
$\varsigma$ & Gating link \\
$s_b(x)$ & Split function, $\varsigma\big((x_{j_b} - c_b)/\tau\big)$ \\
$\omega_l(x)$ & Gate for leaf $l$, $\prod_{b \in \mathrm{path}(l)} \tilde{s}_b(x)$, where $\tilde{s}_b = s_b$ if $l$ is left of $b$ and $\tilde{s}_b = 1 - s_b$ if right \\
$\Omega_n$, $\Omega_{\text{split},n}$ & Classes of gates and of split functions \\
$\ell_l \in [-\bar{\ell}_n, \bar{\ell}_n]$ & Leaf value, with $\bar{\ell}_n$ the sieve bound \\
$\mathcal{V}_n$ & Leaf-value vectors, $[-\bar{\ell}_n, \bar{\ell}_n]^{\Lambda}$ \\
$\beta_\omega$, $\beta_\ell$ & Covering radii for gates and leaf values
\end{longtable}
\egroup
\vspace{0.25cm}

\subsection{Stochastic Interventions for Longitudinal Settings}
In this subsection, we extend the main ideas of the paper to the longitudinal setting. The updated notation introduced here will be used throughout the remainder of the appendix. 

Consider independent and identically distributed data $(Z_1, \ldots, Z_n)$, generated from an unknown distribution $P_0 \in \mathcal{P}$, where
\begin{equation*}
    Z = (X_1, A_1, X_2, A_2, \ldots, 
    X_T, A_T, Y).
\end{equation*}
Here, $X_t$ denotes a vector of covariates measured prior to treatment $A_t$ at time $t$, and $Y$ is the final outcome. We assume binary treatment at time $t$ (i.e., the support of $A_t$ is $\mathcal{A} = \{0,1\}$, for $t = 1, \dots, T$). Overbars are used to denote the past history of a variable, such that $\Bar{X}_t = (X_1, \ldots, X_t)$ and $\Bar{A}_t = (A_1, \ldots, A_t)$, and we let $H_t = (\Bar{X}_t, \Bar{A}_{t-1})$ denote the past history just prior to treatment at time $t$, with support $\mathcal{H}_t$. Using the updated notation, the data generating distribution admits the following factorization, 
\begin{equation*}
P_0(Z)
=
P_0(Y \mid \bar X_T, \bar A_T)
\prod_{t=1}^T
P_0(X_t \mid \bar X_{t-1}, \bar A_{t-1})
P_0(A_t \mid \bar X_t, \bar A_{t-1}).
\end{equation*}
Finally, let $Y^{\Bar{a}_T}$ denote the potential outcome under treatment sequence $\Bar{a}_T$, that is, $Y^{\Bar{a}_T}$ represents the counterfactual outcome that would have been observed had $\Bar{A}_{T} = \Bar{a}_T$.

We define a longitudinal stochastic intervention, $Q$, as
\begin{equation*}
    dQ(\bar a_T) = \prod_{t=1}^T dQ_{t}(a_t \mid h_t),
\end{equation*}
where $Q_{t}(a_{t} \mid h_{t}) $ is the conditional intervention distribution at time point $t$. Borrowing notation from \citet{Kennedy2018ipsi}, let $Y^{Q}$ denote the potential outcome under a stochastic intervention, and similarly, define the mean counterfactual outcome under stochastic intervention $Q$ as
\begin{equation} \label{eq:long-estimand}
    E\big[Y^Q\big] = \int_{\bar{\mathcal{A}}} \int_{\bar{\mathcal{X}}} E\big[Y^{\Bar{a}_T} \mid \bar{X}_T = \bar{x}_{T}, \bar{A}_{T} = \bar{a}_{T}] \prod_{t = 1}^{T} dQ(a_t \mid h_t) dP_0(x_t \mid h_{t-1}, a_{t-1}),
\end{equation}
 where $\Bar{\mathcal{A}} = \mathcal{A}_1 \times \ldots \times \mathcal{A}_T$ and $\Bar{\mathcal{X}} = \mathcal{X}_1 \times \ldots \times \mathcal{X}_T$. The estimand \eqref{eq:long-estimand} can be identified in the longitudinal setting under the following assumptions.
\begin{assumption} \label{ass:time-varying-identifying}
\text{ }
\begin{enumerate}
    \item[(a)] (Consistency) $Y = Y^{\Bar{a}_T}$ if $\Bar{A}_T = \Bar{a}_T$.
    \item[(b)] (Exchangeability) $A_t \perp\!\!\!\perp Y^{\Bar{a}_T} \mid H_t$.
    \item[(c)] (Weak Positivity) $dP_0(a_t \mid H_t) = 0 \implies dQ_t(a_t \mid h_t) = 0$.
\end{enumerate}
\end{assumption}
\cref{ass:time-varying-identifying} is extremely similar to the single-timepoint case (\cref{ass:identifying}). Consistency must hold for all treatment trajectories and any mention of the covariates X is replaced by the history $H_t$. Consequently, if \cref{ass:time-varying-identifying} holds, $\psi(Q) \equiv E[Y^Q]$ can be expressed as a functional of the observed data distribution, $P_0$.
\begin{lemma}
    Under \cref{ass:time-varying-identifying}, the time-varying stochastic intervention, $\psi(Q) \equiv E\big[Y^Q\big]$, is identified by 
    \begin{equation*}
        \psi_{P_0}(Q) = \int_{\Bar{\mathcal{A}}} \int_{\bar{\mathcal{X}}} \mu(h_T, a_T) \prod_{t=1}^T dQ_{t}(a_t \mid h_t)dP_0(x_t \mid h_{t-1},a_{t-1}),
    \end{equation*}
    where $\mu(h_T, a_T) = E\left[Y \mid H_T = h_T, A_T = a_T\right]$.
\end{lemma}

In other words, $\psi(Q)$ can be written as the expectation over hypothetical treatment trajectories according to the stochastic intervention $Q$ and the covariates drawn from the observed data distribution. A formal proof of this result is given in \citet{Kennedy2018ipsi}; the result essentially follows from the sequential g-formula \citep{robins1986gformula}. 

To define the corresponding nuisance functions, for $t = 0, \ldots, T-1$ we have 
\begin{equation}
m_t(P)(h_t, a_t) = 
\int_{\mathcal{R}_t} 
\mu(h_T, a_T) 
\prod_{s=t+1}^T dQ_{s}(a_s \mid h_s)
dP_0(x_s \mid h_{s-1}, a_{s-1}),
\end{equation}
where $\mathcal{R}_t = (\mathcal{H}_T \times \mathcal{A}_T)\backslash \mathcal{H}_t$, and for $t = T$ we let $m_T(h_T, a_T) = \mu(h_T, a_T)$ and $m_{T+1}(h_{T+1}, a_{T+1}) = Y$. To ease notation we suppress the dependence of $m_t$ on $Q$. Alternatively, $m_t(h_t,a_t)$ can be defined recursively as
\begin{equation} \label{eq:mean-response-recursive}
    m_{t}(P)(h_t,a_t) = \int m_{t+1}(h_{t+1}, a_{t+1}) dQ_{t+1}(a_{t+1} \mid h_{t+1})dP(x_{t+1} \mid h_{t}, a_{t} ).
\end{equation}

Finally, we state the longitudinal efficient influence function (EIF) for $\psi_{P_0}(Q)$ for two cases: $\pi$-fixed and $\pi$-dependent stochastic interventions.
\begin{lemma} 
\label[lemma]{def:EIF-not-dependent}
The efficient influence function (EIF) for a $\pi$-fixed stochastic intervention is
\begin{align*}
    \phi_{\text{fix}, P_0} &= \sum_{t = 0}^T \left\{ \int _{\mathcal{A}_{t + 1}} m_{t+1} (H_{t+1}, a_{t+1}) dQ_{t+1} (a_{t+1}\mid H_{t+1}) - m_t(H_t, A_t) \right\} \prod_{s=0}^t \frac{dQ_{s} (A_s \mid H_s)}{dP_0(A_s \mid H_s)} \\
    &= \sum_{t = 1}^T 
    \left\{ \int _{\mathcal{A}_t} m_{t} 
    (H_{t}, a_{t}) dQ_{t} (a_{t}\mid H_{t}) - m_t(H_t, A_t) \frac{dQ_{t} (A_t \mid H_t)}{dP_0(A_t \mid H_t)} \right\} 
    \prod_{s=0}^{t-1} \frac{dQ_{s} (A_s \mid H_s)}{dP_0(A_s \mid H_s)} \\
    & \;\;\;\;\;\;\;\;\;\;\;\; + \prod_{s=1}^T \frac{dQ_{s} (A_s \mid H_s)}{dP_0(A_s \mid H_s)} Y - \psi_{P_0}(Q),
\end{align*}
where we define $dQ_{T+1} = 1$, $dQ_{0} (a_0\mid h_0) / dP_0(a_0\mid h_0) = 1$, and $m_0 = \psi_{P_0}(Q)$.
\end{lemma}

\begin{lemma} 
\label[lemma]{def:EIF-dependent}
The efficient influence function (EIF) for a $\pi$-dependent stochastic intervention is
\begin{equation*}
    \phi_{\text{dep}, P_0} = \phi^*_{\text{fix}, P_0} + \sum_{t = 1}^T \left\{ \prod_{s=0}^{t-1} \frac{dQ_s(A_s \mid H_s)}{dP_0(A_s \mid H_s)} \right\} \int_{\mathcal{A}_t} \chi_{P_0,t}(H_t,A_t;a_t)m_t(H_t,a_t)d\nu(a_t) - \psi_{P_0}(Q),
\end{equation*}
where $\phi^{*}_{\text{fix}, P_0}$ is the uncentered efficient influence function from \cref{def:EIF-not-dependent} under an intervention $Q$ not depending on $P_0$, $\nu$ is a dominating measure for the distribution of $A_t$, and $\{ \mathbf{1}(H_t = h_t)/ dP_0(h_t) \} \chi_{P_0,t}(H_t, A_t; a_t)$ is the efficient influence function for $dQ_{t}(a_t \mid h_t)$.
\end{lemma}

The proofs of \cref{def:EIF-not-dependent,def:EIF-dependent} follow from an application of the chain rule property of influence functions and are outlined in \citet{Kennedy2018ipsi} and \citet{kennedy2024semiparametric}.

\paragraph{EIF Derivation Examples}
We now provide the explicit uncentered EIF derivations for the IPSI and the PTI in the single timepoint setting as illustrative examples.

\noindent\textit{Example Three} (IPSI EIF)
First note that 
\begin{equation*}
    \Gamma_\delta(a,\pi(x);x) = \frac{a\delta \pi(X) + (1-a)(1-\pi(X))}{\delta \pi(X) + 1 - \pi(X)}.
\end{equation*}
Then after some calculations we arrive at 
\begin{equation*}
    \frac{\partial}{\partial \pi} \Gamma_\delta(a,\pi(x);x) = 
    \frac{(2a -1)\delta}{(\delta \pi(X) + 1 - \pi(X))^2},
\end{equation*}
which implies by \eqref{eq:chi-def}
\begin{equation*}
    \chi(X, A; a) = \frac{(2a -1)\delta (A-\pi(X))}{(\delta \pi(X) + 1 - \pi(X))^2}.
\end{equation*}
Then using \cref{def:static_eif} and \cref{def:dynamic_eif},
\begin{align*}
    \phi^*_{\text{IPSI}} &=  \int_{\mathcal{A}} \mu(X,a)dQ(a \mid X) +   \frac{dQ(A\mid X)}{dP_0(A\mid X)} \left[  Y - \mu(X,A) \right] + \int_{\mathcal{A}} \chi(X,A;a)\mu(X,a)d\nu(a), \\
    &=  \frac{\mu(X,0)(1-\pi(X)) + \mu(X,1) \, \delta \, \pi(X) }{\delta \pi(X) + 1 - \pi (X)} \\
    &+ \frac{A \, \delta \pi(X) + (1-A)(1-\pi(X))}{\left(\delta \pi(X) + 1 - \pi(X)\right) \left( A\pi(X) + (1-A)(1-\pi(X))\right) } 
    \left[  Y - \mu(X,A) \right] \\
    &+ \frac{(\mu(X,1) - \mu(X,0))\delta \,(A-\pi(X))}{(\delta \pi(X) + 1 - \pi(X))^2}.
\end{align*}

\noindent\textit{Example Four} (PTI EIF)
First note that 
\begin{equation*}
    \Gamma_\delta(a,\pi(x);x) = \frac{a \pi(X)^\delta + (1-a)(1-\pi(X))^\delta}{\pi(X)^\delta + (1 - \pi(X))^\delta}.
\end{equation*}
Then after some calculations we arrive at 
\begin{equation*}
    \frac{\partial}{\partial \pi} \Gamma_\delta(a,\pi(x);x) = 
    \frac{(2a -1) \, \delta \,
    \pi(X)^{\delta -1} \, (1-\pi(X))^{\delta-1}}{(\pi(X)^\delta + (1-\pi(X))^\delta)^2},
\end{equation*}
which implies by \eqref{eq:chi-def}
\begin{equation*}
    \chi(X, A; a) = 
    \frac{(2a -1) \, \delta \,
    \pi(X)^{\delta -1} \, (1-\pi(X))^{\delta-1} (A-\pi(X))}{(\pi(X)^\delta + (1-\pi(X))^\delta)^2}.
\end{equation*}
Then using \cref{def:static_eif} and \cref{def:dynamic_eif},
\begin{align*}
    \phi^*_{\text{PTI}} &=  \int_{\mathcal{A}} \mu(X,a)dQ(a \mid X) +   \frac{dQ(A\mid X)}{dP_0(A\mid X)} \left[  Y - \mu(X,A) \right] + \int_{\mathcal{A}} \chi(X,A;a)\mu(X,a)d\nu(a), \\
    &= \frac{\mu(X,0)(1-\pi(X))^\delta + \mu(X,1) \pi(X)^\delta}{\pi(X)^\delta + (1 - \pi(X))^\delta} \\
    &+ \frac{A \pi(X)^\delta + (1-A)(1-\pi(X))^\delta}{\left(\pi(X)^\delta + (1 - \pi(X))^\delta\right) \left( A\pi(X) + (1-A)(1-\pi(X))\right)} \left[  Y - \mu(X,A) \right] \\
    &+ \frac{(\mu(X,1) - \mu(X,0)) \, \delta \,
    \pi(X)^{\delta -1} \, (1-\pi(X))^{\delta-1} (A-\pi(X))}
    {(\pi(X)^\delta + (1-\pi(X))^\delta)^2}
\end{align*}

\subsection{Main BvM Assumptions}
We now update \cref{ass:main-bvm} to the longitudinal setting.
\begin{assumption} 
 \label{ass:time-varying-conditions}
Fix $\delta \in \mathcal{D}$.
There exists a sequence of measurable subsets $(\mathcal{S}_n)_n$ of $\mathcal{P}$ satisfying the following properties for all $t$:

\begin{enumerate}
    \item[(a)] \textit{Posterior mass}
    \begin{equation*}
    \Pi(\mathcal{S}_n \mid Z_{1:n}) \xrightarrow{P_0} 1,
\end{equation*}
    \item[(b)] \textit{Nuisance contraction and second-order rate}.

    There exist sequences $\rho_{t,n}, \varepsilon_{t,n} \to 0$ such that
    \begin{equation*}
        \sup_{P \in \mathcal{S}_n}\|\pi_{P,t} - \pi_{P_0,t}\|_{2} \le \rho_{t,n},
        \quad
        \sup_{P \in \mathcal{S}_n} \max_{a \in \{0,1\}}
        \|m_{P,t}(\cdot, a) - m_{P_0,t}(\cdot, a)\|_{2} \le \varepsilon_{t,n},
    \end{equation*}
    where the norm on $m_{P,t}(\cdot,a)$ is $L_2(P_0)$ over $H_t$. In addition, for each $t = 1, \ldots, T$,
    \begin{equation*}
        \sqrt{n}\, \rho_{t,n}\left(\rho_{t,n} + \varepsilon_{t,n}\right) \to 0.
    \end{equation*}
    
    \item[(c)] \textit{Regularity conditions.}

    For some $C < \infty$ and some $\kappa > 4$, the following hold,
    \begin{enumerate}
        \item[(i)] (Bounds on the truth).
\begin{equation*}
     P_0\big[|Y|^{\kappa}\big] \le C, \qquad
    \sup_{h_T,a_T} P_0\big[Y^2 \mid H_T = h_T, A_T = a_T\big] \le C,\footnotemark
\end{equation*}
\footnotetext{%
    The conditional second-moment bound can be dropped. Under $P_0|Y|^\kappa\le C$,
    H\"older's inequality gives
    $\sup_{P\in \mathcal{S}_n}\|\phi^*_P-\phi^*_{P_0}\|_{P_0,2}\lesssim\sum_{t=1}^T\big(\varepsilon_{t,n}+\rho_{t,n}^{(\kappa-2)/\kappa}\big)$,
    and all results continue to hold with
    $\gamma_n\asymp\sum_{t=1}^T\big(\varepsilon_{t,n}+\rho_{t,n}^{(\kappa-2)/\kappa}\big)$,
    at the cost of much messier rates.}
\begin{equation*}
    \|m_{P_0,t}\|_\infty \le C, \qquad
    \left\| \frac{dQ_{P_0,t,\delta}}{dP_{0,t}} \right\|_\infty \le C.
\end{equation*}
        \item[(ii)] (Bounds on the sieve).
        \begin{equation*}
            \sup_{P \in \mathcal{S}_n} \|m_{P,t} \|_\infty \le C, \quad \sup_{P \in \mathcal{S}_n} \left\| \frac{dQ_{P, t,\delta}}{dP_{t}}  \right\|_\infty \le C.
        \end{equation*}
        \item[(iii)] (Smooth intervention map). There is a function $\Gamma_t: \{0,1\} \times [0,1] \times \mathcal{H}_t \to [0,1]$, jointly measurable in $(\pi, h_t)$ and not depending on $P$, with $\Gamma_t(0, \pi; h_t) + \Gamma_t(1, \pi; h_t) = 1$, such that for all $P \in \mathcal{S}_n \cup \{P_0\}$,
        \begin{equation*}
            dQ_{P,t} (A_t = a \mid H_t = h_t) = \Gamma_t(a, \pi_{P,t}(h_t); h_t).
        \end{equation*}
        Further, for each $a \in \{0,1\}$ and $h_t \in \mathcal{H}_t$, the map $\pi \mapsto \Gamma_t(a, \pi; h_t)$ is twice continuously differentiable on $[0,1]$ with 
        \begin{equation*}
            \sup_{a \in \{0,1\}} \sup_{h_t \in \mathcal{H}_t} \sup_{\pi \in [0,1]} 
            \left\{ |\partial_\pi \Gamma_t(a, \pi; h_t)| + |\partial^2_\pi \Gamma_t(a, \pi; h_t)| \right\}
            \leq C.
        \end{equation*}
        Moreover, $\Gamma_t(1, 0; h_t) = 0$ and $\Gamma_t(0, 1; h_t) = 0$ for all $h_t$.\footnote{Together, these conditions imply that the intervention depends on $P$ only through the propensity score at the current history, and that the intervention never assigns a treatment that has zero probability under $P$.}
        
    \end{enumerate}

    \item[(d)] \textit{Entropy condition on nuisance classes}.
Define the centered nuisance function classes
\begin{equation*}
\mathscr{P}_{t,n} = \{ \pi_{P,t} - \pi_{P_0,t} : P \in \mathcal{S}_n \}, \qquad
\mathscr{M}_{t,n} = \{ m_{P,t} - m_{P_0,t} : P \in \mathcal{S}_n \},
\end{equation*}
and let $\gamma_n \asymp \sum_{t=1}^T (\rho_{t,n} + \varepsilon_{t,n})$. Then
\begin{equation*}
\sum_{t=1}^T \Big\{ J_\infty\big(\gamma_n,\mathscr{P}_{t,n}\big)
+ J_\infty\big(\gamma_n,\mathscr{M}_{t,n}\big) \Big\}
= o\!\left(1 \wedge \gamma_n\, n^{\frac{\kappa-2}{4\kappa-4}} \right),
\end{equation*}
where $J_\infty(\delta,\mathscr F)=\int_0^\delta\sqrt{1+\log N(\varepsilon,\mathscr F,\|\cdot\|_\infty)}\,d\varepsilon$
is the sup-norm entropy integral.

\end{enumerate}
\end{assumption}

\subsection{Stochastic Intervention Smoothness Lemmas}

The following \cref{lem:smoothness,lem:weight_lip} allow us to generalize to a large class of stochastic interventions, needing only smoothness requirements with respect to $\pi$. The inequalities in \cref{lem:smoothness,lem:weight_lip} are used throughout the proofs of the no second-order bias condition, \cref{lem:no-snd-bias}, and in the pointwise Lipschitz bound, \cref{lem:Lipschitz_multiple_time}. Later in this subsection, we present the proof of \cref{prop:logodds_family}, which gives a class of stochastic interventions that satisfy the assumptions of \cref{lem:smoothness,lem:weight_lip}. The subsection ends with a discussion of the power tilt intervention and the proof of \cref{cor:clipped}, which permits the use of clipped one-step posteriors in our setup.

\begin{lemma}\label[lemma]{lem:smoothness}
    Suppose \cref{ass:time-varying-conditions}(c)(iii) holds. Then for all $P, P' \in \mathcal{S}_n \cup \{P_0\}$, $a \in \{0,1\}$, $a' \in \mathcal{A}_t$, and $h_t$, with constants only depending on $C$:  
    \begin{enumerate}
        \item [(i)] 
        $\left|dQ_{P,t} - dQ_{P', t}+ \int_{\mathcal{A}_t}\chi_{P,t}(h_t, a';a) \, d\pi_{P', t} \right| 
        \lesssim 
        \left|\pi_{P,t} - \pi_{P',t}\right|^2,$

         \item [(ii)] $\left|\int_{\mathcal{A}_t}\chi_{P,t}(h_t, a'; a)  d\pi_{P', t} (a' \mid h_t) \right| \lesssim  \left|\pi_{P, t} - \pi_{P', t} \right|,$
         
        \item [(iii)]
        $ \left| \chi_{P,t} (h_t, a' ; a) \right|\le C,$

        \item[(iv)]
        $\left| dQ_{P,t} - dQ_{P', t} \right|
        \lesssim \left| \pi_{P, t} - \pi_{P',t} \right|,$
        
        \item [(v)]
        $
        \left| \chi_{P,t}(h_t, a' ;a) - \chi_{P', t} (h_t, a' ;a) \right|
        \lesssim \left| \pi_{P, t} - \pi_{P', t} \right|$.
    \end{enumerate}
\end{lemma}
\begin{proof}
Fix $t$, $a \in \{0,1\}$, and $h_t$. All statements are pointwise in $(a, h_t)$, and every constant depends only on $C$. For $P \in \mathcal{S}_n \cup \{P_0\}$, we define 
\begin{equation*}
    \pi_{P,t} := \pi_{P,t} (1 \mid h_t), \quad \pi_{P,t} (a' \mid h_t) = a' \pi_{P,t} + (1- a')(1 - \pi_{P,t}), \quad a' \in \{0,1\}.
\end{equation*}
Let $\Gamma_a(\pi) := \Gamma_t(a, \pi; h_t)$, with $\Gamma_a'$ and $\Gamma_a''$ its first and second derivatives in $\pi$. By \cref{ass:time-varying-conditions}(c)(iii), 
\begin{equation*}
    dQ_{P,t} := dQ_{P,t} (A_t = a \mid H_t = h_t) = \Gamma_a(\pi_{P,t}),
\end{equation*}
and 
\begin{equation*}
    \sup_{\pi \in [0,1]} \left\{ |\Gamma'_a(\pi)| + | \Gamma''_a(\pi)| \right\} \le C.
\end{equation*}
Recall 
\begin{equation*}
    \chi_{P,t} (h_t, a' ; a) = \Gamma'_a(\pi_{P,t})(a' - \pi_{P,t}).
\end{equation*}
\paragraph{Step 1}
First, since $\mathcal{A}_t = \{0,1\}$,
\begin{align*}
    \int_{\mathcal{A}_t} \chi_{P,t} (h_t, a';a) d\pi_{P', t} (a' \mid h_t)
    &= \sum_{a' \in \{0,1\}} \Gamma_a'(\pi_{P,t}) (a' - \pi_{P,t}) \pi_{P',t} (a' \mid h_t), \\
    &= \Gamma_a'(\pi_{P,t}) \left[ - \pi_{P,t} (1 - \pi_{P', t}) + (1 - \pi_{P,t})\pi_{P',t} \right], \\
    &= \Gamma_a'(\pi_{P,t}) (\pi_{P',t} - \pi_{P,t}).
\end{align*}
\paragraph{(ii)} 
By Step 1 and $|\Gamma_a'| \le C$, 
\begin{equation*}
    \left|\int_{\mathcal{A}_t} \chi_{P,t} d\pi_{P',t} \right| \le C|\pi_{P,t} - \pi_{P',t}|.
\end{equation*}
\paragraph{(i)} By \cref{ass:time-varying-conditions}(c)(iii) and Step 1, the left hand side of (i) equals 
\begin{equation*}
    \Gamma_a(\pi_{P,t}) - \Gamma_a(\pi_{P',t}) + \Gamma_a'(\pi_{P,t}) ( \pi_{P',t} - \pi_{P,t}).
\end{equation*}
Since $\Gamma_a$ is $C^2$ on [0, 1], Taylor's theorem with Lagrange remainder gives some $\xi$ between $\pi_{P,t}$ and $\pi_{P',t}$ such that 
\begin{equation*}
    \Gamma_a(\pi_{P', t}) = \Gamma_a(\pi_{P,t}) + \Gamma_a'(\pi_{P,t})(\pi_{P', t} - \pi_{P,t}) + \frac{1}{2}\Gamma_a''(\xi)(\pi_{P', t} - \pi_{P,t})^2.
\end{equation*}
This implies that the left hand side of (i) equals 
\begin{equation*}
    - \frac{1}{2} \Gamma_a''(\xi)(\pi_{P', t} - \pi_{P,t})^2.
\end{equation*}
Since $\xi \in [0,1]$ and $|\Gamma_a''| \le C$ on $[0,1]$, 
\begin{equation*}
    \left|dQ_{P,t} - dQ_{P', t}+ \int_{\mathcal{A}_t}\chi_{P,t}(h_t, a';a) \, d\pi_{P', t} \right| 
        \le \frac{C}{2}
        \left|\pi_{P,t} - \pi_{P',t}\right|^2.
\end{equation*}
\paragraph{(iii)}
Since $a' \in \{0,1\}$ and $\pi_{P,t} \in [0,1]$, we have $|a' - \pi_{P,t}| \le 1$, so 
\begin{equation*}
    |\chi_{P,t}(h_t, a' ; a)| \le |\Gamma_a'(\pi_{P,t})| \le C.
\end{equation*}
\paragraph{(iv)} By the Mean Value Theorem, there is some $\eta$ between $\pi_{P,t}$ and $\pi_{P',t}$ with 
\begin{align*}
    |dQ_{P,t} - dQ_{P',t}| 
    &= 
    |\Gamma_a(\pi_{P,t}) - \Gamma_a(\pi_{P',t})|, \\
    &= |\Gamma_a'(\eta)| |\pi_{P,t} - \pi_{P',t}|, \\
    &\le C | \pi_{P,t} - \pi_{P',t}|.
\end{align*}
\paragraph{(v)} Add and subtract $\Gamma_a'(\pi_{P',t})(a' - \pi_{P,t})$,
\begin{equation*}
    \chi_{P,t} - \chi_{P',t} = \left[ \Gamma_a'(\pi_{P,t}) - \Gamma_a'(\pi_{P',t}) \right](a' - \pi_{P,t}) + \Gamma_a'(\pi_{P', t}) (\pi_{P', t} - \pi_{P,t}).
\end{equation*}
For the first term, the Mean Value Theorem applied to $\Gamma_a'$, together with $|\Gamma_a''| \le C$, and $|a' - \pi_{P,t}| \le 1$, gives a bound $C|\pi_{P,t} - \pi_{P',t}|$. For the second term, $|\Gamma'_a| \le C$ gives the same bound. Hence we have shown the desired results. 
\end{proof}

\begin{lemma}
\label[lemma]{lem:weight_lip}
     Suppose \cref{ass:time-varying-conditions}(c)(iii) holds. For $P \in \mathcal{S}_n \cup \{P_0\}$ and $a \in \{0,1\}$, define 
     \begin{equation*}
         w_{P,t}(a \mid h_t) = \frac{dQ_{P,t}(A_t = a \mid H_t = h_t)}{\pi_{P,t}(a \mid h_t)} = 
         \frac{\Gamma_t(a, \pi_{P,t}(h_t); h_t)}{\pi_{P,t}(a \mid h_t)}.
     \end{equation*}
     Here $\pi_{P,t}(h_t) = \pi_{P,t}(1 \mid h_t)$, extended by continuity where the denominator is zero. Then for all $P, P' \in \mathcal{S}_n \cup \{P_0\}$, $a \in \{0,1\}$ and $h_t$, 
     \begin{equation*}
         |w_{P,t}(a \mid h_t) - w_{P', t}(a \mid h_t)| \lesssim |\pi_{P,t}(h_t) - \pi_{P',t}(h_t)|.
     \end{equation*}
\end{lemma}
\begin{proof}
    Fix $h_t$. Write $\pi_{P,t} = \pi_{P,t}(1 \mid h_t)$ and $\Gamma_a(\pi) = \Gamma_t(a, \pi; h_t)$. By \cref{ass:time-varying-conditions}(c)(iii), $\Gamma_a \in C^2([0,1])$, $|\Gamma'_a|, |\Gamma_a''| \le C$, and $\Gamma_1(0) = 0, \Gamma_0(1) = 0$.

    Consider $a = 1$. Then $w_{P,t}(1 \mid h_t) = w(\pi_{P,t})$, where 
    \begin{equation*}
        w(\pi) := \frac{\Gamma_1(\pi)}{\pi}.
    \end{equation*}
    On $(0,1]$, $w$ is well defined and continuously differentiable with
    \begin{equation*}
        w'(\pi) = \frac{\pi \Gamma_1'(\pi) - \Gamma_1(\pi)}{\pi^2}.
    \end{equation*}
    \paragraph{Behavior at 0.} At $0$, since $\Gamma_1 \in C^2$ and $\Gamma_1(0) = 0$, Taylor's theorem with Peano remainder gives 
    \begin{align*}
        \Gamma_1(\pi) &= \Gamma_1'(0)\pi + \frac{1}{2} \Gamma_1''(0)\pi^2 + o(\pi^2), \\
        \Gamma_1'(\pi) &= \Gamma_1'(0) + \Gamma_1''(0) \pi + o(\pi).
    \end{align*}
    Hence 
    \begin{equation*}
        w(\pi) = \Gamma_1'(0) + \frac{1}{2}\Gamma_1''(0) \pi + o(\pi) \to \Gamma_1'(0),
    \end{equation*}
    and 
    \begin{equation*}
        w'(\pi) = \frac{\Gamma_1''(0)\pi^2 - \frac{1}{2} \Gamma_1''(0) \pi^2 + o(\pi^2)}{\pi^2} \to \frac{1}{2}\Gamma_1''(0).
    \end{equation*}
    Set $w(0) := \Gamma_1'(0)$, so $w$ is continuous on $[0,1]$.
    \paragraph{Differentiability at 0.} For $\pi \in (0,1]$, the Mean Value Theorem gives some $\eta_\pi \in (0, \pi)$ with 
    \begin{equation*}
        \frac{w(\pi) - w(0)}{\pi} = w'(\eta_\pi).
    \end{equation*}
    As $\pi \to 0$, $\eta_\pi \to 0$, so the right hand side tends to $\frac{1}{2}\Gamma_1''(0)$. Hence $w'(0) = \frac{1}{2}\Gamma_1''(0)$, $w'$ is continuous at 0, and $w \in C^1([0,1])$.
    \paragraph{Lipschitz bound.} For $\pi \in (0,1]$, since $\Gamma_1(0) = 0$, integration by parts gives
    \begin{equation*}
        \pi \Gamma_1'(\pi) - \Gamma_1(\pi) = \int_0^\pi t\, \Gamma_1''(t)\, dt,
    \end{equation*}
    so
    \begin{equation*}
        |w'(\pi)| = \frac{1}{\pi^2}\left| \int_0^\pi t\, \Gamma_1''(t)\, dt \right| \le \frac{C}{\pi^2} \int_0^\pi t\, dt = \frac{C}{2}.
    \end{equation*}
    Since $w' $ is continuous on $[0,1]$, also $|w'(0)| \le C/2$, so $\|w'\|_\infty \le C/2$, with a constant that does not depend on $h_t$. By the Mean Value Theorem, there is some $\tilde \pi$ between $\pi_{P,t}$ and $\pi_{P',t}$ such that 
    \begin{equation*}
        |w(\pi_{P,t}) - w(\pi_{P', t})| = |w'(\tilde \pi)| |\pi_{P,t} - \pi_{P', t}| 
        \le
        \frac{C}{2} |\pi_{P,t} - \pi_{P', t}|.
    \end{equation*}
    \paragraph{Case $a = 0$.} Here the weight is $\frac{\Gamma_0(\pi)}{1-\pi}$. Let $u = 1-\pi$ and $g(u) = \Gamma_0(1-u)$. Then $g \in C^2([0,1])$, $g(0) = \Gamma_0(1) = 0$, and $|g'|, |g''| \le C$. The weight is $\frac{g(u)}{u}$, so the argument above applies with $\Gamma_1$ replaced by $g$. Since $|u_P - u_{P'}| = |\pi_{P,t} - \pi_{P', t}|$, we get the same bound as before. Hence we have proven the desired result.
\end{proof}

The weight bounds and smoothness conditions used above hold for a broad class of interventions, as stated in \cref{prop:logodds_family}, which we now prove.
\paragraph{Proof of \cref{prop:logodds_family}}
\begin{proof}
Write $\|s\|_\infty = \sup_{x,\pi}|s(\pi;x)|$, and define $\|\partial_\pi s\|_\infty$ and $\|\partial_\pi^2 s\|_\infty$ analogously. Recall that all three are finite by assumption. Fix $x$ and write $D(\pi)=\pi e^{s(\pi;x)}+1-\pi$. Since $e^{s(\pi;x)}\in[m,M]$ with $m=e^{-\|s\|_\infty}>0$ and $M = e^{\|s\|_\infty}$,
$D(\pi)\ge\pi m+(1-\pi)\ge\min(m,1)>0$. The intervention weights
\begin{equation*}
\frac{\Gamma_s(1,\pi;x)}{\pi}=\frac{e^{s(\pi;x)}}{D(\pi)}, \qquad \frac{\Gamma_s(0,\pi;x)}{1-\pi}=\frac{1}{D(\pi)}
\end{equation*}
are therefore bounded by $\max(M,1)/\min(m,1)$ for every $\pi \in [0,1]$ and $x$, and in particular for $\pi_P$ with any $P \in \mathcal{S}_n \cup \{P_0\}$. This gives the weight bounds in \cref{ass:main-bvm}(c)(i)--(ii).

Next, $\pi \mapsto \Gamma_s(1,\pi;x)$ is a quotient of $C^2$ functions with a denominator bounded below by $\min(m,1)$, hence $C^2$ on $[0,1]$. Its first and second derivatives are bounded by a constant depending only on $\|s\|_\infty$, $\|\partial_\pi s\|_\infty$, and $\|\partial_\pi^2 s\|_\infty$, and hence uniformly in $x$. Since $\Gamma_s(0,\cdot\,;x) = 1 - \Gamma_s(1,\cdot\,;x)$, the same bounds hold for $\Gamma_s(0,\cdot\,;x)$. Joint measurability of $\Gamma_s$ in $(\pi,x)$ follows from that of $s$. Finally, $\Gamma_s(1,0;x) = 0$ and $\Gamma_s(0,1;x) = 1 - e^{s(1;x)}/e^{s(1;x)} = 0$ for every $x$. Together these give \cref{ass:main-bvm}(c)(iii).

For the incremental propensity score intervention, $s\equiv\log\delta$, so $\partial_\pi s = \partial_\pi^2 s = 0$ and $\|s\|_\infty = |\log \delta|$. This is bounded for $\delta$ in any compact subset of $(0,\infty)$, so all of the above bounds hold uniformly over such $\delta$. The longitudinal analogue, with $h_t$ in place of $x$, verifies \cref{ass:time-varying-conditions}(c) in the same way.
\end{proof}

\paragraph{The power tilt intervention}
The power tilt \eqref{eq:PT_density} does not depend on $x$, and it is not a bounded log-odds shift. Its log-odds map $\operatorname{logit}\Gamma_\delta(1,\pi)=\delta\operatorname{logit}\pi$ scales rather than translates, and it is unbounded as $\pi\to0,1$. The denominator $\pi^\delta+(1-\pi)^\delta$ is bounded below by $2^{1-\delta} \wedge 1 > 0$ on $[0,1]$, so the boundary behavior is governed by the numerator. As $\pi\to0$,
\begin{equation*}
\frac{\Gamma_\delta(1,\pi)}{\pi}=
\frac{\pi^{\delta-1}}{\pi^\delta+(1-\pi)^\delta}
\sim
\pi^{\delta-1},
\qquad
\partial_\pi^2\, \Gamma_\delta(1,\pi) \asymp \delta(\delta-1)\pi^{\delta-2},
\end{equation*}
and symmetrically as $\pi \to 1$. The weights are therefore bounded if and only if $\delta \ge 1$, and the second derivative is bounded on $[0,1]$ if and only if $\delta \ge 2$ or $\delta = 1$. Moreover, $\Gamma_\delta(1,0) = 0$ and $\Gamma_\delta(0,1) = 0$ for every $\delta > 0$. Hence the power tilt satisfies \cref{ass:main-bvm}(c)(i)--(iii) for $\delta$ in any compact subset of $[2,\infty)$, and for $\delta = 1$, where $\Gamma_\delta(1,\pi) = \pi$; the same holds for the longitudinal analogue in \cref{ass:time-varying-conditions}. For $\delta \in (0,2)\setminus\{1\}$, these conditions fail on $[0,1]$. They hold on $[\tilde c, 1-\tilde c]$ for every $\tilde c > 0$, which is the setting of \cref{cor:clipped}.

\paragraph{Proof of \cref{cor:clipped}}
\begin{proof}
Since $\sum_i \widetilde W_i = 1$, each one-step draw satisfies $\psi^{(b)}_{\text{1-step}} = \sum_i \widetilde W^{(b)}_i \phi^*_{P^{(b)}}(Z_i)$, so the one-step posterior depends on a posterior draw only through its nuisance functions $(\mu, \pi)$. It therefore suffices to verify the conditions of \cref{thm:bvm_single_time} for the clipped nuisances $(T_C(\mu_P), T_{\tilde c}(\pi_P))$, $P \in \mathcal{S}_n$, which receive the same posterior mass as $\mathcal{S}_n$.

Both $T_C$ and $T_{\tilde c}$ are projections onto intervals containing the true values, so they fix $(\mu_{P_0}, \pi_{P_0})$ and are $1$-Lipschitz. Hence $|T_C(\mu_P) - \mu_{P_0}| \le |\mu_P - \mu_{P_0}|$ and $|T_{\tilde c}(\pi_P) - \pi_{P_0}| \le |\pi_P - \pi_{P_0}|$ pointwise, so the rates in \cref{ass:main-bvm}(b) are preserved, and the sup-norm covering numbers of the clipped classes are no larger than those of $\mathscr{P}_n$ and $\mathscr{M}_n$, so (d) is preserved.

For (i), $\|T_C(\mu_P)\|_\infty \le C$ for every $P$, which is the sieve bound in (c)(ii).

For (ii), all clipped propensity scores and $\pi_{P_0}$ lie in $[\tilde c, 1-\tilde c]$, so the intervention weights are at most $1/\tilde c$, giving the weight bounds in (c)(i)--(ii). In \cref{lem:smoothness,lem:weight_lip} and \cref{lem:no-snd-bias}, $\Gamma$ and its derivatives are evaluated only at points between two propensity scores, hence only on $[\tilde c, 1-\tilde c]$, so smoothness there suffices. On this interval $w(\pi) = \Gamma(1,\pi;x)/\pi$ satisfies $|w'| \le C/\tilde c + 1/\tilde c^2$, and similarly for $a = 0$, so \cref{lem:weight_lip} holds with a constant depending on $C$ and $\tilde c$. The boundary condition is used only where a propensity score equals zero or one, which does not occur. The conclusion follows from \cref{thm:bvm_single_time}.

Finally, for the power tilt, $\Gamma_\delta(1,\pi) = \pi^\delta/\{\pi^\delta + (1-\pi)^\delta\}$ is infinitely differentiable on $[\tilde c, 1 - \tilde c]$, with denominator bounded below by $2^{1-\delta} \wedge 1$, and its first two derivatives are bounded uniformly for $\delta$ in any compact subset of $(0, \infty)$.
\end{proof}

\subsection{Empirical Process Lemmas}

\cref{lem:covering-2-entrp-int,lem:entrp-int-const-ineq,lem:maximal_inequality} are self-contained empirical process results that aid the proof of  \cref{prop:emp-process} at the end of this subsection.
\begin{lemma}
\label[lemma]{lem:covering-2-entrp-int}
Let $\mathcal F$ be a class of functions with envelope $F$, let $\mathscr G_1,\dots,\mathscr G_m$ be classes of functions. Suppose there exist constants $c>0$ and $K\ge 1$ such that, for all $\varepsilon\in(0,\delta]$,
\begin{equation}\label{eq:cover}
\sup_Q N\big(\varepsilon\|F\|_{Q,2},\mathcal F,L_2(Q)\big) \le \prod_{j=1}^m N\big(c\varepsilon,\mathscr G_j,\|\cdot\|_\infty\big)^{K}.
\end{equation}
Then
\begin{equation*}
J(\delta,\mathcal F) \le \frac{\sqrt K}{c}\sum_{j=1}^m J_\infty(c\delta,\mathscr G_j),
\end{equation*}
which implies $J(\delta,\mathcal F)\lesssim \sum_{j=1}^m J_\infty(c\delta,\mathscr G_j)$.
\end{lemma}

\begin{proof}
Fix $Q$ and $\varepsilon\in(0,\delta]$. Covering numbers are at least $1$, so all logarithms are nonnegative.
Because $K\ge1$, \eqref{eq:cover} gives
\begin{equation*}
1+\log N\big(\varepsilon\|F\|_{Q,2},\mathcal F,L_2(Q)\big)
\le 1+K\sum_{j=1}^m\log N\big(c\varepsilon,\mathscr G_j,\|\cdot\|_\infty\big)
\le K\Big(1+\sum_{j=1}^m\log N\big(c\varepsilon,\mathscr G_j,\|\cdot\|_\infty\big)\Big).
\end{equation*}
Next, we take square roots and use $\sqrt{1+\sum_j a_j}\le\sum_j\sqrt{1+a_j}$ for $a_j\ge0$. Further, integrating over $\varepsilon\in(0,\delta]$ gives
\begin{equation*}
\int_0^\delta \sqrt{1+\log N\big(\varepsilon\|F\|_{Q,2},\mathcal F,L_2(Q)\big)}\,d\varepsilon
\le \sqrt K\sum_{j=1}^m\int_0^\delta \sqrt{1+\log N(c\varepsilon,\mathscr G_j,\|\cdot\|_\infty)}\,d\varepsilon.
\end{equation*}
Using the substitution $u=c\varepsilon$ gives
\begin{equation*}
\sqrt K\sum_{j=1}^m\int_0^\delta \sqrt{1+\log N(c\varepsilon,\mathscr G_j,\|\cdot\|_\infty)}\,d\varepsilon
= \frac{\sqrt K}{c}\sum_{j=1}^m\int_0^{c\delta}\sqrt{1+\log N(u,\mathscr G_j,\|\cdot\|_\infty)}\,du .
\end{equation*}
The right-hand side does not depend on $Q$, so taking the supremum over $Q$ gives the desired result.
\end{proof}

Next we provide a short lemma showing that the constant in the sup-norm entropy integral is irrelevant asymptotically. 
\begin{lemma}
\label[lemma]{lem:entrp-int-const-ineq}
Let $J$ denote either $J(\cdot,\mathcal F)$ or $J_\infty(\cdot,\mathcal F)$. Then for every fixed $c>0$, $J(c\delta)\lesssim J(\delta)$.
\end{lemma}

\begin{proof}
If $c\le 1$, this holds because $J$ is nondecreasing in $\delta$. If $c>1$, substitute $\varepsilon=cu$. Covering numbers are nonincreasing in the radius, so
\begin{align*}
\int_0^{c\delta}\sqrt{1+\log N\big(\varepsilon\|F\|_{Q,2},\mathcal F,L_2(Q)\big)}\,d\varepsilon
& = c\int_0^{\delta}\sqrt{1+\log N\big(cu\|F\|_{Q,2},\mathcal F,L_2(Q)\big)}\,du, \\
& \le c\int_0^{\delta}\sqrt{1+\log N\big(u\|F\|_{Q,2},\mathcal F,L_2(Q)\big)}\,du .
\end{align*}
Take $\sup_Q$ on both sides. The $\|\cdot\|_\infty$ case is the same argument without $Q$.
\end{proof}

\begin{lemma}[Maximal inequality]
\label[lemma]{lem:maximal_inequality}
Let $\mathcal{F}_n = \{f_\theta : \theta \in \Theta_n \}$ be a class of measurable functions, with semimetric $d$ and envelope function $F$ that does not depend on $n$. Suppose that
\begin{enumerate}
  \item[(i)] $(\Theta_n, d)$ is a separable metric space, and $\theta \mapsto f_\theta(z)$ is continuous for every z,
  \item[(ii)] $\|F\|_{P_0,\kappa}<\infty$ for some $\kappa>4$,
  \item[(iii)] $P_0 f^2<\gamma_n^2\,P_0F^2$ for every $f\in\mathcal F_n$ and some $\gamma_n\in(0,1)$,
  \item[(iv)] $J(\gamma_n, \mathcal{F}_n) = o\left(1 \wedge \gamma_n n^{\frac{\kappa - 2}{4\kappa - 4}}\right)$.
\end{enumerate}
 Then
\begin{equation*}
\mathbb E^*_{P_0}\|\mathbb G_n\|_{\mathcal F_n}
= o(1).
\end{equation*}
\end{lemma}

\begin{proof}
For ease of expression, let $J(\gamma_n, \mathcal{F}_n) \equiv J(\gamma_n)$. We will first show that $\mathcal{F}_n$ is $P_0-$measurable \citep[Def 2.3.3]{vandervaartwellner2023weak}. From assumption (i), $(\Theta_n, d)$ is a separable metric space, so it has a countable dense subset $\Theta_0$. Density gives some sequence $\theta_k \in \Theta_0$ with $d(\theta_k, \theta) \to 0$, and by continuity in (i), we have $f_{\theta_k}(z) \to f_{\theta}(z)$ for every $z$. Thus by \citet{vandervaartwellner2023weak} Example 2.3.4, $\mathcal{F}$ is $P_0$-measurable. Next, since $x \mapsto x^2$ and $x \mapsto x^4$ are continuous the same argument applies, so $\mathcal{F}^2$ and $\mathcal{F}^4$ are $P_0-$measurable.

By the above argument, and Assumptions (ii) and (iii), we can invoke \citet[Thm.~2.14.8]{vandervaartwellner2023weak}, which gives\footnote{%
  \citet[Thm.~2.14.8]{vandervaartwellner2023weak} state the moment condition as
  $P_0F^{(4p-2)/(p-1)}<\infty$ for some $p>1$. We set $\kappa=\frac{4p-2}{p-1}$,
  equivalently $p=\frac{\kappa-2}{\kappa-4}$, which maps $p\in(1,\infty)$ onto
  $\kappa\in(4,\infty)$. Then $\frac1p=\frac{\kappa-4}{\kappa-2}$,
  $2-\frac1p=\frac{\kappa}{\kappa-2}$ and $\frac{p}{2p-1}=1-\frac2\kappa$.}
\begin{equation*}
\mathbb E^*_{P_0}\|\mathbb G_n\|_{\mathcal F_n}
\lesssim
J(\gamma_n)\left[1+
\frac{J\big(\gamma_n^{\frac{\kappa-4}{\kappa-2}}\big)}{\gamma_n^2\sqrt n}
\left(\frac{\|F\|_{P_0,\kappa}}{\|F\|_{P_0,2}}\right)^{\frac{\kappa}{\kappa-2}}
\right]^{1-\frac2\kappa}
\|F\|_{P_0,2}.
\end{equation*}
It remains to bound $J(\gamma_n^\beta)$ for $\beta=\frac{\kappa-4}{\kappa-2}\in(0,1)$. Let
\begin{equation*}
    g_Q(\varepsilon)=\sqrt{1+\log N(\varepsilon\|F\|_{Q,2},\mathcal F_n,L_2(Q))}
\end{equation*}
which is nonincreasing in $\varepsilon$. With $c=\gamma_n^{\beta-1}\ge1$, the
substitution $\varepsilon=c\eta$ gives
\begin{equation*}
J(\gamma_n^\beta)
=\sup_Q\int_0^{c\gamma_n}g_Q(\varepsilon)\,d\varepsilon
=c\sup_Q\int_0^{\gamma_n}g_Q(c\eta)\,d\eta
\le c\sup_Q\int_0^{\gamma_n}g_Q(\eta)\,d\eta
=\gamma_n^{\beta-1}J(\gamma_n).
\end{equation*}
Since $\beta-1=-\frac{2}{\kappa-2}$, we have
$\gamma_n^{-2}J(\gamma_n^\beta)\le\gamma_n^{-\frac{2(\kappa-1)}{\kappa-2}}J(\gamma_n)$.
Substituting this into \eqref{eq:thm-2.14.8} and using the fact that $\|F\|_{P_0,\kappa} < C$ gives
\begin{align*}
    \mathbb E^*_{P_0}\|\mathbb G_n\|_{\mathcal F_n}
&\lesssim
J(\gamma_n)
\left[1+
\frac
{J (\gamma_n) \gamma_n^{-\frac{2(\kappa-1)}{\kappa-2}}}
{ \sqrt n}
\right]^{\frac{\kappa -2}{\kappa}}, \\
&\lesssim
\left[J(\gamma_n)^{\frac{\kappa}{\kappa-2}} +
\frac
{J (\gamma_n)^{\frac{2\kappa -2}{\kappa -2}} \gamma_n^{-\frac{2(\kappa-1)}{\kappa-2}}}
{\sqrt n}
\right]^{\frac{\kappa - 2}{\kappa}}, \\
&\lesssim J(\gamma_n) + J(\gamma_n)^{\frac{2\kappa -2}{\kappa}} \gamma_n ^{- \frac{2\kappa -2}{\kappa}}
n^{-\frac{\kappa-2}{2\kappa}}, \\
&\lesssim J(\gamma_n) + \left[ J(\gamma_n) \gamma_n ^{-1 }
n^{-\frac{\kappa-2}{4\kappa-4}}\right] ^ {\frac{2\kappa -2}{\kappa}}.
\end{align*}
Therefore condition (iv) implies that $E_{P_0}^* \|\mathbb{G}_n\|_{\mathcal{F}_n} = o(1)$, which is the desired result.
\end{proof}

\paragraph{Proof of \cref{prop:emp-process}}

\begin{proof}
We will show that the following holds 
\begin{enumerate}
    \item[(I)] $\sup_{P \in \mathcal{S}_n} \|\phi^*_P - \phi^*_{P_0}\|_2 \to 0$,
    \item[(II)] $\sup_{P \in \mathcal{S}_n} \big| \mathbb{G}_n[\phi^*_P - \phi^*_{P_0}] \big| \xrightarrow{P_0} 0,$
    \item[(III)] $F = (1 + D)M$ is an envelope satisfying
    \begin{equation*}
\lim_{C \to \infty} \limsup_{n \to \infty}
P_0\!\left(F^2 \mathbf{1}\{F^2 > C\}\right) = 0,
\quad
P_0 F^4 = o(n).
\end{equation*}
\end{enumerate}
(I) follows from \cref{prop:emp-process} (c) and the fact that $\gamma_n \to 0$. It's important to note, we only need convergence here, not a rate condition. 
For (III), from \cref{prop:emp-process} (a) and (b), we have an envelope function $F = (1+D)M$ with $|\phi^*_P| \le F$. With the moment condition in \cref{prop:emp-process} and the fact that $D < \infty$, we have $P_0[F^\kappa] < \infty$, which is not a function of $n$. This implies 
\begin{equation*}
\lim_{C \to \infty} \limsup_{n \to \infty}
P_0\!\left(F^2 \mathbf{1}\{F^2 > C\}\right) = 0,
\quad
P_0 F^4 = o(n),
\end{equation*}
so (III) holds.

For (II), we first need to show the four Assumptions of \cref{lem:maximal_inequality} hold. To achieve this, we will construct a cover of 
\begin{equation*}
    \mathcal{F}^*_n := \{\phi^*_P - \phi^*_{P_0} : P \in \mathcal{S}_n\}.
\end{equation*}
First, fix $n$ large enough that \cref{prop:emp-process} (d) makes every $J_\infty(\gamma_n, \mathscr{H}_{k,n})$ finite. If some $N(\varepsilon_0, \mathscr{H}_{k,n}, \| \cdot \|_\infty)$ were infinite, it would be infinite for all $\varepsilon \le \varepsilon_0$, and the sup-norm entropy integral would diverge. Therefore,  each $\mathscr{H}_{k,n}$ has finite sup-norm covering numbers at every $\varepsilon \in (0, \gamma_n]$. 

Next, we fix $\varepsilon > 0$ and take a sup-norm $\varepsilon$-cover of each $\mathscr{H}_{k,n}$. We place each $P\in \mathcal{S}_n$ in the cell labeled by which center its $k$-th nuisance is near, for each $k$. There are at most 
\begin{equation}
\label{eq:covering-total}
N_\varepsilon = \prod_k N(\varepsilon, \mathscr{H}_{k,n}, \| \cdot \|_\infty)
\end{equation}
such cells. Two members of the same cell are each within $\varepsilon$ of the center, so they are at most $2\varepsilon$ apart by the triangle inequality. We then choose one representative per nonempty cell, $P_c$, such that $d(P, P_c) < 2\varepsilon$ for all $P \in \mathcal{S}_n$. This implies that $(\mathcal{S}_n, d)$ is totally bounded, and hence separable. Also, note that $\theta \mapsto f_\theta(z)$ is continuous from \cref{prop:emp-process} (b), and therefore \cref{lem:maximal_inequality} condition (i) holds. 

\cref{lem:maximal_inequality} condition (ii) holds by \cref{prop:emp-process} (b). Next, \cref{prop:emp-process} (c) gives 
\begin{equation*}
    \|\phi^*_P - \phi^*_{P_0} \|_2 \le L \gamma_n,
\end{equation*}
for some constant $L < \infty$. We could then take  our envelope to be $\tilde F = (L + 1)F$ to ensure the inequality holds strictly, and therefore \cref{lem:maximal_inequality} condition (iii) holds.

We then move on to showing \cref{lem:maximal_inequality} condition (iv). We have previously constructed a covering such that $d(P, P_c) < 2\varepsilon$ for all $P \in \mathcal{S}_n$. Then, using \cref{prop:emp-process} (b), 
\begin{align*}
    |\phi^*_P(z) - \phi^*_{P_c}(z)| &\le M(z) d(P, P_c), \\
    &\le M(z) 2\varepsilon.
\end{align*}
Therefore, for every finitely discrete $Q$, 
\begin{align*}
    \|\phi^*_P - \phi^*_{P_c}\|_{Q,2} &\le \|M\|_{Q,2} 2\varepsilon, \\
    & \le \left\|\frac{F}{1+D}\right\|_{Q,2} 2\varepsilon, \\
    & \le  \|F\|_{Q,2} \frac{2\varepsilon}{1 +D}.
\end{align*}
The at most $N_\varepsilon$ functions from \cref{eq:covering-total} cover $\mathcal{F}^*_n$ at radius $ \|F\|_{Q,2} \frac{2\varepsilon}{1 +D}$ in every $L_2(Q)$. Therefore 
\begin{equation*}
    \sup_Q \log N( \|F\|_{Q,2} \frac{2\varepsilon}{1 +D}, \mathcal{F}^*_n, L_2(Q)) \le \sum_k \log N(\varepsilon, \mathscr{H}_{k,n}, \|\cdot \|_\infty).
\end{equation*}
Using $\varepsilon = \frac{1 + D}{2}\varepsilon'$ gives, for a constant $c > 0$,
\begin{equation*}
    \sup_Q \log N( \|F\|_{Q,2} \varepsilon', \mathcal{F}^*_n, L_2(Q)) \le \sum_k \log N(c\varepsilon', \mathscr{H}_{k,n}, \|\cdot \|_\infty).
\end{equation*}
Then, we can invoke \cref{lem:covering-2-entrp-int} to get 
\begin{equation*}
    J(\gamma_n, \mathcal{F}^*_n) \lesssim \sum_k J_\infty(c \gamma_n, \mathscr{H}_{k,n}).
\end{equation*}
Combining this with \cref{lem:entrp-int-const-ineq} gives 
\begin{equation*}
    J(\gamma_n, \mathcal{F}^*_n) \lesssim \sum_k J_\infty(\gamma_n, \mathscr{H}_{k,n}).
\end{equation*}
Then by \cref{prop:emp-process} (d), 
\begin{equation*}
    J(\gamma_n, \mathcal{F}^*_n) = o\left(1 \wedge \gamma_n n^{\frac{\kappa - 2}{4\kappa -4}}\right),
\end{equation*}
which is condition (iv) of \cref{lem:maximal_inequality}. 

Hence we can invoke \cref{lem:maximal_inequality} to get 
\begin{equation*}
    E_{P_0}^* \|\mathbb{G}_n \|_{\mathcal{F}^*_n} = o(1).
\end{equation*}
Finally, Markov's inequality gives 
\begin{equation*}
    \sup_{P \in \mathcal{S}_n} |\mathbb{G}_n [\phi^*_P - \phi^*_{P_0}]| \xrightarrow[]{P_0} 0.
\end{equation*}
Hence (II) holds, and we have proven all 3 desired results. 
\end{proof}

\subsection{Lipschitz EIF with respect to the nuisance functions}
\cref{lem:Lipschitz_multiple_time}, stated at the end of this subsection, shows that the uncentered EIF is pointwise Lipschitz in its nuisance functions. This implies \cref{prop:emp-process} (b). We first state some technical results, \cref{lem:A-lipschitz,lem:weight-product-lipschitz,lem-pi-dependent-lipschitz}, that are consumed in proof of \cref{lem:Lipschitz_multiple_time}.

\begin{lemma} 
\label[lemma]{lem:A-lipschitz}
Let $w_{P,s} = \frac{dQ_{P,s} (A_s \mid H_s)}{dP(A_s \mid  H_s)}$ and $U_{P,t}:= \int_{\mathcal{A}_t} m_{P,t}(H_t, a_t)dQ_{P,t}(a_t \mid H_t) - m_{P,t}(H_t, A_t)w_{P,t}$. Under \cref{ass:time-varying-conditions}(c), the following holds pointwise for all $P, P' \in \mathcal{S}_n \cup \{P_0\}$,
\begin{equation*}
    |U_{P,t} - U_{P',t}| \lesssim \sum_{a_t=0}^1 |m_{P,t}(H_t,a_t) - m_{P',t}(H_t,a_t)| + |\pi_{P,t} - \pi_{P',t}|.
\end{equation*}
\end{lemma}
\begin{proof} 
For ease of notation, let $m_{P,t} = m_{P,t}(H_t,A_t)$. Then
\begin{align*}
    |U_{P,t} - U_{P', t}| &= \Bigg|\int_{\mathcal{A}_t}  m_{P,t}(H_t, a_t)dQ_{P,t}(a_t \mid H_t)- m_{P',t}(H_t,a_t)  dQ_{P',t}(a_t \mid H_t) \\
    & \quad -  \left[m_{P,t}w_{P,t} - m_{P',t}w_{P',t} \right]\Bigg|, \\
    & = \Bigg|\sum_{a_t=0}^1 m_{P,t}(H_t, a_t)dQ_{P,t}(a_t \mid H_t)- m_{P',t}(H_t,a_t)  dQ_{P',t}(a_t \mid H_t) \\
    & \quad -  \left[m_{P,t}w_{P,t} - m_{P',t}w_{P',t}\right]\Bigg|, \\
    & = \Bigg|\sum_{a_t=0}^1 (m_{P,t}(H_t, a_t) - m_{P',t}(H_t, a_t))dQ_{P,t}(a_t \mid H_t)  \\
    & \quad +  m_{P',t}(H_t,a_t) (dQ_{P,t}(a_t \mid H_t) - dQ_{P',t}(a_t \mid H_t)) \\
    & \quad - \left[(m_{P,t} - m_{P',t})w_{P,t} + m_{P',t}(w_{P,t} - w_{P',t}) \right]\Bigg|.
\end{align*}
The first equality follows by definition. The second equality follows from binary treatments. The third equality follows by adding and subtracting $m_{P',t}(H_t,a_t)dQ_{P,t}(a_t \mid H_t)$ and $m_{P',t}w_{P,t}$. Next, 
\begin{align*}
    |U_{P,t} - U_{P', t}| &\le
    \left|\sum_{a_t=0}^1 (m_{P,t}(H_t, a_t) - m_{P',t}(H_t, a_t))dQ_{P,t}(a_t \mid H_t) \right| \\
    &\quad + \left|\sum_{a_t=0}^1 m_{P',t}(H_t,a_t) (dQ_{P,t}(a_t \mid H_t) - dQ_{P',t}(a_t \mid H_t))\right| \\
    &\quad + |(m_{P,t} - m_{P',t})w_{P,t}| + |m_{P',t}(w_{P,t} - w_{P',t})|.
\end{align*}
Then, by the triangle inequality and \cref{ass:time-varying-conditions}(c)(i)--(ii), 
\begin{equation*}
    |U_{P,t} - U_{P', t}| \lesssim \sum_{a_t=0}^1 \Big\{ |m_{P,t}(H_t,a_t) - m_{P',t}(H_t, a_t) | + |dQ_{P,t}(a_t \mid H_t) - dQ_{P',t}(a_t \mid H_t)| \Big\} + |w_{P,t} - w_{P',t}|.
\end{equation*}
Finally, by \cref{lem:smoothness,lem:weight_lip}, 
\begin{equation*}
    |U_{P,t} - U_{P', t}| \lesssim \sum_{a_t=0}^1 |m_{P,t}(H_t,a_t) - m_{P',t}(H_t,a_t) | + |\pi_{P,t} - \pi_{P',t}|,
\end{equation*}
which is the desired result.
\end{proof}

\begin{lemma}
\label[lemma]{lem:weight-product-lipschitz}
Let $w_{P,t}$ be defined as in \cref{lem:A-lipschitz}. Then under \cref{ass:time-varying-conditions}(c), the following holds pointwise for all $P, P' \in \mathcal{S}_n \cup \{P_0\}$,
\begin{equation*}
    \left| \prod_{s=1}^{t-1} w_{P,s} - \prod_{s=1}^{t-1} w_{P',s} \right| \lesssim
     \sum_{s=1}^{t-1} |\pi_{P,s} - \pi_{P',s}|.
\end{equation*}
\end{lemma}
\begin{proof}
\begin{align*}
\left| \prod_{s=1}^{t-1} w_{P,s} - \prod_{s=1}^{t-1} w_{P',s} \right|
&= \left| \sum_{s=1}^{t-1} \left( \prod_{r=1}^{s-1} w_{P,r} \right)
   \left( w_{P,s} - w_{P',s} \right)
   \left( \prod_{r=s+1}^{t-1} w_{P',r} \right) \right|, \\
&\le \sum_{s=1}^{t-1} \left| \left( \prod_{r=1}^{s-1} w_{P,r} \right)
   \left( w_{P,s} - w_{P',s} \right)
   \left( \prod_{r=s+1}^{t-1} w_{P',r} \right) \right|, \\
   &\lesssim \sum_{s=1}^{t-1} \left|  
    w_{P,s} - w_{P',s} 
   \right|, \\
   &\lesssim \sum_{s=1}^{t-1} |\pi_{P,s} - \pi_{P',s}|.
\end{align*}
The equality is a telescoping product identity. The first inequality is the triangle inequality. The second inequality uses the weight bounds in \cref{ass:time-varying-conditions}(c)(i)--(ii), which bound each product. The final inequality follows from \cref{lem:weight_lip}.
\end{proof}

\begin{lemma}
\label[lemma]{lem-pi-dependent-lipschitz}
Let 
\begin{equation*}
    \Xi_P := \sum_{t = 1}^T \left\{ \prod_{s=0}^{t-1} \frac{dQ_{P,s}(A_s \mid  H_s)}{dP(A_s\mid H_s)} \right\} \int_{\mathcal{A}_t} \chi_{P,t}(H_t,A_t;a_t)\, m_{P,t}(H_t,a_t)\, d\nu(a_t).
\end{equation*}
Then under \cref{ass:time-varying-conditions}, the following holds pointwise for all $P, P' \in \mathcal{S}_n \cup \{P_0\}$,
\begin{equation*}
    |\Xi_P - \Xi_{P'}| \lesssim \sum_{t=1}^T  \sum_{a_t=0}^1  |m_{P,t}(H_t,a_t) - m_{P', t}(H_t,a_t)| +  \sum_{t=1}^T |\pi_{P,t} - \pi_{P',t} |.
\end{equation*}
\end{lemma}

\begin{proof}
For ease of notation, we suppress all function arguments, and write $w_{P,s}$ as in \cref{lem:A-lipschitz}. We consider $\Xi_P - \Xi_{P'}$ and add and subtract $\prod_{s=0}^{t-1} w_{P',s} \int_{\mathcal{A}_t} \chi_{P,t}\, m_{P,t}\, d\nu$,
\begin{equation}
\label{eq:xi_term}
\begin{aligned}
    \Xi_{P} - \Xi_{P'} = \sum_{t=1}^T \Bigg[ &\left( \prod_{s=0}^{t-1} w_{P,s} - \prod_{s=0}^{t-1} w_{P',s} \right) \int_{\mathcal{A}_t} \chi_{P,t}\, m_{P,t}\, d\nu \\
    & + \prod_{s=0}^{t-1} w_{P',s} \int_{\mathcal{A}_t} \big(\chi_{P,t}\, m_{P,t} - \chi_{P',t}\, m_{P',t}\big)\, d\nu \Bigg].
\end{aligned}
\end{equation}
For the first term in \eqref{eq:xi_term},
\begin{align*}
    \left| \left( \prod_{s=0}^{t-1} w_{P,s} - \prod_{s=0}^{t-1} w_{P',s}\right) \int_{\mathcal{A}_t} \chi_{P,t}\, m_{P,t}\, d\nu \right|
    &=  \left| \left( \prod_{s=0}^{t-1} w_{P,s} - \prod_{s=0}^{t-1} w_{P',s} \right) \sum_{a_t=0}^1 \chi_{P,t}\, m_{P,t}  \right|, \\
    &\lesssim \left| \prod_{s=0}^{t-1} w_{P,s} - \prod_{s=0}^{t-1} w_{P',s} \right|, \\
    &\lesssim \sum_{s = 1}^{t-1} |\pi_{P,s} - \pi_{P',s} |.
\end{align*}
The equality follows from binary treatments. The first inequality is the triangle inequality, using \cref{ass:time-varying-conditions}(c)(ii) and \cref{lem:smoothness}(iii). The final inequality follows from \cref{lem:weight-product-lipschitz}.

Next we bound the second term in \eqref{eq:xi_term},
\begin{align*}
   \left| \prod_{s=0}^{t-1} w_{P',s} \int_{\mathcal{A}_t} \big(\chi_{P,t}\, m_{P,t} - \chi_{P',t}\, m_{P',t}\big)\, d\nu \right |
   &\lesssim 
   \left |\sum_{a_t=0}^1  \chi_{P,t}\, m_{P,t} - \chi_{P',t}\, m_{P',t} \right |, \\
   &= \left |\sum_{a_t=0}^1  (\chi_{P,t} - \chi_{P',t})\, m_{P,t} + (m_{P,t} - m_{P',t})\, \chi_{P',t} \right|, \\
  &\le \sum_{a_t = 0}^1 |(\chi_{P,t} - \chi_{P',t})\, m_{P,t}| + \sum_{a_t=0}^1 |(m_{P,t} - m_{P',t})\, \chi_{P',t}|, \\
  &\lesssim \sum_{a_t = 0}^1 |\chi_{P,t} - \chi_{P',t} | + \sum_{a_t=0}^1 |m_{P,t} - m_{P',t} |, \\
  &  \lesssim |\pi_{P,t} - \pi_{P',t} | + \sum_{a_t=0}^1 |m_{P,t} - m_{P',t}|.
\end{align*}
The first inequality follows from H\"older's inequality, using \cref{ass:time-varying-conditions}(c)(i)--(ii) and binary treatments. The equality follows from adding and subtracting $m_{P,t}\,\chi_{P',t}$. The next inequality is the triangle inequality. The fourth follows from H\"older's inequality, using \cref{ass:time-varying-conditions}(c)(ii) and \cref{lem:smoothness}(iii). The final inequality follows from \cref{lem:smoothness}(v).

Finally, using our above two bounds for \eqref{eq:xi_term} and the triangle inequality, we arrive at 
\begin{equation*}
    |\Xi_{P} - \Xi_{P'}| \lesssim 
    \sum_{t=1}^T \sum_{a_t=0}^1  |m_{P,t} - m_{P', t}| +  \sum_{t=1}^T |\pi_{P,t} - \pi_{P',t} |.
\end{equation*}
Hence we have shown the desired result.
\end{proof}


\paragraph{Main Lipschitz Proof}
We show that the uncentered efficient influence function is pointwise Lipschitz in the nuisance functions. This is used in \cref{thm:bvm_single_time} to verify \cref{ass:yiu-main} (b), the $L_2$ convergence condition, and to construct covers for the EIF as part of verifying \cref{ass:yiu-main} (c).

\begin{lemma} \label[lemma]{lem:Lipschitz_multiple_time}
    Under \cref{ass:time-varying-conditions}(c), the uncentered efficient influence function for a general stochastic intervention $Q$ is pointwise Lipschitz in its nuisance functions $\eta = (\pi, m)$. That is, the following holds pointwise for all $P, P' \in \mathcal{S}_n \cup \{P_0\}$,
\begin{equation*}
    | \phi^*_P - \phi^*_{P'} | \lesssim
    \sum_{t=1}^T \sum_{a_t=0}^1 |m_{P,t}(H_t,a_t) - m_{P', t}(H_t,a_t)| + (1 + |Y|) \sum_{t=1}^T |\pi_{P,t} - \pi_{P',t} |,
\end{equation*}
where the constant depends only on $C$ and $T$. 
\end{lemma}
\begin{proof}
\noindent \textbf{$\pi$-fixed case} 

We will first prove that the uncentered efficient influence function for a $\pi$-fixed stochastic intervention is pointwise Lipschitz in $\eta = (\pi, m)$, given suitable regularity conditions. From \cref{def:EIF-not-dependent} we can define the uncentered efficient influence function as 
\begin{equation*}
    \phi^*_{P} = \sum_{t = 1}^T \Psi_{P,t} + R_{P},
\end{equation*}
where 
\begin{align*}
    \Psi_{P,t} &=  \left\{ \int _{\mathcal{A}_t} m_{P,t} (H_{t}, a_{t})\, dQ_{P,t} (a_{t} \mid H_{t}) - m_{P,t}(H_t, A_t)\, w_{P,t}\right\} \prod_{s=0}^{t-1} w_{P,s}, \\
    R_P &= \prod_{s=1}^T w_{P,s}\, Y ,
\end{align*} 
and $w_{P,s}$ is defined as in \cref{lem:A-lipschitz},
\begin{equation*}
    w_{P,s} = \frac{dQ_{P,s} (A_s \mid H_s)}{dP(A_s \mid  H_s)}  = \frac{dQ_{P,s}(A_s \mid H_s)}{\pi_{P,s}( H_s)^{A_s}(1-\pi_{P,s}(H_s))^{1-A_s}}.
\end{equation*}
Therefore, 
\begin{align*}
     | \phi^*_{P} - \phi^*_{P'}| 
     & = 
     \left| \left(\sum_{t = 1}^T \Psi_{P,t} + R_P \right) - \left(\sum_{t = 1}^T \Psi_{P',t} + R_{P'} \right) \right| , \\
     & =
    \left| 
    \left(\sum_{t = 1}^T \Psi_{P,t} - \Psi_{P',t}\right)  + R_P - R_{P'} \right|, \\
    & \le 
    \sum_{t=1}^T | \Psi_{P,t} - \Psi_{P',t} | + |R_P - R_{P'} |, 
\end{align*}
where the last inequality utilizes the triangle inequality.

First, we focus on bounding $|\Psi_{P,t} - \Psi_{P',t} |$. With $U_{P,t}$ as in \cref{lem:A-lipschitz},
\begin{equation}
\begin{aligned}
    |\Psi_{P,t} - \Psi_{P', t} |
    &= \left | U_{P,t} \prod_{s=0}^{t-1} w_{P,s} - U_{P',t} \prod_{s=0}^{t-1} w_{P',s} \right |, \\
    &= \left | \left(U_{P,t} - U_{P',t} \right)  \prod_{s=0}^{t-1} w_{P,s} + U_{P',t} \left( \prod_{s=0}^{t-1} w_{P,s} - \prod_{s=0}^{t-1} w_{P',s} \right) \right |,
\end{aligned}
\end{equation}
and the second equality holds by adding and subtracting $U_{P',t} \prod_{s=0}^{t-1} w_{P,s}$. Then by the triangle inequality,
\begin{equation*}
    |\Psi_{P,t} - \Psi_{P',t} | \le 
    \left | \left(U_{P,t} - U_{P',t} \right)  \prod_{s=0}^{t-1} w_{P,s} \right| +  
    \left | U_{P',t} \left( \prod_{s=0}^{t-1} w_{P,s} - \prod_{s=0}^{t-1} w_{P',s} \right) \right |.
\end{equation*}
Then, using \cref{ass:time-varying-conditions}(c)(i)--(ii),
\begin{equation}
\label{eq:Cap-Psi-diff}
|\Psi_{P,t} - \Psi_{P',t} | 
\lesssim
|U_{P,t} - U_{P',t}| +
\left | \prod_{s=0}^{t-1} w_{P,s} - \prod_{s=0}^{t-1} w_{P',s} \right |.
\end{equation}
Then by \cref{lem:A-lipschitz,lem:weight-product-lipschitz},
\begin{equation*}
    |\Psi_{P,t} - \Psi_{P',t} |  \lesssim    
    \sum_{a_t=0}^1  |m_{P,t}(H_t,a_t) - m_{P', t}(H_t,a_t)| +  \sum_{s=1}^{t} |\pi_{P,s} - \pi_{P',s} |.
\end{equation*}
Next, by \cref{lem:weight-product-lipschitz} (applied with $t = T+1$),
\begin{equation*}
    |R_P - R_{P'}| \lesssim |Y| \sum_{s=1}^T |\pi_{P,s} - \pi_{P',s}|.
\end{equation*}
Therefore, summing over $t$ and combining all of our work gives 
\begin{equation*}
     | \phi^*_P - \phi^*_{P'} | \lesssim
    \sum_{t=1}^T  \sum_{a_t=0}^1  |m_{P,t}(H_t,a_t) - m_{P',t}(H_t,a_t)| + (1 + |Y|) \sum_{t=1}^T|\pi_{P,t} - \pi_{P',t}|,
\end{equation*}
for the $\pi$-fixed case.

\noindent \textbf{$\pi$-dependent case} 

Let $\Xi_P$ be defined as in \cref{lem-pi-dependent-lipschitz}. For the $\pi$-dependent case, by \cref{def:EIF-dependent}, the uncentered efficient influence function is 
\begin{equation*}
    \phi^*_P = \phi^*_{\text{fix},P} + \Xi_P,
\end{equation*}
where $\phi^*_{\text{fix},P}$ is the uncentered efficient influence function for a $\pi$-fixed stochastic intervention.
Then by \cref{lem-pi-dependent-lipschitz},
\begin{equation*}
    |\Xi_{P} - \Xi_{P'}| \lesssim 
    \sum_{t=1}^T   \sum_{a_t=0}^1 |m_{P,t}(H_t,a_t) - m_{P', t}(H_t,a_t)| +  \sum_{t=1}^T |\pi_{P,t} - \pi_{P',t} |.
\end{equation*}
Therefore, combining this with the $\pi$-fixed case gives
\begin{equation*}
     | \phi^*_P - \phi^*_{P'} | \lesssim
    \sum_{t=1}^T  \sum_{a_t=0}^1  |m_{P,t}(H_t,a_t) - m_{P', t}(H_t,a_t)| + (1 + |Y|) \sum_{t=1}^T |\pi_{P,t} - \pi_{P',t} |.
\end{equation*}
Hence we have proven the desired result.
\end{proof}


\subsection{Bernstein-von Mises Theory}
We first provide \cref{lem:no-snd-bias}, which proves the no second-order bias condition as a standalone lemma. This subsection then ends with the proof of \cref{thm:bvm_single_time}.
\begin{lemma}
\label[lemma]{lem:no-snd-bias}
Under \cref{ass:time-varying-conditions} (b) and (c),
\begin{equation*}
\sup_{P \in \mathcal{S}_n}
\big|\sqrt{n}\, R_2(P_0, P)\big|
=
\sup_{P \in \mathcal{S}_n}
\Big|
\sqrt{n} \, \big\{ \psi_{P_0} - \psi_P - P_0[\phi_P] \big\}
\Big|
\to 0.
\end{equation*}
\end{lemma}

\begin{proof}
We begin with the following.
\paragraph{Reduction.} Since $\phi_P = \phi^*_P - \psi_P$ and $\psi_P$ is constant,
\begin{equation*}
    R_2(P_0,P) = \psi_{P_0} - \psi_P - P_0[\phi^*_P - \psi_P] = \psi_{P_0} - P_0[\phi^*_P].
\end{equation*}
It therefore suffices to show that $\sup_{P \in \mathcal{S}_n} \sqrt{n}\,\big|P_0[\phi^*_P] - \psi_{P_0}\big| \to 0$. We first express $P_0[\phi^*_P] - \psi_{P_0}$ exactly as a sum over time points, then integrate out the current treatment to expose products of nuisance errors, and finally bound these products.

Fix $P \in \mathcal{S}_n$. For $P' \in \{P, P_0\}$, let $w_{P',t}(a \mid h_t) = dQ_{P',t}(a \mid h_t)/\pi_{P',t}(a \mid h_t)$ denote the intervention weight of \cref{lem:weight_lip}, and define
\begin{align*}
    \widetilde W_{P',t}
    &= \prod_{s=1}^t w_{P',s}(A_s \mid H_s), 
    \quad 
    \widetilde W_{P',0} = 1, \\
\bar m_{P',t}(H_t) 
&= \sum_{a \in \{0,1\}} dQ_{P',t}(a \mid H_t)\, m_{P',t}(H_t,a),
\quad\bar m_{P',T+1} = Y .
\end{align*}
Write $\Delta^m_t = m_{P,t} - m_{P_0,t}$ and $\Delta^\pi_t = \pi_{P,t} - \pi_{P_0,t}$. By \cref{def:EIF-not-dependent,def:EIF-dependent}, the uncentered efficient influence function is
\begin{equation*}
    \phi^*_{P} = \sum_{t=1}^T \widetilde W_{P,t-1} \Big\{ \bar m_{P,t}(H_t) - w_{P,t}(A_t \mid H_t)\, m_{P,t}(H_t,A_t)
    + \sum_{a} \chi_{P,t}(H_t,A_t;a)\, m_{P,t}(H_t,a) \Big\} + \widetilde W_{P,T}\, Y .
\end{equation*}

\paragraph{Telescoping the remainder.} Throughout this step we abbreviate $\widetilde W_t = \widetilde W_{P,t}$ and $w_t = w_{P,t}(A_t \mid H_t)$, so that $\widetilde W_t = \widetilde W_{t-1}\, w_t$. Substituting $m_{P,t} = \Delta^m_t + m_{P_0,t}$ in the second term of each summand of $\phi^*_P$ gives
\begin{equation*}
    -\widetilde W_{t-1}\, w_t\, m_{P,t}(H_t,A_t) = -\widetilde W_t\, \Delta^m_t(H_t,A_t) - \widetilde W_t\, m_{P_0,t}(H_t,A_t),
\end{equation*}
and therefore
\begin{equation}\label{eq:phi-substituted}
\begin{aligned}
    P_0[\phi^*_P]
    &= \sum_{t=1}^T P_0\Big[ \widetilde W_{t-1} \Big\{ \bar m_{P,t}(H_t) + \sum_{a} \chi_{P,t}(H_t,A_t;a)\, m_{P,t}(H_t,a) \Big\} \Big]
    - \sum_{t=1}^T P_0\big[ \widetilde W_t\, \Delta^m_t(H_t,A_t) \big] \\
    &\quad - \sum_{t=1}^T P_0\big[ \widetilde W_t\, m_{P_0,t}(H_t,A_t) \big] + P_0\big[ \widetilde W_T\, Y \big].
\end{aligned}
\end{equation}
We now evaluate the third sum. By \eqref{eq:mean-response-recursive}, $m_{P_0,t}(H_t,A_t) = P_0[\bar m_{P_0,t+1}(H_{t+1}) \mid H_t, A_t]$, and $\widetilde W_t$ is a function of $(H_t, A_t)$. The tower property therefore gives
\begin{equation*}
    P_0\big[ \widetilde W_t\, m_{P_0,t}(H_t,A_t) \big] = P_0\big[ \widetilde W_t\, \bar m_{P_0,t+1}(H_{t+1}) \big], \qquad t = 1, \ldots, T.
\end{equation*}
Since $\bar m_{P_0,T+1} = Y$, the term with $t = T$ equals $P_0[\widetilde W_T\, Y]$ and cancels the last term in \eqref{eq:phi-substituted}. Shifting the index in the remaining terms,
\begin{equation*}
    - \sum_{t=1}^T P_0\big[ \widetilde W_t\, m_{P_0,t} \big] + P_0\big[ \widetilde W_T\, Y \big]
    = - \sum_{t=1}^{T-1} P_0\big[ \widetilde W_t\, \bar m_{P_0,t+1} \big]
    = - \sum_{t=2}^{T} P_0\big[ \widetilde W_{t-1}\, \bar m_{P_0,t} \big].
\end{equation*}
Finally, by the identification result, $\psi_{P_0} = P_0[\bar m_{P_0,1}(H_1)] = P_0[\widetilde W_0\, \bar m_{P_0,1}(H_1)]$, since $\widetilde W_0 = 1$. Subtracting $\psi_{P_0}$ therefore supplies the missing term with $t = 1$,
\begin{equation*}
    - \sum_{t=2}^{T} P_0\big[ \widetilde W_{t-1}\, \bar m_{P_0,t} \big] - \psi_{P_0} = - \sum_{t=1}^{T} P_0\big[ \widetilde W_{t-1}\, \bar m_{P_0,t} \big].
\end{equation*}
Substituting this into \eqref{eq:phi-substituted}, writing $\widetilde W_t\, \Delta^m_t = \widetilde W_{t-1}\, w_t\, \Delta^m_t$, and collecting terms by $\widetilde W_{t-1}$, we arrive at

\begin{equation}\label{eq:remainder-identity}
\begin{aligned}
    P_0[\phi^*_{P}] - \psi_{P_0}
    &= \sum_{t=1}^T P_0\Big[ \widetilde W_{P,t-1} \Big\{ \bar m_{P,t} - \bar m_{P_0,t}
    - w_{P,t}(A_t \mid H_t)\, \Delta^m_t(H_t,A_t)\\
    & \quad + \sum_{a} \chi_{P,t}(H_t,A_t;a)\, m_{P,t}(H_t,a) \Big\} \Big].
\end{aligned}
\end{equation}
When $T = 1$, the third sum in \eqref{eq:phi-substituted} reduces to $P_0[w_1 m_{P_0,1}] = P_0[w_1 Y]$, which cancels the outcome term directly, and no reindexing is needed.

\paragraph{Integrating out $A_t$.} Since $\widetilde W_{P,t-1}$ is a function of $H_t$, the tower property allows us to replace the term in braces in \eqref{eq:remainder-identity} by its conditional expectation given $H_t$ under $P_0$. To compute this conditional expectation, we decompose
\begin{equation*}
    \bar m_{P,t} - \bar m_{P_0,t} = \sum_a \{dQ_{P,t}(a) - dQ_{P_0,t}(a)\}\, m_{P,t}(H_t,a) + \sum_a dQ_{P_0,t}(a)\, \Delta^m_t(H_t,a),
\end{equation*}
and use the identity $dQ_{P_0,t}(a \mid H_t) = \pi_{P_0,t}(a \mid H_t)\, w_{P_0,t}(a \mid H_t)$. Where $\pi_{P_0,t}(a \mid H_t) = 0$, this identity holds because $\Gamma_t(1,0;H_t) = \Gamma_t(0,1;H_t) = 0$ under \cref{ass:time-varying-conditions}(c)(iii). Defining
\begin{equation*}
    B_t(a) = dQ_{P,t}(a \mid H_t) - dQ_{P_0,t}(a \mid H_t) + \int_{\mathcal{A}} \chi_{P,t}(H_t,a';a)\, \pi_{P_0,t}(a' \mid H_t)\, d\nu(a'),
\end{equation*}
we obtain
\begin{equation}\label{eq:remainder-conditional}
    P_0[\phi^*_{P}] - \psi_{P_0}
    = \sum_{t=1}^T P_0\Big[ \widetilde W_{P,t-1} \sum_{a} \Big\{ B_t(a)\, m_{P,t}(H_t,a)
    + \pi_{P_0,t}(a \mid H_t) \big( w_{P_0,t} - w_{P,t} \big)(a \mid H_t)\, \Delta^m_t(H_t,a) \Big\} \Big].
\end{equation}
We note that each regression error $\Delta^m_t(H_t,a)$ is weighted by $\pi_{P_0,t}(a \mid H_t)$, so that only values on the support of $P_0$ enter.

\paragraph{Bounding the remainder.} We bound the two terms in \eqref{eq:remainder-conditional}. By \cref{lem:smoothness} (i) with $P' = P_0$, we have $|B_t(a)| \le \tfrac{C}{2}(\Delta^\pi_t)^2$, and by \cref{lem:weight_lip}, $|w_{P_0,t}(a \mid H_t) - w_{P,t}(a \mid H_t)| \lesssim |\Delta^\pi_t|$. Moreover, $|m_{P,t}| \le C$ and $|\widetilde W_{P,t-1}| \le C^{t-1}$ by \cref{ass:time-varying-conditions} (c) (ii), and
\begin{equation*}
    \sum_a \pi_{P_0,t}(a \mid H_t)\, |\Delta^m_t(H_t,a)| = P_0\big[|\Delta^m_t(H_t,A_t)| \,\big|\, H_t\big].
\end{equation*}
Combining these bounds, taking expectations, and applying the Cauchy--Schwarz inequality to the cross term yields
\begin{equation*}
    \big|P_0[\phi^*_{P}] - \psi_{P_0}\big|
    \lesssim \sum_{t=1}^T \|\Delta^\pi_t\|_2 \Big( \|\Delta^\pi_t\|_2 + \|\Delta^m_t(H_t, A_t)\|_2 \Big).
\end{equation*}
Finally, since $\pi_{P_0,t}(a \mid H_t) \le 1$,
\begin{equation*}
    \|\Delta^m_t(H_t, A_t)\|_2^2 = P_0\Big[ \sum_a \pi_{P_0,t}(a \mid H_t)\, \Delta^m_t(H_t, a)^2 \Big]
    \le 2 \max_{a} \|\Delta^m_t(\cdot, a)\|_2^2 .
\end{equation*}
Consequently, for all $P \in \mathcal{S}_n$, with a constant depending only on $C$ and $T$,
\begin{equation*}
    \big|P_0[\phi^*_{P}] - \psi_{P_0}\big| \lesssim \sum_{t=1}^T \rho_{t,n}\big(\rho_{t,n} + \varepsilon_{t,n}\big).
\end{equation*}
By \cref{ass:time-varying-conditions} (b), $\sqrt{n}\,\rho_{t,n}(\rho_{t,n} + \varepsilon_{t,n}) \to 0$ for each $t$. Hence
\begin{equation*}
    \sup_{P \in \mathcal{S}_n}\sqrt{n}\,\big|P_0[\phi^*_{P}] - \psi_{P_0}\big| \to 0,
\end{equation*}
which completes the proof.
\end{proof}

\paragraph{Proof of \cref{thm:bvm_single_time}}
Here we prove that, under \cref{ass:time-varying-conditions}, the one-step posterior for the mean counterfactual outcome under stochastic intervention $Q$ satisfies the semiparametric BvM theorem; \cref{thm:bvm_single_time} is the case $T = 1$. Our proof proceeds by verifying the conditions of \cref{ass:yiu-main}, namely, no second-order bias, $L_2$-convergence, and the requisite empirical process conditions.

\begin{proof}
We denote the class of uncentered EIF differences as 
\begin{equation} \label{eq:uncentered-function-class}
    \mathcal{F}^*_n := \{\phi^*_P - \phi^*_{P_0}: P \in \mathcal{S}_n\}.
\end{equation}
\cref{ass:yiu-main}(a) follows from \cref{lem:no-snd-bias}, under \cref{ass:time-varying-conditions}(b) and (c). We show \cref{ass:yiu-main}(b) and (c) by verifying the conditions of \cref{prop:emp-process}, with nuisance functions $\eta_P = \{\pi_{P,t}, m_{P,t}(\cdot, a) : t \le T, a \in \{0,1\}\}$.

For \cref{prop:emp-process}(a), the sieve and truth bounds in \cref{ass:time-varying-conditions}(c)(i)--(ii), together with $\pi_{P,t} \in [0,1]$, give $\sup_{P \in \mathcal{S}_n} d(P, P_0) \le 2C + 1$.

For \cref{prop:emp-process}(b), \cref{lem:Lipschitz_multiple_time} gives $|\phi^*_P - \phi^*_{P'}| \le M\, d(P,P')$ for all $P, P' \in \mathcal{S}_n$, with $M = C'(1 + |Y|)$ for a constant $C' \ge 1$ depending only on $C$ and $T$. The same bounds give $|\phi^*_{P_0}| \lesssim 1 + |Y|$, so enlarging $C'$ if necessary, $|\phi^*_{P_0}| \le M$. Since $P_0|Y|^\kappa \le C$, $P_0[M^\kappa] < \infty$.

For \cref{prop:emp-process}(c), taking the $L_2(P_0)$ norm in \cref{lem:Lipschitz_multiple_time} with $P' = P_0$ and using the triangle inequality gives
\begin{equation*}
   \sup_{P \in \mathcal{S}_n} \|\phi^*_P - \phi_{P_0}^* \|_2 \lesssim \sum_{t=1}^T \max_a \|m_{P,t}(\cdot,a) - m_{P_0,t}(\cdot,a) \|_2 + \sum_{t=1}^T \|(1 + |Y|)(\pi_{P,t} - \pi_{P_0,t})\|_2.
\end{equation*}
Since $(H_t, A_t)$ is a function of $(H_T, A_T)$, the tower property and \cref{ass:time-varying-conditions}(c)(i) give $P_0[Y^2 \mid H_t, A_t] \le C$, so
\begin{equation*}
    \|(1 + |Y|)(\pi_{P,t} - \pi_{P_0,t})\|_2^2 
    = P_0 \Big[ P_0 \big[(1 + |Y|)^2 \mid H_t, A_t \big] (\pi_{P,t} - \pi_{P_0,t})^2 \Big]
    \le 2(1 + C) \|\pi_{P,t} - \pi_{P_0,t}\|_2^2.
\end{equation*}
By \cref{ass:time-varying-conditions}(b), it follows that $\sup_{P \in \mathcal{S}_n} \|\phi^*_P - \phi^*_{P_0}\|_2 \lesssim \sum_{t=1}^T (\rho_{t,n} + \varepsilon_{t,n}) \asymp \gamma_n$, which is \cref{prop:emp-process}(c).

Finally, \cref{prop:emp-process}(d) is \cref{ass:time-varying-conditions}(d), since the nuisance classes $\mathscr{H}_{k,n}$ are the classes $\mathscr{P}_{t,n}$ and $\{m_{P,t}(\cdot,a) - m_{P_0,t}(\cdot,a) : P \in \mathcal{S}_n\}$, whose sup-norm covering numbers are bounded by those of $\mathscr{M}_{t,n}$.

Hence \cref{prop:emp-process} gives \cref{ass:yiu-main}(b) and (c). Together with (a), all conditions of \cref{ass:yiu-main} hold, and the result follows from the main theorem of \citet{Yiu2025pcorrect}.
\end{proof}

\subsection{Proofs of Nuisance Function Conditions}
We first provide \cref{lem:entrp-int-bound}, which is a standard entropy integral bound under a common bound on the sup-norm log-covering number. The proof is not novel, but it is included for completeness. Then in \cref{lem:gamma-entrp-condition}, we translate this into a verification of \cref{ass:time-varying-conditions} (d). This subsection ends with the proof of \cref{prop:rate-condition}.

\begin{lemma}
\label[lemma]{lem:entrp-int-bound}
Suppose $\gamma_n \le 1$ and there exist $A_n \ge 1$ and $B_n \ge \max(1, e\gamma_n)$ such that
\begin{equation*}
    \log N(\varepsilon, \mathcal{F}_n, \|\cdot\|_\infty) \lesssim A_n \log (B_n / \varepsilon)
    \quad \text{for all } \varepsilon \in (0, \gamma_n].
\end{equation*}
Then
\begin{equation*}
    J_\infty(\gamma_n, \mathcal{F}_n) \lesssim \gamma_n \sqrt{A_n \log(B_n/\gamma_n)}.
\end{equation*}
In particular, if $\gamma_n \ge B_n^{-c}$ for a constant $c > 0$, then
$J_\infty(\gamma_n, \mathcal{F}_n) \lesssim \gamma_n\sqrt{A_n \log B_n}$.
\end{lemma}

\begin{proof}
Recall the definition of the entropy integral,
\begin{equation*}
    J_\infty(\gamma_n, \mathcal{F}_n) = \int_{0}^{\gamma_n} \sqrt{1 + \log N(\varepsilon, \mathcal{F}_n, \|\cdot\|_\infty)}\, d \varepsilon.
\end{equation*}
Then using $\sqrt{a+b} \le \sqrt{a} + \sqrt{b}$ for $a,b \ge 0$,
\begin{equation*}
    \int_{0}^{\gamma_n} \sqrt{1 + \log N(\varepsilon, \mathcal{F}_n, \|\cdot\|_\infty)}\, d \varepsilon 
    \le 
    \gamma_n + \int_{0}^{\gamma_n} \sqrt{\log N(\varepsilon, \mathcal{F}_n, \|\cdot\|_\infty)}\, d \varepsilon.
\end{equation*}
We then focus on the second integral. Our assumption on the log covering number gives
\begin{equation*}
    \int_{0}^{\gamma_n} \sqrt{\log N(\varepsilon, \mathcal{F}_n, \|\cdot\|_\infty)}\, d \varepsilon 
    \lesssim
    \int_{0}^{\gamma_n} \sqrt{A_n \log(B_n / \varepsilon)}\, d \varepsilon.
\end{equation*}
Then, since $B_n \ge 1$ and $\varepsilon \le \gamma_n \le 1$, both $\log B_n$ and $\log(1/\varepsilon)$ are nonnegative, so
\begin{align*}
    \int_{0}^{\gamma_n} \sqrt{A_n \log(B_n / \varepsilon)}\, d \varepsilon 
    &= \sqrt{A_n} \int_{0}^{\gamma_n} \sqrt{\log(B_n / \varepsilon)}\, d \varepsilon, \\
    &\le   \sqrt{A_n} \left[ \int_{0}^{\gamma_n} \sqrt{\log(B_n)}\, d \varepsilon + \int_{0}^{\gamma_n} \sqrt{\log(1/ \varepsilon)}\, d \varepsilon \right], \\
    &= \sqrt{A_n \log(B_n)}\,\gamma_n + \sqrt{A_n}\int_{0}^{\gamma_n} \sqrt{\log(1/ \varepsilon)}\, d \varepsilon.
\end{align*}
We then focus on the remaining integral. We use the substitution $\varepsilon = \gamma_n e^{-t^2}$, which implies that $d\varepsilon = -2t \gamma_n e^{-t^2}\, dt$. Then, since $\log(1/\gamma_n) \ge 0$, the integral becomes 
\begin{align*}
    \int_{0}^{\gamma_n} \sqrt{\log(1/ \varepsilon)}\, d \varepsilon &= \int_{0}^\infty \sqrt{\log(1/\gamma_n) + t^2}\,2t \gamma_n e^{-t^2}dt, \\
    &\le \int_{0}^\infty \left[\sqrt{\log(1/\gamma_n)} + t \right] 2t \gamma_n e^{-t^2} dt, \\
    &= \gamma_n \sqrt{\log(1/\gamma_n)} \int_{0}^\infty 2t e^{-t^2}dt + \gamma_n \int_{0}^\infty 2t^2 e^{-t^2} dt.
\end{align*}
The first integral is the pdf for Weibull($\gamma =2, \beta = 1$), so it integrates to 1. The expectation of the Weibull distribution is $\beta^{1/\gamma} \Gamma(1 + \frac{1}{\gamma})$, which implies that the second integral is 
\begin{equation*}
    \gamma_n \int_{0}^\infty 2t^2 e^{-t^2} dt = \gamma_n \Gamma\left(\frac{3}{2}\right) = \gamma_n \frac{\sqrt{\pi}}{2}.
\end{equation*}
Therefore combining everything gives 
\begin{align*}
     J_\infty(\gamma_n, \mathcal{F}_n) & \lesssim \gamma_n + \gamma_n \sqrt{A_n \log(B_n)}  +  \gamma_n \sqrt{A_n \log(1/\gamma_n)} + \gamma_n \frac{\sqrt{A_n \pi}}{2} , \\
     & \lesssim \gamma_n \sqrt{A_n \log(B_n/\gamma_n)},
\end{align*}
where the last step uses $\sqrt{\log B_n} + \sqrt{\log(1/\gamma_n)} \le \sqrt{2 \log(B_n/\gamma_n)}$, and $A_n \ge 1$ and $\log(B_n/\gamma_n) \ge 1$ to absorb the remaining terms. Hence we have shown the desired result. 
\end{proof}

To translate the bound in \cref{lem:entrp-int-bound} into a verification of \cref{ass:time-varying-conditions}(d), we have the following lemma.
\begin{lemma}
\label[lemma]{lem:gamma-entrp-condition}
Suppose $\gamma_n \le 1$ and there exist sequences $A_n \ge 1$ and $B_n \ge \max(1, e\gamma_n)$ such that 
\begin{equation*}
    \log N(\varepsilon, \mathscr{H}_{k,n}, \| \cdot \|_\infty) \lesssim A_n \log (B_n / \varepsilon), \qquad \varepsilon \in (0, \gamma_n],
\end{equation*}
for each nuisance class $\mathscr{H}_{k,n}$. Then if 
\begin{equation*}
    A_n \log (B_n / \gamma_n) = o\left(\gamma_n^{-2} \wedge n^{\frac{2(\kappa - 2)}{4\kappa -4}}\right),
\end{equation*}
\cref{ass:time-varying-conditions}(d) holds.
\end{lemma}

\begin{proof}
Recall that we must show 
\begin{equation*}
\sum_{k = 1}^K J_\infty (\gamma_n, \mathscr{H}_{k,n}) = o\left(1 \wedge \gamma_n n^{\frac{\kappa - 2}{4 \kappa -4}}\right),
\end{equation*}
where the number of classes is fixed. Then given the log-covering number bound, \cref{lem:entrp-int-bound} gives
\begin{equation*}
    J_\infty (\gamma_n, \mathscr{H}_{k,n}) \lesssim \gamma_n \sqrt{A_n \log(B_n / \gamma_n)}. 
\end{equation*}
Therefore we must show that 
\begin{equation*}
    \gamma_n \sqrt{A_n \log(B_n / \gamma_n)} = o\left(1 \wedge \gamma_n n^{\frac{\kappa - 2}{4 \kappa -4}}\right).
\end{equation*}
We will proceed using cases. 
\paragraph{Case 1:} $o\left(1 \wedge \gamma_n n^{\frac{\kappa - 2}{4 \kappa -4}}\right) = o(1)$. Then this implies $A_n \log (B_n / \gamma_n) = o(\gamma_n^{-2})$. 
\paragraph{Case 2:} $o\left(1 \wedge \gamma_n n^{\frac{\kappa - 2}{4 \kappa -4}}\right) = o(\gamma_n n^{\frac{\kappa - 2}{4 \kappa -4}})$. Then this implies $A_n \log (B_n / \gamma_n) = o\left(n^{\frac{2(\kappa - 2)}{4\kappa -4}}\right).$ 
Combining the two cases gives 
\begin{equation*}
     A_n \log (B_n / \gamma_n) = o\left(\gamma_n^{-2} \wedge n^{\frac{2(\kappa - 2)}{4\kappa -4}}\right),
\end{equation*}
which is the desired result.
\end{proof}

Combining \cref{lem:gamma-entrp-condition} with the polynomial growth of the sieve sequences yields \cref{prop:rate-condition}, which we prove in its longitudinal form under \cref{ass:time-varying-conditions}. The main-text statement is the case $T = 1$.

\paragraph{Proof of \cref{prop:rate-condition}}
\begin{proof}
Replacing $A_n$ by $\max(A_n, 1)$ and $B_n$ by $\max(B_n, e)$ only weakens the covering bound, and $\gamma_n \le 1$ for $n$ large, so the hypotheses of \cref{lem:gamma-entrp-condition} hold. By \cref{lem:gamma-entrp-condition}, it suffices for (d) to show
\begin{equation*}
 A_n \log (B_n / \gamma_n) = o\left(\gamma_n^{-2} \wedge n^{\frac{2(\kappa - 2)}{4\kappa -4}}\right).
\end{equation*}
Since $B_n \lesssim n^c$ and $\gamma_n \gtrsim n^{-c}$, we have $\log(B_n/\gamma_n) \lesssim \log n$. Together with $A_n \lesssim n\gamma_n^2$, it therefore suffices to show
\begin{equation*}
    n\gamma_n^2 \log n = o\left(\gamma_n^{-2} \wedge n^{\frac{2(\kappa - 2)}{4\kappa -4}}\right).
\end{equation*}
We proceed by cases.
\paragraph{Case 1:} $o\left(\gamma_n^{-2} \wedge n^{\frac{2(\kappa - 2)}{4\kappa -4}}\right) = o\left( n^{\frac{2(\kappa - 2)}{4\kappa -4}}\right)$. This implies that
\begin{equation*}
    \gamma_n = o\left(n^{-\frac{\kappa}{4(\kappa-1)}}(\log n)^{-1/2} \right).
\end{equation*}
\paragraph{Case 2:} $o\left(\gamma_n^{-2} \wedge n^{\frac{2(\kappa - 2)}{4\kappa -4}}\right) = o(\gamma_n^{-2})$. This implies that
\begin{equation*}
    \gamma_n = o \left( n^{-1/4} (\log n)^{-1/4}\right).
\end{equation*}
Since $\frac{\kappa}{4(\kappa -1)} > \frac{1}{4}$, the first condition implies the second, and hence (d) holds.

Also, since $\gamma_n \asymp \sum_t(\rho_{t,n}+\varepsilon_{t,n})$, we have $\rho_{t,n}(\rho_{t,n}+\varepsilon_{t,n}) \lesssim \gamma_n^2 = o\big(n^{-\kappa/(2\kappa-2)}\big) = o(n^{-1/2})$ for each $t$, since $\kappa/(2\kappa-2) > 1/2$. Hence \cref{ass:time-varying-conditions}(b) also holds and we've shown the desired result.
\end{proof}



\subsection{SoftBART Prior \& Sieve}
\begin{assumption}
\label{ass:SoftBart_prior}
 (SoftBART Prior) Assume that the gating function $\varsigma$ satisfies Assumption G of \citet{Linero2018brte}. In particular, it is continuously differentiable with values in $[0,1]$, with Lipschitz constant $L_\varsigma:= \|\varsigma'\|_\infty < \infty$. All constants below are positive, with $C_{\tau,2}, C_{\tau,4} \in (0,1)$ and $C_{S,2} > 1$. Further assume the following:
\begin{enumerate}
    \item[(a)] \textit{Number of trees.} For all integers $m \ge 0$, 
    \begin{equation*}
        \Pi(K = m) \ge C_{k,1} \exp\{-C_{k,2}m\}, \quad 
        \Pi(K \ge m) \le C_{k,1} \exp\{-C_{k,3}m\}.
    \end{equation*}
    \item[(b)] \textit{Bandwidth.} All branches of all trees share one bandwidth $\tau$. For sufficiently large $x$, 
    \begin{equation*}
        \Pi(\tau \ge x) \le C_{\tau,1} \exp\{-x ^{C_{\tau,2}} \}, \quad 
        \Pi(\tau^{-1} \ge x) \le C_{\tau,3} \exp\{-x ^{C_{\tau,4}} \},
    \end{equation*}
    and the density of $\tau^{-1}$ satisfies
    \begin{equation*}
        \pi_{\tau^{-1}}(x) \ge C_{\tau,5} \exp\{-C_{\tau,6} x\}.
    \end{equation*}
    \item[(c)] \textit{Splitting Rules.} Let 
    \begin{equation*}
        S \sim \mathrm{Dir}(C_{S,1}/p^{C_{S,2}}, \ldots, C_{S,1}/p^{C_{S,2}}).
    \end{equation*}
    Given $S$, the coordinates $j_b$ are i.i.d. Categorical($S$). Given $j_b$, the split location is 
    \begin{equation*}
        c_b \sim \text{Uniform}(u_b^-, u_b^+),
    \end{equation*}
    where $(u_b^-, u_b^+) \subseteq [0,1]$ is the range of $x_{j_b}$ in the cell defined by the path from the root to $b$. This means no split is redundant.
    \item[(d)] \textit{Leaf Values.} The $\ell_{k,l}$ are i.i.d. with density $\pi_{\ell}$, where 
    \begin{align*}
        \pi_{\ell}(x) &\ge C_{\ell, 1} \exp\{-C_{\ell,2} |x| \}, \quad \text{for all }x, \\
        \Pi(|\ell| \ge x) &\le C_{\ell, 3} \exp\{- x^{C_{\ell,4}} \}, \quad \text{for all }x \ge 0.
    \end{align*}
    \item[(e)] \textit{Depth.} Each tree's depth $D_k$ satisfies $\Pi(D_k > d_0) = 0$ for some $d_0 \ge 2d$, and 
    \begin{equation*}
        \Pi(D_k = m) > 0
    \end{equation*}
    for $m = 0,1, \ldots, 2d$.
\end{enumerate}
\end{assumption}

\begin{definition}
\label[definition]{def:SoftBART-Sieve}
(SoftBART Sieve) Let $\bar \epsilon_n$ be the SoftBART contraction rate and fix an integer $\bar D \ge d_0$. Let $\tilde{\mathcal{F}}_n$ be the set of functions $f$ of the form in \cref{ass:SoftBart_prior} that satisfy: 
\begin{enumerate}
    \item[(a)]\textit{Number of trees.} Let $K \le K_n$ for $K_n := C_{K}n\bar \epsilon_n^2$,
    \item[(b)]\textit{Depth.} Every tree has depth at most $\bar D$, so it has at most $\Lambda := 2^{\bar D}$ leaves,
    \item[(c)]\textit{Bandwidth.} All branches share a bandwidth $ \tau \in [\tau_n^L, \tau_n^U]$, for 
    \begin{equation*}
        (\tau_n^L)^{-1} = (C_4 n \bar \epsilon_n^2)^{1/C_{\tau,4}}, \quad 
        (\tau_n^U) = (C_3 n \bar \epsilon_n^2)^{1/C_{\tau,2}},
    \end{equation*}
    \item[(d)]\textit{Splitting coordinates.} At most
    \begin{equation*}
        d_{\text{max}} = \max\{d, \lfloor C_{1}n \bar \epsilon_n^2 / \log p \rfloor \},
    \end{equation*}
    distinct coordinates are used for splitting, 
    \item[(e)]\textit{Leaf values.} Every leaf value lies in $[-\bar\ell_n, \bar\ell_n]$ with 
    \begin{equation*}
        \bar\ell_n = (C_5 n \bar \epsilon_n^2)^{1/C_{\ell,4}}.
    \end{equation*}
\end{enumerate}
\end{definition}

\subsection{SoftBART Entropy \& Contraction}

This subsection proves \cref{lem:SoftBART-covering}, the sup-norm covering bound for the SoftBART sieve, and \cref{lem:SoftBART-contraction}, which gives random-design posterior contraction for the outcome regression and probit propensity score. First, we derive a bound on the sup-norm covering number for SoftBART, given in \cref{lem:SoftBART-covering}. This proof is very similar to that of \citet{Linero2018brte}, except we leave the radius general to allow for the calculation of the entropy integral. 

\paragraph{Proof of \cref{lem:SoftBART-covering}}

\begin{proof}
We note that a SoftBART function in the sieve, with $K \le K_n$ trees, can be written as 
\begin{equation*}
    f(x) = \sum_{k = 1}^{K_n} g_k(x; \mathcal{K}_k, \mathcal{L}_k), 
\end{equation*}
for the tree structures $\mathcal{K}_k$ and leaf parameters $\mathcal{L}_k$, where trees beyond $K$ are taken to have zero leaf values. If each of the $K_n$ trees is within $\varepsilon/K_n$ of a center in sup-norm, then the sum is within $\varepsilon$ of the sum of the centers. Therefore, 
\begin{equation*}
    N\left(\varepsilon, \tilde{\mathcal{F}}_n, \|\cdot\|_\infty \right) \le \prod_{k=1}^{K_n} N\left(\varepsilon/K_n, \mathcal{G}_n, \|\cdot\|_\infty \right),
\end{equation*}
where $\tilde{\mathcal{F}}_n$ is the SoftBART sieve of \cref{def:SoftBART-Sieve} and $\mathcal{G}_n$ is the class of possible functions for a single tree. For now we focus on covering across the gates and leaf values for a single tree. Note that a single SoftBART tree can be expressed as
\begin{equation*}
    g(x; \mathcal{K}, \mathcal{L}) = \sum_{l = 1}^{\Lambda} \omega_l(x) \ell_l,
\end{equation*}
for the gates $\omega_l$ and leaf values $\ell_l$ for $\Lambda$ leaves. The gates 
\begin{equation*}
    \omega_l(x) = \prod_{b \in \text{path}(l)}
    \begin{cases}
        \varsigma \left(\frac{x_{j_b} - c_b}{\tau} \right) & \text{if } l \text{ is to the left of } b \\
        1 - \varsigma \left(\frac{x_{j_b} - c_b}{\tau} \right) & \text{if } l \text{ is to the right of } b
    \end{cases}
\end{equation*}
are probabilities, so they satisfy $\sum_l \omega_l(x) = 1$ for all $x$. Let $\Omega_n$ denote the class of gate vectors $\omega = (\omega_1, \dots, \omega_\Lambda)$ for a fixed tree configuration, and let $\mathcal{V}_n = [-\bar\ell_n, \bar\ell_n]^{\Lambda}$ denote the leaf values. Fix $g(x) = \sum_{l=1}^{\Lambda} \ell_l \omega_l(x) \in \mathcal{G}_n$, and let $\tilde{\omega} \in \Omega_n$ and $\tilde{\ell} \in \mathcal{V}_n$ satisfy
\begin{equation*}
    \max_l \|\omega_l - \tilde{\omega}_l\|_\infty \le \beta_\omega,
    \qquad
    \max_l |\ell_l - \tilde{\ell}_l| \le \beta_\ell.
\end{equation*}
Define $\tilde{g}(x) = \sum_{l=1}^{\Lambda} \tilde{\ell}_l \tilde{\omega}_l(x)$.
Using $ab - cd = (a - c)b + c(b - d)$, for every $x$,
\begin{align*}
    |g(x) - \tilde{g}(x)|
    &\le \sum_{l=1}^{\Lambda} |\ell_l - \tilde{\ell}_l|\, \omega_l(x)
       + \sum_{l=1}^{\Lambda} |\tilde{\ell}_l|\, |\omega_l(x) - \tilde{\omega}_l(x)| \\
    &\le \beta_\ell \sum_{l=1}^{\Lambda} \omega_l(x) + \bar\ell_n \Lambda \beta_\omega \\
    &= \beta_\ell + \bar\ell_n \Lambda \beta_\omega,
\end{align*}
where the second inequality uses the leaf-value bound in \cref{def:SoftBART-Sieve}(e) and the equality uses $\sum_l \omega_l(x) = 1$. Taking the supremum over $x$,
\begin{equation*}
    \|g - \tilde{g}\|_\infty \le \beta_\ell +  \bar\ell_n \Lambda \beta_\omega.
\end{equation*}
Hence choosing
\begin{equation}\label{eq:beta_definitions}
    \beta_\ell = \frac{\varepsilon}{2K_n},
    \qquad
    \beta_\omega = \frac{\varepsilon}{2 \bar\ell_n \Lambda K_n}
\end{equation}
gives $\|g - \tilde{g}\|_\infty \le \varepsilon / K_n$, so for a fixed configuration
\begin{equation*}
    N\left(\varepsilon / K_n, \mathcal{G}_n, \|\cdot\|_\infty \right)
    \le
    N\left(\beta_\ell, \mathcal{V}_n, \|\cdot\|_\infty \right)
    N\left(\beta_\omega, \Omega_n, \|\cdot\|_\infty \right).
\end{equation*}

Then, we focus on obtaining a bound on the covering number for the splits. A SoftBART splitting function can be expressed as 
\begin{equation*}
    s(c,\tau)
    =
    \varsigma\left(
    \frac{x_j-c}{\tau}
    \right).
\end{equation*}
At this point we consider arbitrary $x_j$ parameters. Let $\Theta_n := [0,1] \times [\tau_n^L, \tau_n^U]$. Further, the class of splitting functions at a single internal split for some $x_j$ can be expressed as 
\begin{equation*}
    \Omega_{\text{split},n}
    =\left\{ s(c,\tau)
    =
    \varsigma\left(
    \frac{x_j-c}{\tau}
    \right)
    :
    (c,\tau)
    \in
    \Theta_n
    \right\}.
\end{equation*}

Since $\varsigma$ is continuously differentiable with uniformly bounded derivative, the partial derivatives of the splitting function satisfy
\begin{align*}
    \left|
    \frac{\partial s(c,\tau)}{\partial c}
    \right|
    &=
    \left| \frac{1}{\tau}
    \varsigma'\left( \frac{x_j-c}{\tau} \right)\right|
    \lesssim
    \frac{1}{\tau_n^L},
    \\
    \left|
    \frac{\partial s(c,\tau)}{\partial \tau}
    \right|
    &=
    \left| \frac{x_j-c}{\tau^2} \varsigma' \left( \frac{x_j-c}{\tau} \right) \right|
    \lesssim
    \frac{1}{(\tau_n^L)^2},
\end{align*}
where the second inequality follows from the bounded support of
$x_j$ and $c$ after normalization. Therefore, by the mean value theorem, the splitting function is Lipschitz in the supremum norm:
\begin{equation*}
    \|s(c_1,\tau_1)-s(c_2,\tau_2)\|_{\infty}
    \lesssim
    \frac{|c_1-c_2|}{\tau_n^L}
    +
    \frac{|\tau_1-\tau_2|}{(\tau_n^L)^2}.
\end{equation*}

To obtain a $\beta_\omega$-cover of $\Omega_{\text{split},n}$ in
$\|\cdot\|_\infty$, place a grid on $\Theta_n$ with spacings
\begin{equation*}
    \Delta_c \asymp \beta_\omega \,\tau_n^L,
    \qquad
    \Delta_\tau \asymp \beta_\omega \,(\tau_n^L)^2,
\end{equation*}
with all grid points lying in $\Theta_n$. For any $(c,\tau) \in \Theta_n$ there is a grid point $(c_k,\tau_l)$ with $|c - c_k| \le \Delta_c$ and $|\tau - \tau_l| \le \Delta_\tau$, so by the Lipschitz bound
\begin{equation*}
    \|s(c,\tau) - s(c_k,\tau_l)\|_\infty
    \lesssim
    \frac{\Delta_c}{\tau_n^L} + \frac{\Delta_\tau}{(\tau_n^L)^2}
    \lesssim \beta_\omega .
\end{equation*}
Choosing the implicit constants in $\Delta_c$ and $\Delta_\tau$ small enough makes the right-hand side at most $\beta_\omega$. Since split locations lie in $[0,1]$, the number of grid points needed for the split location is at most
\begin{equation*}
    \left\lceil \frac{1}{\beta_\omega \,\tau_n^L} \right\rceil,
\end{equation*}
and for the bandwidth at most
\begin{equation*}
    \left\lceil \frac{\tau_n^U-\tau_n^L}{\beta_\omega \,(\tau_n^L)^2} \right\rceil.
\end{equation*}
Since $\tau_n^U-\tau_n^L\lesssim \tau_n^U$, this gives
\begin{equation*}
    N\left(
    \beta_\omega,
    \Omega_{\text{split},n},
    \|\cdot\|_\infty
    \right)
    \lesssim
    \frac{
    \tau_n^U
    }
    {
    \beta_\omega^2
    (\tau_n^L)^3
    }.
\end{equation*}

A gate corresponding to a leaf can be written as a product of at most $\bar D \le \Lambda - 1$ splitting functions or their complements,
\begin{equation*}
    \omega_l(x) = \prod_{b \in \text{path}(l)} \tilde{s}_b(x),
\end{equation*}
where $\tilde s_b \in \{s_b, 1 - s_b\}$ and $0 \le \tilde{s}_b(x) \le 1$. Since $\|(1-s) - (1-s')\|_\infty = \|s - s'\|_\infty$, a cover of the splitting functions also covers their complements. By a telescoping product inequality, for $m \le \Lambda - 1$ factors,
\begin{equation*}
    \left\| \prod_{b=1}^m \tilde{s}_b - \prod_{b=1}^m \tilde{s}'_b\right\|_{\infty} \le \sum_{b=1}^m \|\tilde{s}_b - \tilde{s}'_b\|_{\infty}.
\end{equation*}
Hence a $\beta_\omega / \Lambda$-covering of each splitting function induces a $\beta_\omega$-covering of each gate. Since a fixed configuration has at most $\Lambda - 1$ internal splits, we have the bound
\begin{equation*}
     N\left(
    \beta_\omega,
    \Omega_n,
    \|\cdot\|_\infty
    \right) \le
     N\left(
    \beta_\omega / \Lambda,
    \Omega_{\text{split},n},
    \|\cdot\|_\infty
    \right)^{\Lambda}
    \lesssim
    \left(
    \frac{
    \tau_n^U \Lambda^2
    }
    {
    \beta_\omega^2
    (\tau_n^L)^3
    }
    \right)^{\Lambda}.
\end{equation*}

Next, the space of possible leaf values is a $\Lambda$-dimensional hypercube, $\mathcal{V}_n = [-\bar\ell_n, \bar\ell_n]^{\Lambda}$. Therefore for a single leaf the number of covers needed is at most $2\bar\ell_n/\beta_\ell$. Across $\Lambda$ leaves we have
\begin{equation*}
    N\left(\beta_\ell, \mathcal{V}_n, \|\cdot\|_\infty \right) \lesssim
    \left(\frac{2\bar\ell_n}{\beta_\ell} \right) ^{\Lambda}.
\end{equation*}
Combining these results we get a bound on the covering number for a single SoftBART tree with a fixed configuration, 
\begin{equation*}
    N\left(\varepsilon/K_n, \mathcal{G}_n, \|\cdot\|_\infty \right) \lesssim \left(
    \frac{2 \bar\ell_n \tau_n^U \Lambda^2}
    {\beta_\omega^2 \beta_\ell(\tau_n^L)^3}
    \right)^{\Lambda}.
\end{equation*}

Next, we multiply this bound across all possible tree configurations. A configuration is defined by the tree topology (the branching structure) and the assignment of split variables to each internal node. Every tree of depth at most $\bar D$ represents the same function as some complete binary tree of depth $\bar D$, which has $\Lambda = 2^{\bar D}$ leaves. That is, splitting a leaf with value $\ell$ into two children that both take value $\ell$ leaves the function unchanged for any choice of split variable and any $(c,\tau) \in \Theta_n$, since $\ell\, s(c,\tau) + \ell(1-s(c,\tau)) = \ell$. Applying this to every leaf at depth less than $\bar D$, at most $\bar D$ times, yields a complete tree. It therefore suffices to count configurations of the complete tree of depth $\bar D$. Its topology is fixed, and each of its $\Lambda - 1$ internal nodes is assigned one of $p$ covariates, so the number of discrete configurations per tree is at most
\begin{equation*}
    p^{\Lambda - 1} \le p^{\Lambda}.
\end{equation*}
Therefore we can bound the covering number for a single SoftBART tree across all configurations by 
\begin{equation*}
    N\left(\varepsilon/K_n, \mathcal{G}_n, \| \cdot \|_\infty\right) \lesssim 
    p^{\Lambda} \left(
    \frac{2 \bar\ell_n \tau_n^U \Lambda^2}{\beta_\omega^2 \beta_\ell(\tau_n^L)^3} \right)^{\Lambda}.
\end{equation*}
Hence across all $K_n$ trees we have the bound
\begin{equation*}
    N\left(\varepsilon, \tilde{\mathcal{F}}_n, \|\cdot \|_\infty \right) \lesssim  \left(\frac{2 \,p \,\bar\ell_n \,\tau_n^U \Lambda^2}{\beta_\omega^2 \, \beta_\ell \,(\tau_n^L)^3} \right)^{\Lambda K_n}.
\end{equation*}
Using our previous definition of $\beta_\omega$ and $\beta_\ell$ in \eqref{eq:beta_definitions} and taking logarithms,
\begin{equation*}
    \log N\left(\varepsilon, \tilde{\mathcal{F}}_n, \|\cdot\|_\infty \right) \le \Lambda K_n \log\left(\frac{16 \,p\, \tau_n^U \, \bar\ell_n^3 \,\Lambda^4 \, K_n^3}{\varepsilon^3 \,(\tau_n^L)^3 }\right) \lesssim 
    K_n \log \left(\frac{\tau_n^U \, \bar\ell_n \, K_n} {\varepsilon \, \tau_n^L} \right),
\end{equation*}
since $\Lambda$ and $p$ are fixed, $\tau_n^U \ge 1$, and the argument of the logarithm is bounded away from one for $\varepsilon$ small. This is the desired result.
\end{proof}

We will now prove SoftBART contraction. Recall that $\pi_{P_0}(x) = P_0(A = 1 \mid X=x) = \Phi(f_0(x))$ and $\mu_{P_0}(x,a) = P_0[Y \mid X=x, A=a]$. Let $\bar \epsilon_n = n^{-\alpha / (2\alpha + d)}(\log n)^t + \sqrt{n^{-1}d \log p}$ be the rate of \citet{Linero2018brte} Theorem 2, and let the clipped draws be 
\begin{equation*}
    T_C(v) = \max\{ -C, \min\{ C, v\}\}.
\end{equation*}
We assume the following.
\begin{assumption}
\label{ass:prior-assumptions}
Suppose that the following hold:
\begin{enumerate}
    \item[(a)] \textit{Smoothness.} Suppose $f_0 \in C^{\alpha,R}([0,1]^p)$ and depends on at most $d$ coordinates. Also, $\mu_{P_0} \in C^{\alpha,R}([0,1]^{p+1})$ and depends on at most $d$ coordinates. Therefore, $\|f_0\|_\infty \le R$, which means $\pi_{P_0} \in [c_\pi, 1-c_\pi]$ with $c_\pi = \Phi(-R)$, and $\|\mu_{P_0}\|_\infty < \infty$.
    \item[(b)] \textit{Errors.} Suppose $Y \mid X,A \sim N(\mu_{P_0}(X,A), \sigma^2_0)$, with $\sigma_0$ known.
    \item[(c)] \textit{Prior \& Sieve.} Suppose \cref{ass:SoftBart_prior} holds for the prior and that the sieve is constructed as in \cref{def:SoftBART-Sieve}.
\end{enumerate}
\end{assumption}

Let $p$ and $q$ be probability densities with respect to a common $\sigma$-finite measure $\mu$. The \emph{Hellinger distance} is
\begin{equation*}
    d_H(p, q) = \left( \int \big( \sqrt{p} - \sqrt{q} \big)^2 \, d\mu \right)^{1/2},
\end{equation*}
which is a metric on densities, bounded by $\sqrt{2}$, and does not depend on the choice of $\mu$. For a true density $p_0$, the \emph{Kullback--Leibler divergence} and the second \emph{Kullback--Leibler variation} are
\begin{equation*}
    K(p_0; p) = \int p_0 \log \frac{p_0}{p} \, d\mu,
    \qquad
    V_{2,0}(p_0; p) = \int p_0 \left| \log \frac{p_0}{p} - K(p_0; p) \right|^2 d\mu,
\end{equation*}
with the convention that both equal $+\infty$ if $p_0$ is not absolutely continuous with respect to $p$. They enter only through the prior-mass condition, via the Kullback--Leibler neighborhoods
\begin{equation*}
    B_2(p_0, \epsilon) = \big\{ p : K(p_0; p) < \epsilon^2, \ V_{2,0}(p_0; p) < \epsilon^2 \big\}.
\end{equation*}
This leads to a general SoftBART contraction lemma.

\begin{lemma}
\label[lemma]{lem:gen-contract}
Let $Z_1, \ldots, Z_n$ be iid from $p_0 = p_{g_0}$ from a model $\{p_g\}$, where $g$ carries a SoftBART prior satisfying \cref{ass:prior-assumptions}(c), and \cref{ass:prior-assumptions}(a) and (b) hold. Further assume 
\begin{enumerate}
    \item[(g1)] $K(p_0;p_g)$ and $V_{2,0}(p_0;p_g)$ are both at most $C_1\|g - g_0\|_\infty^2$ whenever $\|g - g_0\|_\infty \le 1$.
    \item[(g2)] $d_H(p_g, p_{g'}) \le L \|g - g'\|_\infty$ for all $g,g'$.
\end{enumerate}
Then for $\epsilon_n = K \bar \epsilon_n \sqrt{\log n}$, with $K$ and $M$ large, in $P_0$-probability,
\begin{equation*}
    \Pi(d_H(p_g, p_0) > M \epsilon_n \mid Z_{1:n}) \to 0, \quad \Pi(\tilde{\mathcal{F}}_n^c \mid Z_{1:n}) \to 0.
\end{equation*}
\end{lemma}
\begin{proof}
    We will apply \citet{ghosal2017fundamentals} Theorem 8.9 with $d = d_H$, $\mathcal{P}_{n,1} = \{p_g : g \in \tilde{\mathcal{F}}_n\}$, $\mathcal{P}_{n,2} = \mathcal{P} \setminus \mathcal{P}_{n,1}$, and rates $\bar \epsilon_n \le \epsilon_n$.

    \paragraph{8.4 (Prior mass).} \citet{Linero2018brte} Theorem 2 gives 
    \begin{equation*}
        \Pi(\|g - g_0\|_\infty \le A \bar \epsilon_n) \ge \exp (-C n \bar \epsilon_n^2).
    \end{equation*}
    Then once $A \bar \epsilon_n \le 1$, (g1) places this ball in $B_2(p_0, A\sqrt{C_1} \bar \epsilon_n)$. Thus rescaling by a constant gives (8.4). 

    \paragraph{8.6 (Sieve mass).} By \citet{Linero2018brte} Lemma 1, with sieve constants chosen large enough,
    \begin{equation*}
        \Pi(\tilde{\mathcal{F}}_n^c) \le \exp(-(C+4)n \bar \epsilon_n^2),
    \end{equation*}
    so (8.6) holds.

    \paragraph{8.5 (Entropy).} By (g2) and \cref{lem:SoftBART-covering}, for large $K$,
    \begin{align*}
        \log N(\xi \epsilon_n, \mathcal{P}_{n,1}, d_H) &\le \log N (\xi \epsilon_n / L, \tilde{\mathcal{F}}_n, \|\cdot \|_\infty), \\
        &\lesssim n \bar \epsilon_n^2 \log n, \\
        &\le n \epsilon_n^2.
    \end{align*}
    Hence \citet{ghosal2017fundamentals} Theorem 8.9 gives the first claim. 

    For the second claim, (8.4) and (8.6) give 
    \begin{equation*}
        \frac{\Pi(\mathcal{P}_{n,2})}{\Pi(B_2(p_0, \bar \epsilon_n))} \le e^{-4n \bar \epsilon_n^2} = o(e^{-2n \bar \epsilon_n^2}),
    \end{equation*}
    so $\Pi(\mathcal{P}_{n,2} \mid Z_{1:n}) \to 0$ in $P_0$-probability by Lemma 1 of \citet{ghosal2007noniid}. Hence we have shown the desired results.
\end{proof}

Now we will specialize the above general SoftBART contraction lemma to the outcome and propensity score regressions.

\begin{lemma}[Outcome regression]
\label[lemma]{lem:outcome-regression}
Let $p_m$ be the density of $(X, A, Y)$ relative to $P_{0,(X,A)} \otimes \lambda$, where $\lambda$ denotes Lebesgue measure on $\mathbb{R}$, under $Y \mid X, A \sim N\big(m(X,A), \sigma_0^2\big)$, and let $p_0 = p_{\mu_{P_0}}$ with $\|\mu_{P_0}\|_\infty \le C$. Then
\begin{enumerate}
    \item[(i)] the model satisfies \textup{(g1)} and \textup{(g2)} of \cref{lem:gen-contract} with $g = m$;
    \item[(ii)] $\|T_C(m) - \mu_{P_0}\|_{L_2(P_0)} \lesssim d_H(p_m, p_0)$ for every $m$.
\end{enumerate}
\end{lemma}

\begin{proof}
Lemma~2.7 of \citet{ghosal2017fundamentals}, with $\sigma = \tau = \sigma_0$ and covariate law $P_{0,(X,A)}$, gives
\begin{equation*}
    K(p_0; p_m) = \frac{\|m - \mu_{P_0}\|_{L_2(P_0)}^2}{2\sigma_0^2}, \qquad
    V_{2,0}(p_0; p_m) = \frac{\|m - \mu_{P_0}\|_{L_2(P_0)}^2}{\sigma_0^2}, \qquad
    d_H(p_m, p_{m'}) \le \frac{\|m - m'\|_{L_2(P_0)}}{2\sigma_0}.
\end{equation*}
Bounding $\|\cdot\|_{L_2(P_0)}$ by $\|\cdot\|_\infty$ gives (i).

For (ii), the affinity formula in the proof of Lemma~2.7 gives
\begin{equation*}
    d_H^2(p_m, p_0) = 2\, E_0\Big[ 1 - e^{-(m - \mu_{P_0})^2(X, A)/(8\sigma_0^2)} \Big].
\end{equation*}
The integrand is nondecreasing in $|m - \mu_{P_0}|$, and $|T_C(m) - \mu_{P_0}| \le |m - \mu_{P_0}|$, so $d_H(p_{T_C(m)}, p_0) \le d_H(p_m, p_0)$. Also $|T_C(m) - \mu_{P_0}| \le 2C$, so the exponent $s = (T_C(m) - \mu_{P_0})^2 / (8\sigma_0^2)$ lies in $[0, s_0]$ with $s_0 = C^2/(2\sigma_0^2)$. Since $s \mapsto 1 - e^{-s}$ is concave, $1 - e^{-s} \ge s\,(1 - e^{-s_0})/s_0$ on $[0, s_0]$. Therefore
\begin{equation*}
    \|T_C(m) - \mu_{P_0}\|_{L_2(P_0)}^2
    \le \frac{4\sigma_0^2 s_0}{1 - e^{-s_0}}\, d_H^2\big(p_{T_C(m)}, p_0\big)
    \le \frac{4\sigma_0^2 s_0}{1 - e^{-s_0}}\, d_H^2(p_m, p_0).
\end{equation*}
Therefore we have shown the desired result.
\end{proof}

\begin{lemma}[Probit propensity]
\label[lemma]{lem:probit-propensity}
Let $p_f(a \mid x) = \Phi\big(f(x)\big)^a \big\{1 - \Phi\big(f(x)\big)\big\}^{1-a}$, relative to $P_{0,X} \otimes \#$, where $\#$ denotes the counting measure on $\{0,1\}$, and let $p_0 = p_{f_0}$ with $\|f_0\|_\infty \le R$, so that $\pi_{P_0} = \Phi(f_0)$. Then
\begin{enumerate}
    \item[(i)] the model satisfies \textup{(g1)} and \textup{(g2)} of \cref{lem:gen-contract} with $g = f$;
    \item[(ii)] $\|\Phi(f) - \pi_{P_0}\|_{L_2(P_0)} \le \sqrt{2}\, d_H(p_f, p_0)$ for every $f$.
\end{enumerate}
\end{lemma}

\begin{proof}
We apply Lemma~2.8 of \citet{ghosal2017fundamentals} with the Gaussian link $H = \Phi$.

\emph{(g1).} By Lemma~2.8(iii)--(iv), with $(f, g) = (f_0, f)$ and $S(u) = |u| \vee 1$: if $\|f - f_0\|_\infty \le 1$, then $|f_0| \vee |f| \le R + 1$, so
\begin{equation*}
    K(p_0; p_f) \lesssim (R+1)\, \|f - f_0\|_\infty^2, \qquad
    V_{2,0}(p_0; p_f) \le V_2(p_0; p_f) \lesssim (R+1)^2\, \|f - f_0\|_\infty^2.
\end{equation*}

\emph{(g2).} For the Gaussian link, $I(\phi) = \int t^2 \, d\Phi(t) = 1$, so Lemma~2.8(ii) gives
\begin{equation*}
    d_H(p_f, p_{f'}) \le \|f - f'\|_{L_2(P_0)} \le \|f - f'\|_\infty.
\end{equation*}

\emph{$L_2$ bound.} Lemma~2.8(i) with $r = 2$ gives $\|p_f - p_0\|_2 = \sqrt{2}\, \|\Phi(f) - \pi_{P_0}\|_{L_2(P_0)}$. Since $p_f, p_0 \le 1$, Lemma~B.1(vi) of \citet{ghosal2017fundamentals} gives
\begin{equation*}
    \|p_f - p_0\|_2 \le \big\|\sqrt{p_f} + \sqrt{p_0}\big\|_\infty \, d_H(p_f, p_0) \le 2\, d_H(p_f, p_0).
\end{equation*}
Combining the two displays gives $\|\Phi(f) - \pi_{P_0}\|_{L_2(P_0)} \le \sqrt{2}\, d_H(p_f, p_0)$.
\end{proof}

Combining \cref{lem:gen-contract,lem:outcome-regression,lem:probit-propensity} gives \cref{lem:SoftBART-contraction}.

\paragraph{Proof of \cref{lem:SoftBART-contraction}}

\begin{proof}
By \cref{lem:outcome-regression,lem:probit-propensity}, both models satisfy \textup{(g1)}--\textup{(g2)}, so \cref{lem:gen-contract} gives contraction in Hellinger distance at rate $\epsilon_n$. The $L_2(P_0)$ bounds in the same two lemmas convert this into contraction in $L_2(P_0)$.
\end{proof}

\subsection{Additional Simulation Details}
\noindent \textbf{Further Main Simulation Details}
Estimator performance is assessed using several metrics, including integrated bias and root-mean-squared error, 
\begin{equation*}
    \widehat{\text{Bias}} = \frac{1}{I} \sum_{i = 1}^I \left|\frac{1}{J} \sum_{j=1}^J \hat{\psi}_j(\delta_i) - \psi(\delta_i) \right|, \quad  \widehat{\text{RMSE}} = \frac{\sqrt{n}}{I} \sum_{i = 1}^I \left[\frac{1}{J} \sum_{j=1}^J \left\{\hat{\psi}_j(\delta_i) - \psi(\delta_i) \right\}^2\right]^{1/2}.
\end{equation*}
Further, we assess the uniform coverage of the estimated confidence bands, 
\begin{equation*}
    \widehat{\text{Cov}} = \frac{1}{J} \sum_{j=1}^J \mathbf{1} \left\{\psi(\delta_i) \in \hat{\mathcal{C}}_j(\delta_i) \quad \forall i = 1, \ldots, I \right\},
\end{equation*}
where $\hat{\mathcal{C}}_j$ denotes the estimated uniform 95\% confidence band in simulation $j$, as well as the mean 95\% interval length,

\begin{equation*}
    \widehat{\text{Int Len}} = \frac{1}{IJ} \sum_{i = 1}^I \sum_{j=1}^J \left|\hat{\mathcal{C}}_j^{\text{upper}}(\delta_i) - \hat{\mathcal{C}}_j^{\text{lower}}(\delta_i) \right|.
\end{equation*}

\noindent \textbf{Frequentist Estimators and Uncertainty Quantification}

We have the following frequentist plug-in estimator, 
\begin{align*}
    \hat{\psi}_{\text{freq-plug}} &= \mathbb{P}_n [ \psi(\hat{\eta}, \delta)].
\end{align*}
Next is the frequentist one-step estimator,
\begin{equation*}
    \hat{\psi}_{\text{freq-step}} = \frac{1}{K}\sum_{k=1}^K \mathbb{P}_n^k[\phi^*(Z; \hat{\eta}_{-k},\delta)] = \mathbb{P}_n \{ \phi^*(Z; \hat{\eta}_{-S}, \delta)\},
\end{equation*}
where $\hat{\eta}_{-k}$ is the nuisance estimator constructed excluding group $k$. Uniform confidence intervals in the frequentist setting are obtained using the multiplier bootstrap process according to Theorem 4 in \citet{Kennedy2018ipsi}. First note that 
\begin{equation*}
    \hat{\sigma}^2(\delta) =  \mathbb{P}_n \left[\{\phi^*(Z;\hat{\eta}_{-S}, \delta) - \hat{\psi}(\delta)\}^2 \right].
\end{equation*}
Then pointwise 95\% confidence intervals for $\psi(\delta)$ are 
\begin{equation*}
    \hat{\psi}(\delta) \pm (1.96) \hat{\sigma}(\delta) / \sqrt{n}.
\end{equation*}
Next, our estimate for the critical value of the supremum of the multiplier bootstrap process is 
\begin{equation*}
    P\left(\sup_{\delta \in \mathcal{D}} \left|\sqrt{n}\mathbb{P}_n \left[\xi \left\{\frac{\phi^*(Z;\hat{\eta}_{-S}, \delta) - \hat{\psi}(\delta)}{\hat{\sigma}(\delta)} \right\} \right] \right| \right) \ge \hat{\mathcal{C}}_\alpha | Z_1, \ldots, Z_n = \alpha,
\end{equation*}
with $(\xi_1, \ldots, \xi_n)$ being iid Rademacher random variables. To obtain the Rademacher draws in practice we use $X_{\text{Rademacher}} \sim 2 * \text{Bernoulli}(0.5) - 1.$ Finally, the uniform confidence bands \citet{Kennedy2018ipsi} are 
\begin{equation*}
    \mathbb{P} \left\{\hat{\psi}(\delta) - \frac{\hat{\mathcal{C}}_\alpha \hat{\sigma}(\delta)}{\sqrt{n}} \le \psi(\delta) \le \hat{\psi}(\delta) + \frac{\hat{\mathcal{C}}_\alpha \hat{\sigma}(\delta)}{\sqrt{n}}, \quad \forall \delta \in \mathcal{D} \right\} = 1 - \alpha + o(1).
\end{equation*}

The multiplier bootstrap approach utilizes the uncentered influence function of a given estimator. For comparison in the frequentist case, we will derive the influence function for the plug-in case. That is, considering $\hat{\pi}$ and $\hat{m}$ fixed we define, 
\begin{equation*}
    h_{\hat{\eta}, \delta}(X) = \frac{\delta \hat{\pi}(X) \hat{\mu}(X,1) + \{ 1 - \hat{\pi}(X)\}\hat{\mu}(X,0)}{\delta \hat{\pi}(X) + 1 - \hat{\pi}(X)}.
\end{equation*}
Then the plug-in influence function is 
\begin{equation*}
    \varphi_{\text{IPSI}} = h_{\hat{\eta}, \delta}(X) - \mathbb{P}_n[h_{\hat{\eta}, \delta}].
\end{equation*}
Fixing the nuisance functions reduces the target parameter to a mean functional of the observed data. Since integration with respect to a probability measure is linear, the resulting functional is linear in the underlying distribution. Also, it necessarily has mean 0 by construction.

\noindent \textbf{Additional Bayesian Estimation Details}

In the Bayesian setting we investigate the performance of several nonparametric Bayesian nuisance estimators, including BART \citep{chipman2010bart}, SoftBART \citep{Linero2018brte}, and SoftBCF \citep{hahn2020bcf}. Throughout the simulation study we use $B = 10{,}000$ draws from the one-step posterior distribution and across all instances we use 2,000 samples as burn in and then keep the next 2,000 samples as the nuisance function posteriors. For BART we utilize the package aptly named BART \citet{sparapani2021bart}. For both the propensity score and outcome regression we utilize the gbart() function with type being pbart and wbart respectively. For SoftBART we utilize the SoftBart package \citet{linero2022softbart}. Finally, SoftBCF utilizes a special case of the varying coefficient BART model described in \citet{linero2022softbart} to fit the outcome regression. We utilize the default settings for this procedure \citet{hahn2020bcf} where for alpha we have $\beta = 2, \gamma = 0.95$ across 200 trees and for beta we have $\beta = 3, \gamma = 0.25$ across 50 trees. Finally, across all models, default settings were utilized unless stated otherwise. 

\noindent \textbf{Uniform Coverage Details and Coverage Results}

 To obtain the Bayesian uniform credible bands we first define the empirical posterior mean function 
\begin{equation*}
    \Bar{\psi}_{P_0}(Q) = \frac{1}{B}\sum_{b = 1}^B \tilde{\psi}^{(b)}(\delta), \quad \delta \in \mathcal{D}.
\end{equation*}
Next, we compute the posterior standard deviation at each $\delta$. For each posterior draw $b$, we define the studentized supremum deviation,   
\begin{equation*}
    \hat{\sigma}(\delta) = \sqrt{\frac{1}{B-1} \sum_{b=1}^B (\tilde{\psi}^{(b)} - \Bar{\psi}(\delta))^2}, 
    \quad 
    T^{(b)} = \sup_{\delta \in \mathcal{D}} \frac{ \left|\tilde{\psi}^{(b)} - \Bar{\psi}(\delta) \right|}{\hat{\sigma}(\delta)}.
\end{equation*}
The critical value, $\hat{\mathcal{C}}_{1-\alpha}$, is the $(1-\alpha)$-quantile of $\{T\}_{b=1}^B$ which is used to construct the uniform credible band, 
\begin{equation*}
    \left[\Bar{\psi}(\delta) - \hat{\mathcal{C}}_{1-\alpha}\hat{\sigma}(\delta),\quad \Bar{\psi}(\delta) + \hat{\mathcal{C}}_{1-\alpha}\hat{\sigma}(\delta) \right], \quad \forall \delta \in \mathcal{D}.
\end{equation*}
Below are the uniform coverage results from the main simulation study. These results closely align with those from the pointwise intervals.

\begin{table}
\small
\caption{IPSI regular data comparison, for uniform coverage and length, of plug-in and EIF estimators across methods and sample sizes, with best in class performance being bold-faced.}
\centering
\begin{tabular}[t]{llcccc}
\toprule
\multicolumn{2}{c}{ } & \multicolumn{2}{c}{Plug-in} & \multicolumn{2}{c}{EIF} \\
\cmidrule(l{3pt}r{3pt}){3-4} \cmidrule(l{3pt}r{3pt}){5-6}
$n$ & Method & Un Cov & Un Len & Un Cov & Un Len\\
\midrule
 & Frequentist & 0.736 & \textbf{3.681} & 0.942 & \textbf{5.137}\\

 & BART & \textbf{0.954} & 4.849 & \textbf{0.978} & 5.458\\

 & SoftBART & 0.932 & 4.862 & 0.972 & 5.185\\

\multirow{-4}{*}{\raggedright\arraybackslash 500} & SoftBCF & 0.896 & 4.864 & 0.964 & 5.186\\
\cmidrule{1-6}
 & Frequentist & 0.732 & \textbf{2.620} & 0.942 & \textbf{3.626}\\

 & BART & \textbf{0.950} & 3.472 & \textbf{0.958} & 3.820\\

 & SoftBART & 0.902 & 3.479 & 0.948 & 3.646\\

\multirow{-4}{*}{\raggedright\arraybackslash 1000} & SoftBCF & 0.888 & 3.479 & 0.950 & 3.645\\
\cmidrule{1-6}
 & Frequentist & 0.720 & \textbf{1.176} & 0.938 & \textbf{1.627}\\

 & BART & \textbf{0.944} & 1.582 & \textbf{0.956} & 1.674\\

 & SoftBART & 0.912 & 1.598 & 0.942 & 1.629\\

\multirow{-4}{*}{\raggedright\arraybackslash 5000} & SoftBCF & 0.902 & 1.593 & 0.940 & \textbf{1.627}\\
\bottomrule
\end{tabular}
\end{table}

\begin{table}
\small
\caption{IPSI transformed data comparison, for uniform coverage and length, of plug-in and EIF estimators across methods and sample sizes, with best in class performance being bold-faced.}
\centering
\begin{tabular}[t]{llcccc}
\toprule
\multicolumn{2}{c}{ } & \multicolumn{2}{c}{Plug-in} & \multicolumn{2}{c}{EIF} \\
\cmidrule(l{3pt}r{3pt}){3-4} \cmidrule(l{3pt}r{3pt}){5-6}
$n$ & Method & Un Cov & Un Len & Un Cov & Un Len\\
\midrule
 & Frequentist & 0.386 & \textbf{3.235} & 0.926 & 5.221\\

 & BART & 0.872 & 4.766 & 0.968 & 5.365\\

 & SoftBART & \textbf{0.906} & 4.913 & \textbf{0.972} & 5.162\\

\multirow{-4}{*}{\raggedright\arraybackslash 500} & SoftBCF & 0.876 & 4.897 & 0.968 & \textbf{5.156}\\
\cmidrule{1-6}
 & Frequentist & 0.212 & \textbf{2.279} & 0.916 & 3.678\\

 & BART & 0.778 & 3.415 & \textbf{0.952} & 3.758\\

 & SoftBART & \textbf{0.866} & 3.526 & 0.950 & 3.640\\

\multirow{-4}{*}{\raggedright\arraybackslash 1000} & SoftBCF & 0.822 & 3.514 & 0.944 & \textbf{3.637}\\
\cmidrule{1-6}
 & Frequentist & 0.006 & \textbf{1.025} & 0.882 & 1.643\\

 & BART & 0.752 & 1.582 & \textbf{0.950} & 1.662\\

 & SoftBART & \textbf{0.866} & 1.656 & 0.946 & \textbf{1.640}\\

\multirow{-4}{*}{\raggedright\arraybackslash 5000} & SoftBCF & 0.860 & 1.637 & 0.942 & \textbf{1.640}\\
\bottomrule
\end{tabular}
\end{table}
\FloatBarrier
 
\noindent \textbf{Second Simulation Study}

In this second simulation study we have a $\pi$-fixed stochastic intervention, $Q \sim \text{Bernoulli}(p_\delta(X_1))$ where 
\begin{equation*}
    p_\delta(X_1) = \begin{cases}
        \delta_1 & X_1 \le 1 \\
        \delta_2 & X_1 > 1
    \end{cases}  
\end{equation*}
This is a $\pi$-fixed stochastic intervention in the sense that the treatment assignment probabilities are fully specified by the analyst and depend only on a pre-defined function of the covariates, not on the data-generating distribution. In this case, the estimand of interest is identified by 
\begin{equation*}
    \psi_{P_0}(Q) = P_0 \left[\mu(X,1)p_\delta(X) + \mu(X, 0)(1 - p_\delta(X)) \right]
\end{equation*}

While this intervention is simpler than an incremental propensity score intervention, it can still be relevant in practice. For instance, in a medical setting the probability of receiving a treatment, like a medication or chemotherapy procedure, could depend on the patient's covariates, like (age, genetics, etc.). These covariate-dependent rules are particularly useful for policy-makers who can't directly control treatment. 

In this case the density of $Q(a \mid X_1)$ is formally defined as, 
\begin{equation*}
    dQ(a\mid x_1) = p(x_1)^a [1-p(x_1)]^{1-a}, \quad a \in \{0,1\}.
\end{equation*}
Then the target estimand, $P_0[Y^Q]$, is the average outcome we would observe if treatment were assigned at random according to $Q(a \mid X_1)$. Then the efficient influence function in this case is 
\begin{align*}
\phi_{\text{$\pi$-fixed}} &= \sum_{a = 0}^1 \mu(X,a)dQ(a \mid X) - \psi(Q) + \frac{dQ(A \mid X)}{dP(A \mid X)}[Y - \mu(X,A)] \\
&= \mu(X,1) p(X) + \mu(X,0)(1-p(X)) - \psi(Q) + \frac{p(X)^A(1-p(X))^{1-A}}{dP(A \mid X)} \left[Y - \mu(X,A) \right].
\end{align*} 
Finally, by the same logic as before, the plug-in influence function is 
\begin{equation*}
    \varphi_{\text{$\pi$-fixed}} = \hat{\mu}(X,1)p_\delta(X) + \hat{\mu}(X, 0)(1 - p_\delta(X)) - \mathbb{P}_n[\hat{\mu}(X,1)p_\delta(X) + \hat{\mu}(X, 0)(1 - p_\delta(X))]
\end{equation*}

In the $\pi$-fixed case we use the same Monte Carlo setup, but we use 10 values of $\delta_1$ and $\delta_2$: $0.05, 0.15, \ldots, 0.95$. Below are the results in the $\pi$-fixed case, 
\begin{table}

\caption{Static regular data comparison of plug-in and EIF estimators across methods and sample sizes, with best in class performance being bold-faced.}
\centering
\begin{tabular}[t]{llcccccccc}
\toprule
\multicolumn{2}{c}{ } & \multicolumn{4}{c}{Plug-in} & \multicolumn{4}{c}{EIF} \\
\cmidrule(l{3pt}r{3pt}){3-6} \cmidrule(l{3pt}r{3pt}){7-10}
$n$ & Method & Bias & RMSE & Cov & Int Len & Bias & RMSE & Cov & Int Len\\
\midrule
 & Frequentist & 0.045 & \textbf{19.410} & \textbf{0.953} & 3.411 & 0.052 & \textbf{19.490} & \textbf{0.956} & 3.431\\

 & BART & 0.221 & 21.092 & 0.923 & 3.459 & 0.203 & 20.953 & 0.929 & 3.476\\

 & SoftBART & 0.102 & 19.804 & 0.949 & 3.446 & 0.095 & 19.774 & 0.950 & 3.445\\

\multirow{-4}{*}{\raggedright\arraybackslash 500} & SoftBCF & \textbf{0.043} & 19.813 & 0.948 & \textbf{3.395} & \textbf{0.042} & 19.819 & 0.948 & \textbf{3.398}\\
\cmidrule{1-10}
 & Frequentist & \textbf{0.021} & \textbf{19.176} & \textbf{0.957} & \textbf{2.416} & \textbf{0.022} & \textbf{19.216} & \textbf{0.958} & 2.427\\

 & BART & 0.144 & 20.424 & 0.938 & 2.438 & 0.133 & 20.310 & 0.942 & 2.448\\

 & SoftBART & 0.068 & 19.677 & 0.955 & 2.438 & 0.063 & 19.652 & 0.955 & 2.436\\

\multirow{-4}{*}{\raggedright\arraybackslash 1000} & SoftBCF & 0.036 & 19.528 & 0.953 & 2.419 & 0.034 & 19.532 & 0.952 & \textbf{2.421}\\
\cmidrule{1-10}
 & Frequentist & \textbf{0.016} & \textbf{19.675} & 0.949 & \textbf{1.083} & \textbf{0.015} & \textbf{19.716} & \textbf{0.950} & \textbf{1.087}\\

 & BART & 0.053 & 20.423 & 0.937 & 1.088 & 0.049 & 20.359 & 0.940 & 1.091\\

 & SoftBART & 0.031 & 19.918 & \textbf{0.951} & 1.101 & 0.029 & 19.889 & 0.949 & 1.092\\

\multirow{-4}{*}{\raggedright\arraybackslash 5000} & SoftBCF & 0.025 & 19.830 & 0.949 & 1.087 & 0.025 & 19.810 & 0.949 & 1.088\\
\bottomrule
\end{tabular}
\end{table}

\begin{table}

\caption{Static transformed data comparison of plug-in and EIF estimators across methods and sample sizes, with best in class performance being bold-faced.}
\centering
\begin{tabular}[t]{llcccccccc}
\toprule
\multicolumn{2}{c}{ } & \multicolumn{4}{c}{Plug-in} & \multicolumn{4}{c}{EIF} \\
\cmidrule(l{3pt}r{3pt}){3-6} \cmidrule(l{3pt}r{3pt}){7-10}
$n$ & Method & Bias & RMSE & Cov & Int Len & Bias & RMSE & Cov & Int Len\\
\midrule
 & Frequentist & 1.140 & 34.869 & 0.617 & \textbf{3.190} & 1.137 & 620.920 & 0.888 & 14.073\\

 & BART & 0.611 & 25.630 & 0.864 & 3.577 & 0.561 & 24.982 & 0.885 & 3.647\\

 & SoftBART & 0.366 & 22.235 & 0.928 & 3.640 & 0.367 & 22.251 & 0.927 & 3.649\\

\multirow{-4}{*}{\raggedright\arraybackslash 500} & SoftBCF & \textbf{0.236} & \textbf{21.550} & \textbf{0.932} & 3.561 & \textbf{0.239} & \textbf{21.553} & \textbf{0.935} & \textbf{3.582}\\
\cmidrule{1-10}
 & Frequentist & 1.260 & 45.760 & 0.378 & \textbf{2.195} & 0.702 & 78.603 & 0.844 & 4.247\\

 & BART & 0.360 & 23.813 & 0.906 & 2.496 & 0.331 & 23.346 & 0.917 & 2.526\\

 & SoftBART & 0.268 & 22.107 & 0.941 & 2.560 & 0.271 & 22.143 & 0.935 & 2.545\\

\multirow{-4}{*}{\raggedright\arraybackslash 1000} & SoftBCF & \textbf{0.181} & \textbf{20.902} & \textbf{0.952} & 2.513 & \textbf{0.183} & \textbf{20.931} & \textbf{0.951} & \textbf{2.515}\\
\cmidrule{1-10}
 & Frequentist & 1.314 & 95.480 & 0.051 & \textbf{0.949} & 0.265 & 37.819 & 0.798 & 1.523\\

 & BART & \textbf{0.149} & \textbf{22.899} & 0.922 & 1.110 & \textbf{0.144} & \textbf{22.752} & \textbf{0.922} & \textbf{1.111}\\

 & SoftBART & 0.175 & 23.666 & \textbf{0.933} & 1.205 & 0.174 & 23.630 & 0.915 & 1.146\\

\multirow{-4}{*}{\raggedright\arraybackslash 5000} & SoftBCF & 0.176 & 23.741 & 0.918 & 1.159 & 0.175 & 23.712 & 0.912 & 1.145\\
\bottomrule
\end{tabular}
\end{table}

\begin{table}

\caption{Static regular data comparison, for uniform coverage and length, of plug-in and EIF estimators across methods and sample sizes, with best in class performance being bold-faced.}
\centering
\begin{tabular}[t]{llcccc}
\toprule
\multicolumn{2}{c}{ } & \multicolumn{2}{c}{Plug-in} & \multicolumn{2}{c}{EIF} \\
\cmidrule(l{3pt}r{3pt}){3-4} \cmidrule(l{3pt}r{3pt}){5-6}
$n$ & Method & Un Cov & Un Len & Un Cov & Un Len\\
\midrule
 & Frequentist & 0.894 & \textbf{4.086} & 0.944 & \textbf{4.275}\\

 & BART & 0.852 & 4.419 & 0.876 & 4.460\\

 & SoftBART & \textbf{0.942} & 4.406 & \textbf{0.948} & 4.411\\

\multirow{-4}{*}{\raggedright\arraybackslash 500} & SoftBCF & 0.924 & 4.343 & 0.924 & 4.352\\
\cmidrule{1-6}
 & Frequentist & 0.912 & \textbf{2.890} & 0.958 & \textbf{3.024}\\

 & BART & 0.840 & 3.100 & 0.866 & 3.128\\

 & SoftBART & \textbf{0.958} & 3.103 & \textbf{0.962} & 3.104\\

\multirow{-4}{*}{\raggedright\arraybackslash 1000} & SoftBCF & 0.940 & 3.065 & 0.936 & 3.070\\
\cmidrule{1-6}
 & Frequentist & 0.896 & \textbf{1.294} & 0.950 & \textbf{1.355}\\

 & BART & 0.878 & 1.371 & 0.894 & 1.381\\

 & SoftBART & \textbf{0.964} & 1.400 & \textbf{0.960} & 1.392\\

\multirow{-4}{*}{\raggedright\arraybackslash 5000} & SoftBCF & 0.934 & 1.368 & 0.938 & 1.371\\
\bottomrule
\end{tabular}
\end{table}

\begin{table}

\caption{Static transformed data comparison, for uniform coverage and length, of plug-in and EIF estimators across methods and sample sizes, with best in class performance being bold-faced.}
\centering
\begin{tabular}[t]{llcccc}
\toprule
\multicolumn{2}{c}{ } & \multicolumn{2}{c}{Plug-in} & \multicolumn{2}{c}{EIF} \\
\cmidrule(l{3pt}r{3pt}){3-4} \cmidrule(l{3pt}r{3pt}){5-6}
$n$ & Method & Un Cov & Un Len & Un Cov & Un Len\\
\midrule
 & Frequentist & 0.354 & \textbf{3.984} & 0.810 & 11.719\\

 & BART & 0.822 & 4.597 & 0.856 & 4.743\\

 & SoftBART & 0.934 & 4.705 & 0.928 & 4.745\\

\multirow{-4}{*}{\raggedright\arraybackslash 500} & SoftBCF & \textbf{0.938} & 4.622 & \textbf{0.938} & \textbf{4.658}\\
\cmidrule{1-6}
 & Frequentist & 0.100 & \textbf{2.738} & 0.696 & 4.924\\

 & BART & 0.878 & 3.209 & 0.888 & 3.277\\

 & SoftBART & \textbf{0.964} & 3.306 & 0.946 & 3.303\\

\multirow{-4}{*}{\raggedright\arraybackslash 1000} & SoftBCF & 0.962 & 3.256 & \textbf{0.962} & \textbf{3.261}\\
\cmidrule{1-6}
 & Frequentist & 0.000 & \textbf{1.166} & 0.428 & 1.893\\

 & BART & 0.914 & 1.428 & 0.890 & \textbf{1.426}\\

 & SoftBART & \textbf{0.952} & 1.567 & \textbf{0.926} & 1.492\\

\multirow{-4}{*}{\raggedright\arraybackslash 5000} & SoftBCF & 0.940 & 1.500 & 0.906 & 1.486\\
\bottomrule
\end{tabular}
\end{table}

\FloatBarrier

As discussed in the main text, in the $\pi$-fixed scenario if $\mu(P_0)$ is estimated well then the posterior correction may not offer much improvement. However the Bayesian paradigm still provides advantages, as the Bayesian methods outperform the frequentist in the transformed data case.

\subsection{Additional Applied Study Details}
Default settings were used for SoftBART as described in the additional simulation study details. The model was fit on $I = 100$ values of $\delta$ log-uniform between $\exp(-2.3)$ and $\exp(2.3)$. Below are trace plots of the nuisance functions for randomly selected individuals and $\sigma$. We further have plots to confirm the normality assumption on the error term $\varepsilon \sim \sigma I$ for the outcome regression. In the probit model the error is assumed to be fixed at 1. Finally we plot the overlap.
\begin{figure}[h]
\includegraphics[width=1\textwidth]{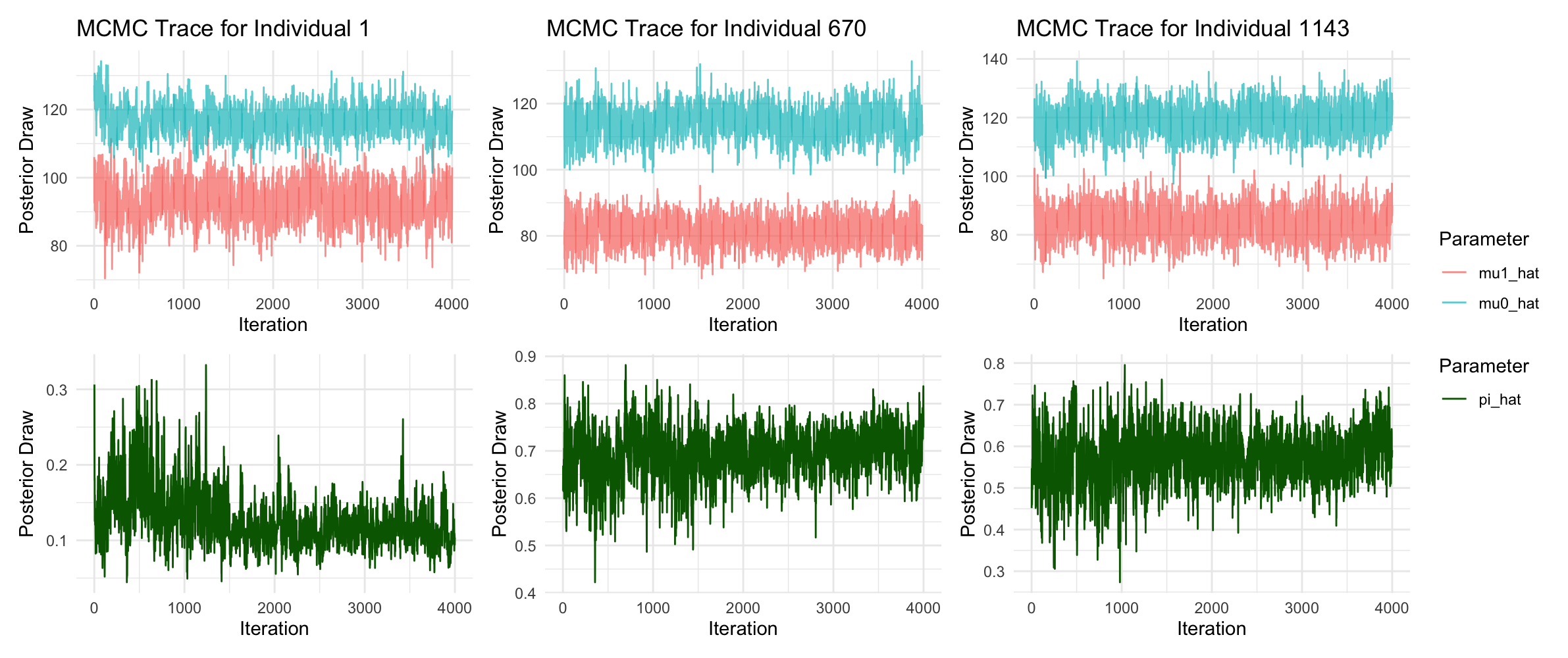}
\caption{Trace plots of nuisance function parameters for randomly selected individuals, shown to evaluate MCMC convergence and mixing.}
\label{trueProp}
\end{figure}

\begin{figure}[h]
\includegraphics[width=0.95\textwidth]{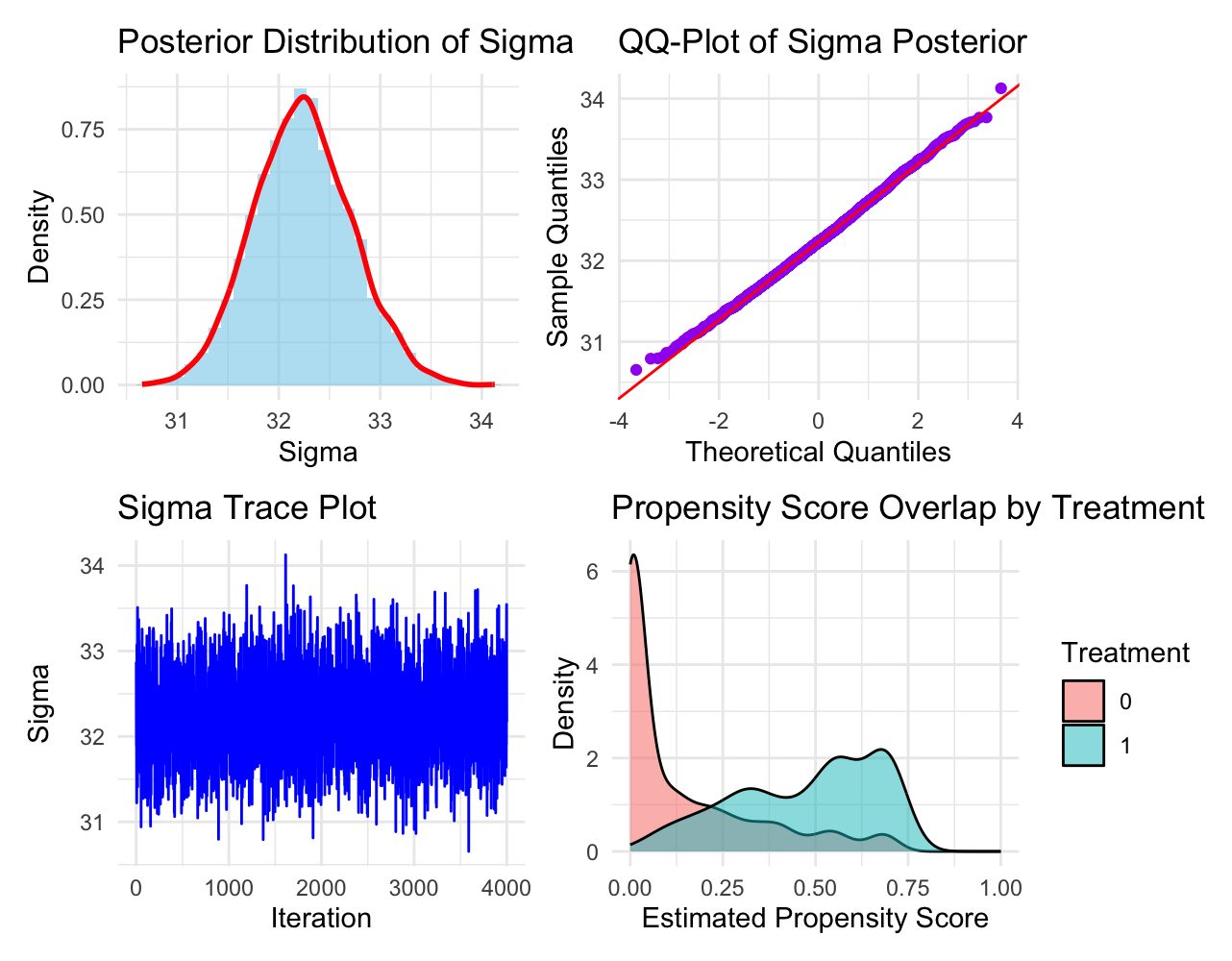}
\caption{Diagnostic plots for model assessment. The histogram and normal Q--Q plot evaluate the normality of $\sigma$. The trace plot of $\sigma$ assesses MCMC convergence. The propensity score overlap plot evaluates the positivity assumption.}
\label{big-fig}
\end{figure}
\FloatBarrier
Burn-in occurs at 2000 iterations, with the chains being mostly stable after that point. Further, as seen above through the histogram and QQ-Plot the normality assumption for $\sigma$ is justified. Finally, we have weak overlap and a lot of the estimated propensity score density is near zero. This validates the necessity of the incremental propensity score intervention for this analysis. Below is the plug-in shift effect per log-delta, as opposed to the one-step version in \cref{sec:Real_data},

\begin{figure}[h]
\includegraphics[width=0.95\textwidth]{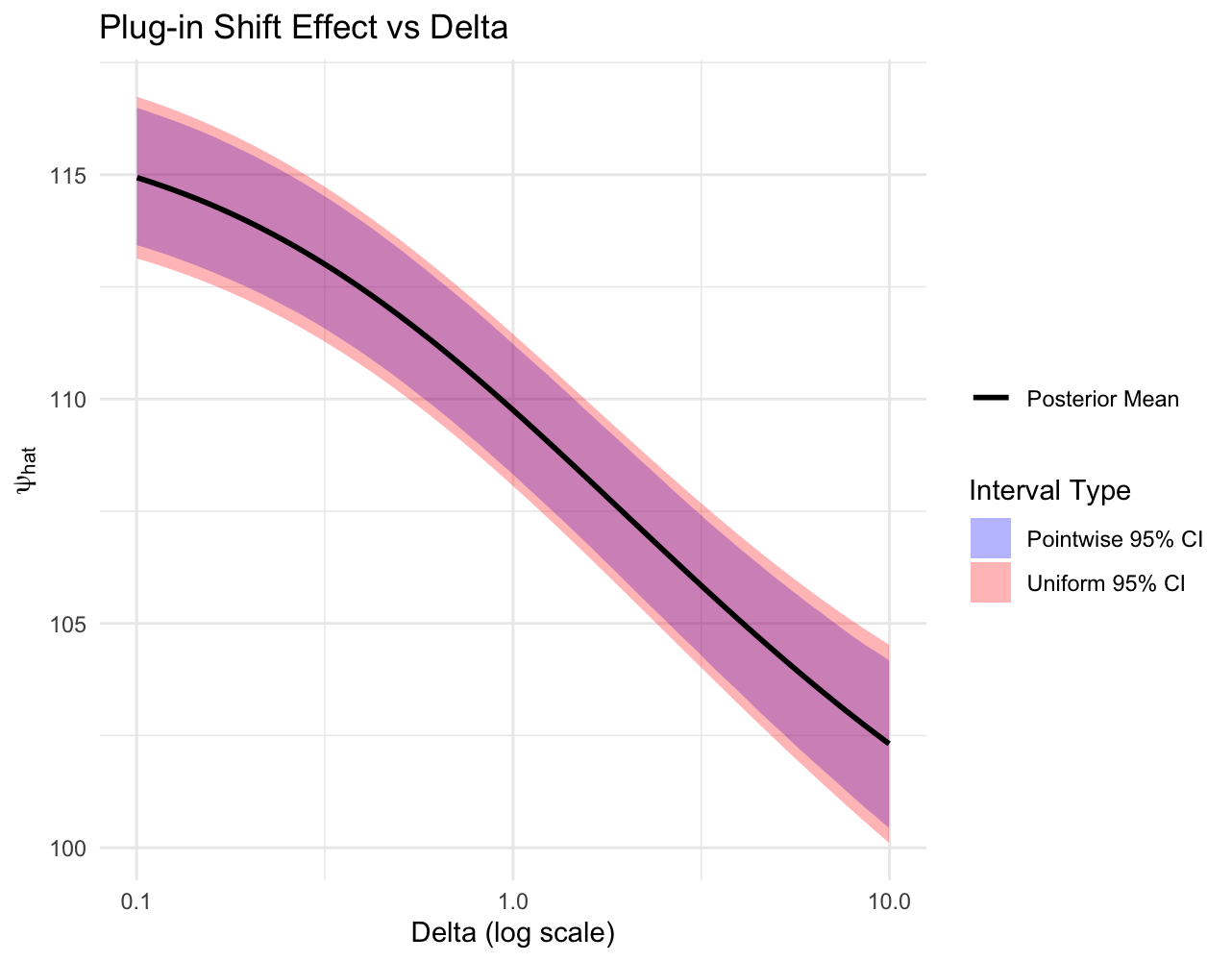}
\caption{Plug-in estimates of the effect of a hypothetical shift in statin treatment probability on LDL levels. The x-axis shows the magnitude of the shift in odds of receiving a statin ($\delta$, log scale), and the y-axis shows the corresponding estimated change in LDL ($\hat{\psi}$). The solid line represents the plug-in posterior mean, while the shaded areas indicate pointwise (blue) and uniform (red) 95\% credible intervals.}
\label{fig:plug_in_psi}
\end{figure}
\FloatBarrier

Finally, if we extend $\delta$ to (exp(-20.3), exp(20.3)) we recover the approximate average treatment effect. 

\begin{figure}[h]
\includegraphics[width=0.95\textwidth]{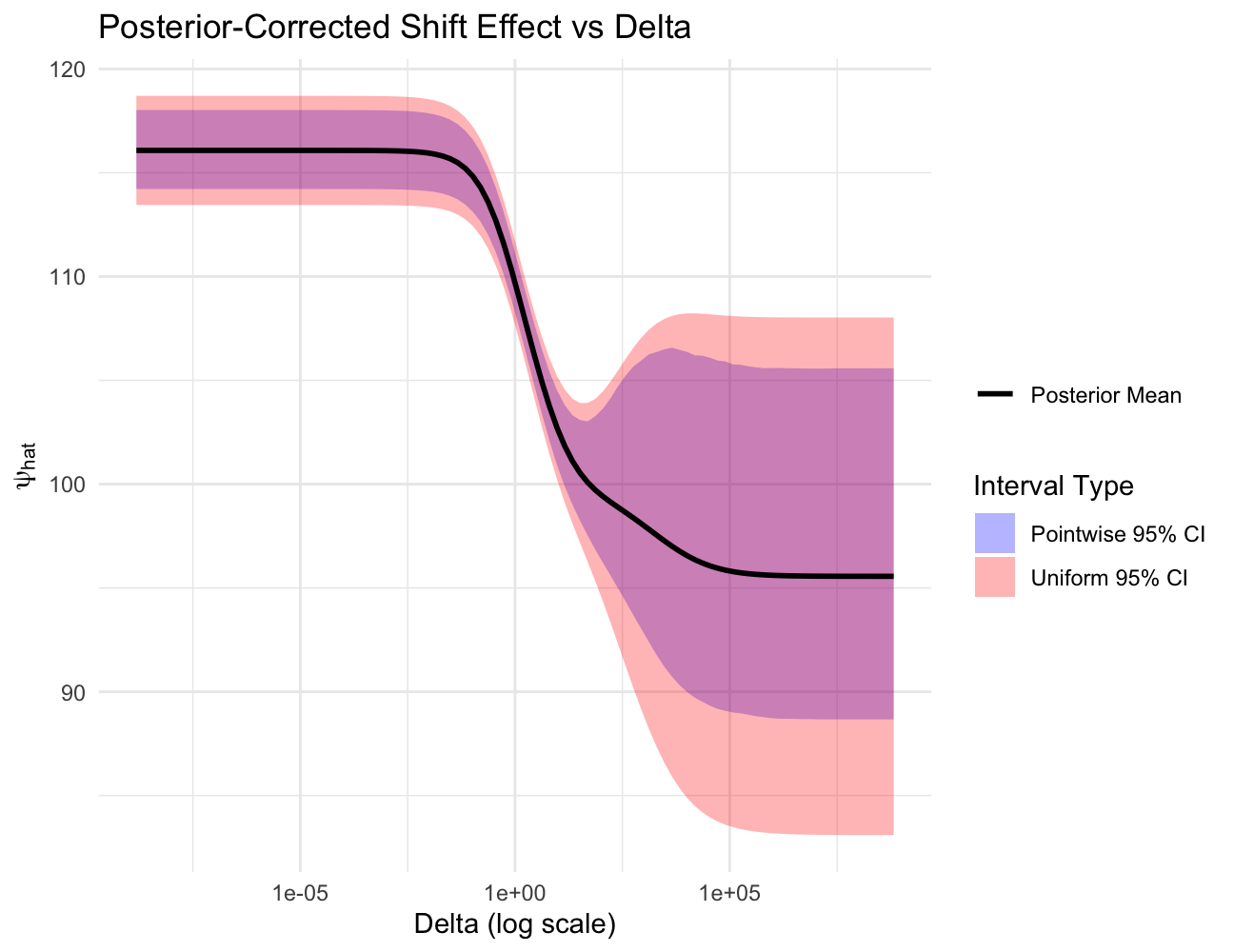}
\caption{Posterior-corrected estimates of the effect of a hypothetical shift in statin treatment probability on LDL levels. The x-axis shows the magnitude of the shift in odds of receiving a statin ($\delta$, log scale), and the y-axis shows the corresponding estimated change in LDL ($\hat{\psi}$). The solid line represents the posterior-corrected estimate, while the shaded areas indicate pointwise (blue) and uniform (red) 95\% credible intervals. Compared to \cref{fig:plug_in_psi} the $\delta$ scale is much larger.}
\label{big_psi}
\end{figure}

\end{document}